\documentclass{article}
\usepackage[a4paper, left=2cm,right=2cm,top=2.5cm,bottom=2cm]{geometry}
\usepackage[T1]{fontenc}
\usepackage[utf8]{inputenc}
\usepackage{lmodern}
\usepackage[english]{babel}
\usepackage{natbib}
\usepackage{graphicx}
\usepackage{amsmath,amsthm,amssymb}
\usepackage{braket}
\usepackage{tikz}
\usepackage{qcircuit}
\usepackage{algorithm}
\usepackage{algpseudocode}
\usepackage{bbm}
\usepackage{authblk}
\usepackage{bm}        
\usepackage{booktabs}
\usepackage{mathtools}
\usepackage{xcolor}
\usepackage{multirow}
\usepackage{hyperref}

\usepackage{array}
\usepackage{ragged2e}

\newcommand{\dd}{\mathrm{d}}

\newcommand{\Rbb}{\mathbb{R}}
\newcommand{\Tcal}{\mathcal{T}}
\newcommand{\Ocal}{\mathcal{O}}
\newcommand{\Dcal}{\mathcal{D}}

\newtheorem{remark}{Remark}

\newtheorem{theorem}{Theorem}
\newtheorem{lemma}{Lemma}
\newtheorem{proposition}{Proposition}
\newtheorem{corollary}{Corollary}
\newtheorem{result}{Result}

\newcommand{\code}[1]{\ifmmode\text{\ttfamily #1}\else\texttt{#1}\fi}
\DeclarePairedDelimiterX\ketbra[2]{\delimsize\vert}{\delimsize\vert}{#1\rangle\!\langle#2}

\newcommand{\hdashline}{\noalign{\hbox to\dimexpr0.91\textwidth+6\tabcolsep\relax{\leaders\hbox to 4pt{\hss\rule{2pt}{0.4pt}\hss}\hfill}}}

\title{Gate-Level Quantum Simulation of Nonunitary Linear Dynamics with Hybrid Oscillator--Qubit Architecture}

\author[1,2]{Elin Ranjan Das}
\author[2]{Muqing Zheng\thanks{muqing.zheng@pnnl.gov}}
\author[2]{Rishab Dutta}
\author[2,3]{Ang Li}
\author[2]{Timothy Stavenger}
\author[1,4,5]{Yuan Liu\thanks{q\_yuanliu@ncsu.edu}}

\affil[1]{Department of Electrical and Computer Engineering, North Carolina State University, Raleigh, North Carolina 27695, USA}
\affil[2]{Pacific Northwest National Laboratory, Richland, WA, USA, 99354}
\affil[3]{Department of Electrical and Computer Engineering, University of Washington, Seattle, WA, USA, 98195}
\affil[4]{Department of Computer Science, North Carolina State University, Raleigh, North Carolina 27695, USA}
\affil[5]{Department of Physics and Astronomy, North Carolina State University, Raleigh, North Carolina 27695, USA}

\begin{document}

\maketitle

\begin{abstract}
We translate the one-mode unitary-dilation framework for nonunitary linear dynamics into a gate-level hybrid oscillator--qubit architecture. An ancillary oscillator encodes the integral kernel through state preparation and postselection. A qubit register represents and simulates the discretized system operator. The construction applies to time-independent dynamics \(\dot u=-(L+iH)u\), including discretized partial differential equations, and removes the \(\mathcal{O}(\log M_a)\) ancilla-qubit overhead of a discrete-variable (DV) \(M_a\)-term quadrature register. We bound the squeezed-Fock coefficient-projection error of the ideal kernel state. It decays superalgebraically with cutoff \(N\) for Schwartz-class kernels and at a stretched-exponential rate under stronger joint decay and smoothness assumptions. The finite squeezed-Fock kernel state generically has stellar rank \(N-1\), making \(N\) a discrete measure of the oracle's non-Gaussian resource. For hybrid oscillator--qubit evolution, a \(p\)th-order product formula requires \(\mathcal{O}(t^{1+1/p}N_{\mathrm{Fock}}^{(p+1)/(2p)}\epsilon_t^{-1/p})\) Trotter steps in the worst case, up to generator-dependent commutator factors, to reach error \(\epsilon_t\), where \(N_{\mathrm{Fock}}\) is the oscillator dimension. A perturbation bound separates the total scaled-map error from the physical postselection probability. We benchmark Law--Eberly synthesis and assess a variational SNAP+\(\mathcal D\) route at the state-preparation level on discretized heat-equation instances. For full circuit-level maps of one-dimensional heat and non-normal advection--diffusion instances up to \(D=32\), the fixed-scale map error is at most \(1.48\%\) and the conditional infidelity at most \(4.68\times10^{-4}\) over all computational-basis inputs. 
At kernel parameters selected on the one-dimensional family, a \(4\times4\) two-dimensional stress case reaches \(7.40\%\) fixed-scale error under a reference norm shrunk by the stronger two-dimensional damping, with worst-input conditional infidelity \(6.20\times10^{-3}\).
At the prescribed DV sizing, the hybrid CV--DV route has smaller fixed-scale error in all ten instances. The DV route accepts with fewer repetitions in every instance. These results provide a block-by-block finite-size resource account of when a single continuous qumode can replace a discretized ancilla register.
\end{abstract}

\section{Introduction}

Hamiltonian simulation approximates time-evolution operators and underlies many quantum algorithms for differential equations~\cite{Berry2015}. The linear combination of unitaries (LCU) method represents target evolution as a weighted sum of unitary operations with provable complexity bounds~\cite{wiebelcu, Berry_2015, Berry_2014}. Other quantum PDE algorithms formulate PDEs as linear systems~\cite{Montanaro_2016, Childs2021highprecision, liu2021carleman}. Schr\"odingerisation~\cite{Jin2024,jin2025schrodingerizationbasedquantumalgorithms} and linear combination of Hamiltonian simulation (LCHS)~\cite{LCHS1,LCHS2} construct evolution operators with controllable error scaling. Existing gate-level implementations and resource analyses of LCHS use discrete-variable (qubit-based) architectures, which require discretizing the continuous integration variable.

Continuous-variable (CV) quantum systems such as harmonic oscillators realized in photonic modes or superconducting cavities naturally encode infinite-dimensional Hilbert spaces~\cite{Braunstein_2005, Weedbrook_2012}. These systems are well suited for representing bosonic modes and field-theoretic models without the need for coarse discretization~\cite{Marshall_2015}. Building on this, CV-based quantum algorithms for solving partial differential equations exploit the native continuous degrees of freedom to bypass the qubit overhead typically associated with high-dimensional systems~\cite{Jin_2024}. By contrast, qubit-based platforms provide strong nonlinearity, robust control, and efficient implementation of conditional operations. Hybrid oscillator-qubit architectures combine the advantages of both paradigms, offering expanded computational expressiveness and improved resource efficiency~\cite{kemper2025hybridcontinuousdiscretevariablequantumcomputing, liuhybrid2026}.

Implementing LCHS on a hybrid CV--DV architecture requires more than substituting an oscillator for the qubit quadrature register. The continuous LCHS kernel must be encoded in a normalizable, finite-energy oscillator state. This encoding limits the postselected-oracle accuracy and dominates the measured circuit-level error in Section~\ref{sec:clean_experiments}. The dominant encoding error has two components. The coefficient cutoff truncates the squeezed-Fock expansion, and the finite oscillator embedding truncates the representation required for the unbounded operator \(\hat{x}\). In these experiments, synthesis and its non-Gaussian cost, product-formula error, and postselection normalization make smaller contributions.

The central task is to construct the finite-energy encoding, separate analytically controlled errors from finite-model effects measured numerically, and account for circuit and postselection costs. Section~\ref{sec:methodology} carries out this analysis.

We construct a hybrid oscillator--qubit formulation of LCHS for time-independent linear dynamics \(\dot u=-(L+iH)u\) with \(L=L^\dagger\succeq0\) and \(H=H^\dagger\), including spatially discretized PDE systems of the same form. An ancillary oscillator encodes the continuous LCHS variable, and oscillator-state preparation, joint CV--DV evolution, and postselection implement the recovery kernel. The construction starts from the established one-mode dilation of Schr\"odingerisation and qumodisation~\cite{https://doi.org/10.1111/sapm.70047} and the existing continuous LCHS representation with the generalized exact kernel. Our contribution is the finite-resource oscillator realization of this recovery filter and the block-by-block approximation analysis in Section~\ref{sec:methodology}.
We also present two circuit-level methods for preparing the truncated non-Gaussian initial oscillator state: Law--Eberly (LE) synthesis using Jaynes--Cummings (JC) pulses and qubit rotations, and a variational SNAP+$\Dcal$ approach based on displacement and selective number-dependent arbitrary phase (SNAP) gates.
We evaluate the framework on heat equations with three boundary conditions, larger one-dimensional grids, a two-dimensional heat stress case, and non-normal advection--diffusion generators.

The rest of this article is organized as follows. Section~\ref{sec:main-results} summarizes the main results, Section~\ref{sec:background} reviews the LCHS kernel and the quadrature-register cost of its DV implementation, and Section~\ref{sec:methodology} develops the finite-resource oscillator realization together with its truncation, synthesis, simulation, and postselection analyses. Section~\ref{sec:examples} constructs the $L$ and $H$ oracles for representative ODE and PDE examples, Section~\ref{sec:simulation} describes the circuit-level realization, Section~\ref{sec:clean_experiments} presents the numerical experiments and the DV comparison, and Section~\ref{sec:conc} concludes.

Throughout this paper, we use the $\hbar=2$ convention. Appendix~\ref{app:notation} summarizes the notation used repeatedly in the main text.


\section{Summary of Main Results}
\label{sec:main-results}

For the benchmark instances, the hybrid oscillator-qubit formulation of LCHS combines analytical error control with a circuit-level oscillator-state preparation route and uses a substantially more compact ancillary-oracle description than the DV LCHS baseline.
Our work targets regimes where the quadrature-register cost of standard LCHS is the dominant limitation.

\subsection{Finite-Resource Oscillator Realization of LCHS}
\label{ss:lchs}

We now state the oscillator realization on which the remainder of the paper is built. Its joint oscillator--register dilation is the established one-mode construction shared with Schr\"odingerisation \cite{Jin2024} and qumodisation \cite{https://doi.org/10.1111/sapm.70047} (Section~\ref{ss:dilation}), and the result below expresses an admissible exact-LCHS kernel as an oscillator preparation-and-postselection matrix element.

\begin{result}[\textbf{Informal version of Theorem~\ref{thm:CV--DV-lchs}}]
Let $A\in \mathbb{C}^{D\times D}$ admit the decomposition \(A=L+iH\) with \(L=L^\dagger\succeq0\) and \(H=H^\dagger\). 
Let \(g\in L^1(\mathbb R)\) be an admissible exact-LCHS kernel for the pair \((L,H)\), in the sense that it satisfies the operator identity established under the hypotheses of Refs.~\cite{LCHS1,LCHS2},
\begin{align*}
\int_{\mathbb R}g(x)e^{-it(xL+H)}\,\dd x  = e^{-t(L+iH)}, \qquad t\geq0.
\end{align*}
Suppose normalized oscillator wavefunctions \(\phi\) and \(\psi\) satisfy
\(g(x)=\alpha_g\phi^*(x)\psi(x)\)
with known \(\alpha_g\). Then the resulting effective system operator is
\begin{align}
\boxed{
\mathcal K_g(t)
= \alpha_g (\bra{\phi}_{\rm osc}\otimes\mathbb I_q)
e^{-it(\hat{x}\otimes L+\mathbb I_{\rm osc}\otimes H)}
(\ket{\psi}_{\rm osc}\otimes\mathbb I_q)
= e^{-At}}.
\end{align}
Every kernel satisfying the cited exact-LCHS hypotheses and admitting a normalized oscillator factorization has an oscillator preparation-and-postselection realization. The scalar Fourier condition in Theorem~\ref{thm:CV--DV-lchs} supplies only a consistency check for the noncommuting operator identity.
\end{result}

For the rest of this work, we use the generalized exact-LCHS kernel \(g_\beta\) in Eq.~\eqref{eq:kernel_def} as an analytically explicit benchmark, while keeping the exact normalization factor \(\alpha_g\) explicit. This choice is discussed in Section~\ref{ss:kernel} and Table~\ref{tab:kernel-comparison}, with a measured comparison against the Huang--An weight at a shared operating point in Table~\ref{tab:ha-comparison}.
The target function $g(x)$ derived from the LCHS kernel is encoded through the overlap of a prepared initial oscillator state $|\psi\rangle_{\rm osc}$ and a postselected squeezed vacuum state $\langle\phi|_{\rm osc}$. This mapping removes the explicit $\Ocal(\log M_a)$ ancilla-qubit overhead that DV quadrature registers require for $M_a$ discretized integral terms, replacing it with $\mathcal{O}(1)$ ancillary oscillators. The formulation outlines the extension to inhomogeneous systems by utilizing a second oscillator to encode the required time-window weight. The formulation is discussed in Appendix~\ref{ss:inh}. 

In practice, the oscillator state \(\ket{\psi}_{\rm osc}\) is implemented using a normalized finite squeezed-Fock expansion with coefficient cutoff \(N\), as described in Section~\ref{ss:state-prep}. The implemented finite construction is measured against the ideal kernel state through the finite-model error \(\epsilon_{\rm model}\) defined in Section~\ref{ss:postselection}.

\begin{result}[\textbf{Informal version of Theorem~\ref{thm:state-prep}}]
The orthogonal squeezed-Fock projection of the ideal kernel state converges to it in the \(L^2\) norm. For Schwartz-class targets, the convergence is superalgebraic in the truncation dimension \(N\), while stronger joint decay and smoothness assumptions yield stretched-exponential convergence.
\end{result}

The hybrid oscillator-qubit evolution introduces an additional approximation error associated with product-formula simulation. The resulting Trotterization cost depends on both the truncated oscillator operators and the commutator structure of the encoded generators.

\begin{result}[\textbf{Informal version of Theorem~\ref{thm:trotter}}]
For the hybrid oscillator-qubit evolution, a \(p\)th-order product formula requires
\(
n_t=\Ocal\!\left[
t^{1+1/p}
N_{\rm Fock}^{(p+1)/(2p)}\epsilon_t^{-1/p}
\right]
\)
Trotter steps to achieve simulation error \(\epsilon_t\), in the worst case and up to generator-dependent commutator factors.
Thus, \(N\) controls the squeezed-Fock coefficient approximation, whereas \(N_{\rm Fock}\) controls the finite oscillator representation and the cutoff dependence of the product-formula bound.
\end{result}

Because the oracle is implemented through postselection, approximation errors in the hybrid evolution also affect the success probability of the protocol.

\begin{result}[\textbf{Informal version of Theorem~\ref{thm:postselection-success}}]
Let \(\ket{u_0}\) be normalized. At the displayed oscillator dimension \(N_{\rm Fock}\), let \(\ket{\phi_r}\) and \(\ket{\psi_N}\) denote the normalized states produced by the finite squeezing-matrix construction in Eq.~\eqref{eq:finite-target-states}, with their \(N_{\rm Fock}\) dependence suppressed in the notation. Assume \(\braket{\phi_r|\psi_N}\neq0\), and define
\begin{align*}
\alpha_{N,r}:=\frac{1}{\braket{\phi_r|\psi_N}}.
\end{align*}
Let \(\widetilde{\mathcal K}_{N,N_{\rm Fock},n_t}(t)\) denote the
unscaled finite postselected operator obtained with coefficient cutoff
\(N\), oscillator dimension \(N_{\rm Fock}\), and \(n_t\)
product-formula steps. Define
\begin{align*}
p_{\rm succ}^{(N,N_{\rm Fock},n_t)}(t;u_0)
&:=
\left\|
\widetilde{\mathcal K}_{N,N_{\rm Fock},n_t}(t)\ket{u_0}
\right\|_2^2,\\
\epsilon_{\rm tot} &:= \left\| \alpha_{N,r} \widetilde{\mathcal K}_{N,N_{\rm Fock},n_t}(t) -e^{-At} \right\|_2,\\
p_{\rm ref}^{(N,r)}(t;u_0) &:= \frac{\|e^{-At}\ket{u_0}\|_2^2} {|\alpha_{N,r}|^2}.
\end{align*}
Then
\begin{align*}
\left| p_{\rm succ}^{(N,N_{\rm Fock},n_t)}(t;u_0) - p_{\rm ref}^{(N,r)}(t;u_0) \right|
\leq
\frac{ 2\|e^{-At}\ket{u_0}\|_2\epsilon_{\rm tot} +\epsilon_{\rm tot}^2 }{ |\alpha_{N,r}|^2 }.
\end{align*}
Thus, the total scaled-map error $\epsilon_{\rm tot}$  controls the finite physical probability
relative to the comparison probability defined with the same finite
overlap scale.
\end{result}

\subsection{Performance and Resource Benchmarks}

The analytical bounds above are each confronted with a corresponding measured quantity (Sections~\ref{subsec:joint_tradeoff} and~\ref{sec:clean_experiments}). The measured coefficient-projection error is consistent with the stretched-exponential upper-bound form of Corollary~\ref{cor:near-optimal-kernel}, and the measured finite-model errors in the three \(M=4\) one-dimensional boundary-condition benchmarks lie more than an order of magnitude below the map-level bound of Lemma~\ref{lem:finite-r-map} for the implemented finite kernel, which evaluates to \(\|g_{N,r}-g\|_{L^1}=8.3\times10^{-2}\) at the selected parameters. The measured product-formula error exhibits first-order scaling over the evaluated Trotter-count range of Theorem~\ref{thm:trotter} and is at numerical precision for the periodic generator, whose commuting Pauli terms make the certified commutator sums vanish. For the Dirichlet benchmark, the first-order bound evaluated with its exact commutator structure certifies an error about $1.2\times10^{3}$ times the measured one at the operating point (Section~\ref{subsec:joint_tradeoff}). The measured postselection-probability deviations remain within the bound of Theorem~\ref{thm:postselection-success}, reaching \(62\%\) of it in the Dirichlet benchmark.

Full circuit-level breadth experiments on one-dimensional heat and non-normal advection--diffusion instances up to $D=32$ keep the fixed-scale map error $\varepsilon_F$ at or below $1.48\%$, while the $4\times4$ two-dimensional stress case at the same kernel parameters reaches $7.40\%$. Section~\ref{sec:clean_experiments} reports the instance-by-instance results.

A DV LCHS baseline sized by its own prescription at accuracy target $\epsilon=0.1$ puts $192$ to $632$ quadrature terms on $8$ to $10$ ancillary qubits against $32$ squeezed-Fock coefficients on one oscillator. At that sizing the CV--DV route attains the smaller fixed-scale error in every instance, while the DV block accepts after fewer attempts (Section~\ref{ss:dv-compare}).

For the 100-step Trotterized Dirichlet heat benchmark, the hybrid formulation reduces the ancillary-state oracle description size from the DV quadrature size $M_{\rm DV}=320$ to $N=32$ squeezed-Fock coefficients. In the Dirichlet benchmark, the Trotter-block two-qubit cost changes from $5600$ CX gates in the DV circuit, with each CRZ decomposed into two CX gates, to $400$ CX gates, $300$ controlled displacements ($c\Dcal$s), and $100$ displacements ($\Dcal$s) in the hybrid circuit, with no CX-equivalent cost assigned to the oscillator operations. Law--Eberly synthesis of the kernel state shared by all benchmark instances uses $31$ JC pulses together with one outer Gaussian squeezing, and the JC-pulse count coincides with the stellar rank $N-1=31$.


\section{Background and Motivation}
\label{sec:background}

Consider the linear ordinary differential equation
\begin{equation}
    \frac{\dd u(t)}{\dd t}
    = -A(t)u(t)+b(t), \qquad u(0)=u_0 . \label{eq:ode}
\end{equation}
LCHS provides an efficient quantum framework for solving such linear ODEs~\cite{LCHS1,LCHS2}. The solution can be written as
\begin{align}
    u(T) = \Tcal e^{-\int_0^T A(s)\dd s}u_0 + \int_0^T \Tcal e^{-\int_s^T A(s')\dd s'}b(s)\dd s , \label{eq:exp-sol}
\end{align}
where \(\Tcal\) denotes time ordering. Writing
\begin{align*}
A(t)=L(t)+iH(t),
\qquad
L(t)=L(t)^\dagger\succeq0,
\qquad
H(t)=H(t)^\dagger,
\end{align*}
LCHS expresses the dissipative evolution as a continuous linear combination of unitary evolutions:
\begin{align}
    \Tcal e^{-\int_0^t A(s)\dd s}
    &= \int_{\Rbb} g(k) \Tcal e^{-i\int_0^t(kL(s)+H(s))\dd s} \dd k , \label{eq:LCHS-sol-homo} \\
    \int^t_0 \Tcal e^{-\int^t_s A(s')\dd s'}b(s) \dd s 
    &= \int^t_0 \int_\Rbb g(k) \left[ \Tcal e^{-i \int^t_s (kL(s')+H(s'))\dd s'} \right]  b(s) \dd k \dd s.   \label{eq:LCHS-sol-inhomo}
\end{align}
The inhomogeneous contribution is treated analogously by inserting the same representation inside the Duhamel integral.
The weight function is defined as
\(
    g(k)=\frac{f(k)}{1-ik},
\)
and we use the generalized exact-LCHS kernel~\cite{LCHS2}
\begin{align}
    g_{\beta}(k) = \frac{e^{2^\beta}} {2\pi(1-ik)e^{(1+ik)^\beta}}, \qquad \beta\in(0,1), \label{eq:kernel_def}
\end{align}
where \(\beta\) controls the analyticity and decay of the kernel. Here and below, \((1+ik)^\beta\) is defined using the principal branch.

In standard LCHS, the continuous integral is truncated and discretized, producing an LCU circuit with a quadrature register of size \(\lceil\log_2(M_a)\rceil\). For the complementary solution,
\begin{align*}
M_a \in  \Ocal\!\left(  T\max_t\|L(t)\|   \left(\log\frac1\epsilon\right)^{1+1/\beta}  \right),
\end{align*}
so the cost includes both this ancilla register and the corresponding multi-controlled Hamiltonian-simulation blocks~\cite{LCHS2}. Although recent constant-factor analyses improve runtime estimates~\cite{pocrnic2025constant}, the discretized quadrature register remains a significant circuit-level overhead.

Motivated by continuous LCU~\cite{Chakraborty2024implementingany,bell2025co}, we study a hybrid CV--DV version of LCHS. In this approach, \(\Ocal(1)\) ancillary oscillators, together with CV state preparation and CV measurement, replace the discrete quadrature register, while the Hamiltonian simulation remains on the qubit register. This removes the explicit discrete quadrature register and its ancilla-qubit overhead and shifts the remaining approximation errors to CV state preparation and hybrid Hamiltonian simulation.


\section{Methodology}
\label{sec:methodology}
We now present the formal continuous--discrete variable dilation representation theorem.
Consider the homogeneous time-independent specialization of Eq.~\eqref{eq:ode}, $A(t)=A, \,\ b(t) = 0,$ whose solution is given by \(\ket{u(t)}=e^{-At}\ket{u_0}.\) 

\begin{theorem}[\textbf{Kernel-modular continuous--discrete variable LCHS}]
\label{thm:CV--DV-lchs}
Let \(A\in\mathbb C^{D\times D}\) admit the Cartesian decomposition
\begin{align}
A=L+iH,
\qquad
L\equiv\frac{A+A^\dagger}{2}\succeq0,
\qquad
H\equiv\frac{A-A^\dagger}{2i}.
\label{eq:cart}
\end{align}
Let \(g\in L^1(\mathbb R)\) be an admissible exact-LCHS kernel for the matrices \(L\) and \(H\) above, meaning that the hypotheses of the generalized exact-LCHS representation in Refs.~\cite{LCHS1,LCHS2} hold and hence
\begin{align}
\int_{\mathbb R}
g(x)e^{-it(xL+H)} \,\dd x = e^{-t(L+iH)},
\qquad t\geq0.
\label{eq:assumed-exact-lchs-identity}
\end{align}
In the scalar specialization, the kernel obeys the consistency relation
\begin{align*}
\int_{\mathbb R}
g(x)e^{-ixy}\,\dd x = e^{-y}, \qquad y\geq0.
\end{align*}
Suppose normalized oscillator wavefunctions \(\phi\) and \(\psi\) satisfy
\begin{align*}
g(x) = \alpha_g\phi^*(x)\psi(x),
\end{align*}
where \(\alpha_g\) is a known normalization factor. 
Define the unscaled physical postselection block
\begin{align*}
\mathcal B_g(t) &:= (\bra{\phi}_{\rm osc}\otimes\mathbb I_q)
e^{-it(\hat{x}\otimes L+\mathbb I_{\rm osc}\otimes H)}
(\ket{\psi}_{\rm osc}\otimes\mathbb I_q).
\end{align*}
The corresponding scaled operator is
\begin{align}
\mathcal K_g(t) &:= \alpha_g\mathcal B_g(t) = e^{-At}, \qquad t\geq0.
\label{eq:operator}
\end{align}
Consequently, for every initial state \(\ket{u_0}\in\mathbb C^D\),
\begin{align*}
\ket{u(t)} = \mathcal K_g(t)\ket{u_0} = e^{-At}\ket{u_0}.
\end{align*}
\end{theorem}
\begin{proof}
    See Appendix~\ref{proof:CV--DV-lchs}.
\end{proof}
A formal treatment of a constant source term $b(t)=b$ through a second ancillary oscillator is outlined in Appendix~\ref{ss:inh}.

The ideal normalization factor \(\alpha_g\) remains explicit throughout the exact construction. The finite construction defined in Section~\ref{ss:state-prep} instead uses the overlap scale \(\alpha_{N,r}\). Recall that, under the $\hbar=2$ convention used throughout this paper, $[\hat x,\hat p]=2i$ and $\hat x=\hat a+\hat a^\dagger$. We take the position representation of the postselection state $\ket{\phi}_{\rm osc}$ from Theorem~\ref{thm:CV--DV-lchs} to be a squeezed vacuum with position-space wavefunction
\begin{align*}
\phi_r(x) &= \frac{1}{(2\pi\sigma^2)^{1/4}} \exp\!\left(-\frac{x^2}{4\sigma^2}\right), \quad \sigma=e^r .
\end{align*}
Formally imposing $\phi_r^*(x)\psi_r(x)=g(x)$ gives us the position representation of the initial oscillator state $\ket{\psi}_{\rm osc},$
\begin{align*}
\psi_r^{\rm form}(x) = (2\pi\sigma^2)^{1/4} g(x) \exp\!\left(\frac{x^2}{4\sigma^2}\right).
\end{align*}
For the generalized exact-LCHS kernel in Eq.~\eqref{eq:kernel_def}, this finite-$r$ function is generally not square-integrable. Thus, we use \(\psi_r^{\rm form}\) only as a formal coefficient-generation target. The prefactor \((2\pi\sigma^2)^{1/4}\) is common to all generated coefficients and cancels when the finite coefficient vector is normalized. After removing this common factor, the coefficient-generating shape \(g(x)\exp[x^2/(4\sigma^2)]\) converges pointwise to \(g(x)\) as \(r\to\infty\). Separately, we define the normalizable ideal kernel state by \(\psi_\infty(x):=\mathcal N g(x)\), where \(\mathcal N\) is its normalization constant.

The full algorithmic circuit is shown in Fig.~\ref{fig:main-circuit}. The oscillator state preparation oracle $U_\psi$ realizes $\ket{\psi}_{\rm osc}$, the middle block implements the joint oscillator-qubit evolution, and the inverse squeezing plus vacuum postselection realizes $\bra{\phi_r}_{\rm osc}$. The one extra ancillary qubit required by the LE protocol is not shown in the figure. For a successful oscillator postselection, the unnormalized system branch is $\ket{v_{N,N_{\rm Fock},n_t}(t)}=\widetilde{\mathcal K}_{N,N_{\rm Fock},n_t}(t)\ket{u_0}$ and the normalized conditional state is $\ket{u_{N,N_{\rm Fock},n_t}(t)} = \frac{ \ket{v_{N,N_{\rm Fock},n_t}(t)}}{\|\ket{v_{N,N_{\rm Fock},n_t}(t)}\|_2}$.

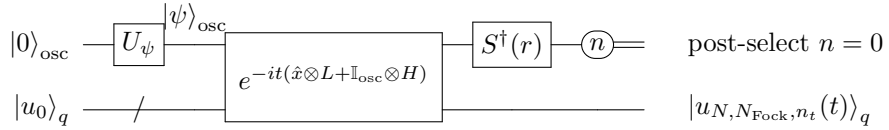
\begin{figure}[h]
    \centering
    \begin{align*}
    \Qcircuit @C=1.15em @R=1.15em {
        \lstick{\ket{0}_{\rm osc}}
            & \gate{U_{\psi}}
            & \ustick{\ket{\psi}_{\rm osc}} \qw
            & \multigate{1}{e^{-it(\hat{x}\otimes L+\mathbb I_{\rm osc}\otimes H)}}
            & \gate{S^\dagger(r)}
            & \measure{n}
            & \cw
            & \rstick{\text{post-select } n=0} \\
        \lstick{\ket{u_0}_q}
            & {/}\qw
            & \qw
            & \ghost{e^{-it(\hat{x}\otimes L+\mathbb I_{\rm osc}\otimes H)}}
            & \qw
            & \qw
            & \qw
            & \rstick{\ket{u_{N,N_{\rm Fock},n_t}(t)}_q}
    }
    \end{align*}
    \caption{Algorithmic circuit for CV--DV LCHS in the homogeneous time-independent setting. The state-preparation block $U_\psi$ loads the finite oscillator state, the middle block implements the hybrid evolution, and postselection after $S^\dagger(r)$ projects onto the squeezed-vacuum bra $\bra{\phi_r}_{\rm osc}=\bra{0}S^\dagger(r)$.}
    \label{fig:main-circuit}
\end{figure}

\subsection{Relation to Schr\"odingerisation and qumodisation}
\label{ss:dilation}

The term qumodisation and its systematic one-mode dilation formulation were introduced by Hu, Jin, Liu, and Zhang~\cite{https://doi.org/10.1111/sapm.70047}. Independently, Bell et al.~\cite{bell2025co} presented a construction that couples a qumode to a many-body quantum system for spectral filtering.
The enlarged Schr\"odingerisation state~\cite{Jin2024,jin2025schrodingerizationbasedquantumalgorithms} satisfies
\begin{align*}
\frac{{\rm d}}{{\rm d}t}w_h
&= -i(P_\mu\otimes H_1)w_h
+i(\mathbb I\otimes H_2)w_h \nonumber\\
&= -i\left( P_\mu\otimes H_1 - \mathbb I\otimes H_2 \right)w_h .
\end{align*}
For $u'(t)=-(L+iH)u(t)$, the substitutions $H_1=L, \qquad H_2=-H$ give
\begin{align*}
\mathcal H_{\rm Schr} = P_\mu\otimes L + \mathbb I\otimes H.
\end{align*}

Let \(p_\mu\) denote the spectral value of \(P_\mu\). The operator identification \(P_\mu=\hat{x}\) gives \(x=p_\mu\). In the Fourier convention in which \(P_\mu=-D_\mu\) and \(D_\mu\) has spectral value \(\mu\), this is equivalently \(x=-\mu\). Thus,
\begin{align*}
\mathcal H_{\rm Schr} = \hat{x}\otimes L + \mathbb I_{\rm osc}\otimes H = \mathcal H_{\rm CV-DV}.
\end{align*}
The equality of these one-mode generators is inherited from the established Schr\"odingerisation~\cite{Jin2024,Jin_2024,jin2025schrodingerizationmethodlinearnonunitary} and qumodisation dilation framework~\cite{https://doi.org/10.1111/sapm.70047} and is not a new solver identity.

The distinction studied here is the recovery map. The LCHS filter is encoded through normalized oscillator preparation and postselection, \(g(x)=\alpha_g\phi^*(x)\psi(x)\), followed by finite oscillator encoding, state synthesis, hybrid simulation, and postselection resource analysis. The same one-mode generator also admits an exact moment-matching interpretation, recorded in Appendix~\ref{app:moment-matching}.

Recent contour-based matrix decomposition (CBMD)~\cite{Wang_2026} places LCHS- and Schr\"odingerisation-type representations within a broader functional calculus for non-Hermitian matrix functions. We quantify the finite-size resource budget of Section~\ref{sec:methodology} for one continuous LCHS filter. The comparison with a discretized auxiliary register is architecture dependent, and general CBMD query complexity lies outside this analysis. For first-order exponential dynamics, CBMD gives a discrete sum of unitary Hamiltonian-simulation branches closely related to LCHS. The continuous LCHS integral admits the CV--DV encoding when its weight has the normalized factorization in Section~\ref{ss:kernel}. More general CBMD constructions use block encoding, QSVT, and an outer LCU to combine potentially nonunitary branch maps \(f(sH+L)\)~\cite{Wang_2026}. Extending the oscillator realization to those maps requires a new encoding and separate cutoff, state-synthesis, and postselection analyses.

\subsection{Kernel dependence and recent optimal constructions}
\label{ss:kernel}

The oscillator architecture is not restricted to the kernel \(g_\beta\) used in this work. Any admissible continuous LCHS weight \(g\) that admits a normalized factorization
\begin{align*}
g(k) = \alpha_g\phi_g^*(k)\psi_g(k)
\end{align*}
can, in principle, be represented through oscillator preparation and postselection. Matrix-query complexity does not determine the continuous-variable resources required to realize that factorization, among them the oscillator cutoff, photon number, non-Gaussianity, preparation depth, and sampling overhead. We make no optimality claim for \(g_\beta\) in either matrix queries or continuous-variable resources. Table~\ref{tab:kernel-comparison} separates the recovery guarantee proved by each construction from the oscillator cutoff and state-synthesis results that a hybrid implementation would additionally require.

\begin{table*}[t]
\centering 
\caption{Comparison of auxiliary-function constructions relevant to hybrid oscillator--qubit realizations. ``Not established'' denotes that the cited source does not provide the corresponding oscillator result.}
\label{tab:kernel-comparison}
\scriptsize
\setlength{\tabcolsep}{3pt}
\renewcommand{\arraystretch}{1.2}
\begin{tabular}{
>{\RaggedRight\arraybackslash}p{0.12\textwidth}
>{\RaggedRight\arraybackslash}p{0.14\textwidth}
>{\RaggedRight\arraybackslash}p{0.20\textwidth}
>{\RaggedRight\arraybackslash}p{0.12\textwidth}
>{\RaggedRight\arraybackslash}p{0.18\textwidth}
>{\RaggedRight\arraybackslash}p{0.14\textwidth}}
\toprule
Kernel family  & Recovery guarantee  & Tail / regularity  & Normalization  & Fock-cutoff bound  & Oscillator synthesis \\
\midrule
Generalized exact-LCHS kernel \(g_\beta\), \(0<\beta<1\)~\cite{LCHS2}
& Exact continuous LCHS identity for \(L(t)\succeq0\)
& Real-axis decay \(e^{-c|k|^\beta}\); exponential decay excluded within the exact analytic class
& Unit-integral weight; overlap factor \(\alpha_g\) from the factorization used here
& Coefficient cutoff \(N\) (Theorem~\ref{thm:state-prep}, this work); finite \(N_{\rm Fock}\) not established analytically
& Law--Eberly and SNAP+\(\Dcal\) (this work) \\ \midrule
Optimal Schr\"odingerisation initialization~\cite{jin2025schrodingerizationmethodlinearnonunitary}
& Exact warped-phase recovery with controlled finite-domain and discretization errors
& Smooth initializations (error-function, cut-off, polynomial-interpolation, and Fourier constructions)
& No oscillator-overlap factor \(\alpha_g\)
& Not established
& Smooth auxiliary initialization; no oscillator synthesis \\ \midrule
Optimal approximate LCHS~\cite{low2025optimalquantumsimulationlinear}
& Controlled approximation of \(e^{-At}\) with optimal matrix-query scaling
& Entire Gaussian-regularized construction; exponentially convergent quadrature
& LCU/block-encoding normalization; no \(\alpha_g\)
& Not established
& Qubit LCU PREP/SELECT; no oscillator synthesis
\\ \midrule
Fourier-transform LCHS~\cite{huang2025fouriertransformbasedlinearcombination}
& Exact LCHS identity under real-axis Fourier and regularity conditions
& Schwartz-class generating function with super-polynomial Fourier decay; exponential decay excluded
& \(F(0)=1\) gives a unit-integral weight
& Coefficient cutoff \(N\) (Theorem~\ref{thm:state-prep}(i) via the Schwartz-class weight, this work); measured \(\epsilon_{\rm tr}(32)\) and \(\tau_{\psi,64}\) (Table~\ref{tab:ha-comparison}, this work)
& Discrete LCU~\cite{huang2025fouriertransformbasedlinearcombination}; Law--Eberly under the common factorization (Table~\ref{tab:ha-comparison}, this work) \\
\bottomrule
\end{tabular}
\end{table*}

The Fourier-transform LCHS weight of Huang and An~\cite{huang2025fouriertransformbasedlinearcombination} is
\begin{align*}
g_{\rm HA}(x) = \frac{f_{\rm HA}(x)}{1-ix}, \qquad f_{\rm HA}(x) = \frac{1}{2\pi}\int_{-\delta}^{0}\Phi_{\rm HA}(w)\,e^{ixw}\,\dd w ,
\end{align*}
where the generating function \(\Phi_{\rm HA}\) of that construction is supported on \([-\delta,0]\) with the closed form 
\begin{align*}
\Phi_{\rm HA}(w)&=\frac{   e^{-1/\left(y^m(1-y)^m\right)}e^{-w}   }{\delta Z_m}, \\
y&=(w+\delta)/\delta,  \\
Z_m&=\int_0^1 e^{-1/\left(p^m(1-p)^m\right)}\dd p.
\end{align*}
In the notation of Huang and An~\cite{huang2025fouriertransformbasedlinearcombination}, \(f_{\rm HA}\) is the function \(f\) defined by the inverse Fourier transform in their Eq.~(11), while \(g_{\rm HA}=f_{\rm HA}/(1-ix)\) is the full LCHS weight in their Eq.~(13). The displayed expression for \(\Phi_{\rm HA}\) follows by substituting the piecewise function \(F\) and glue function \(\psi\) of their Eqs.~(35)--(36) into \(\Phi=F+F'\) from their Eq.~(10). 
Here \(Z_m=c_g^{-1}\), and our variables \(x\) and \(w\) correspond to their \(k\) and \(x\), respectively.
The endpoint values of \(\Phi_{\rm HA}\) are defined by continuous extension, and the exact normalization \(\int_{\mathbb{R}}g_{\rm HA}(x)\,\dd x=F(0)=1\) follows from the same Fourier convention. The weight is Schwartz class, so Theorem~\ref{thm:state-prep}(i) bounds the coefficient-projection tail of its ideal kernel state, and its finite-\(r\) coefficients are generated by the coefficient-generating pairing of Section~\ref{ss:state-prep} with \(g_{\rm HA}\) in place of \(g\).
Table~\ref{tab:ha-comparison} measures the kernel dependence of the oscillator encoding by comparing \(g_\beta\) with the Huang--An weight at \((m,\delta)=(1,2)\), the shape highlighted in that work for its total-query proxy at \(\epsilon\leq 10^{-8}\), both kernels evaluated at the shared operating point \((r,r',N,N_{\rm Fock})=(1.6,\,0.25,\,32,\,64)\) of the numerical study in Section~\ref{sec:clean_experiments}. The state-level entries use the definitions of Sections~\ref{ss:state-prep}--\ref{ss:non-gauss}, and the circuit-level entries, the fixed-scale error \(\varepsilon_F\) and the benchmark-input postselection probability \(p_{\rm succ}\), are computed with the simulation pipeline of Section~\ref{sec:clean_experiments} on its \(M=4\) heat benchmarks.

\begin{table}[ht!]
\centering
\caption{Measured kernel comparison at the shared operating point \((r,r',N,N_{\rm Fock})=(1.6,\,0.25,\,32,\,64)\), \(T=1\). The Huang--An weight uses \((m,\delta)=(1,2)\) and is evaluated under the common oscillator factorization, so every oscillator-specific entry is a result of this work. Stellar ranks are numerically resolved, with the top coefficient exceeding ten times the refinement tolerance of the coefficient quadrature. The fixed-scale error \(\varepsilon_F\) and the benchmark-input postselection probability \(p_{\rm succ}\) are circuit-level values for the \(M=4\) heat benchmarks with Dirichlet, Neumann, and periodic boundaries (D/N/P) at \(n_t=100\) with Law--Eberly (LE) preparation, and the LE counts exclude the two outer squeezing gates.}
\label{tab:ha-comparison}
{\small
\setlength{\tabcolsep}{5pt}
\begin{tabular}{lcc}
\toprule
 & \(g_{\beta}\), \(\beta=0.5\) & Huang--An, \((m,\delta)=(1,2)\) \\
\midrule
Coefficient-projection tail \(\epsilon_{\rm tr}(32)\) & \(2.27\times10^{-2}\) & \(5.11\times10^{-3}\) \\
Ordinary-Fock tail \(\tau_{\psi,64}\) & \(3.75\times10^{-4}\) & \(2.43\times10^{-6}\) \\
Stellar rank & \(31\) & \(31\) \\
Non-Gaussianity \(\delta_{\rm nG}\) & \(1.33\) & \(0.63\) \\
Mean photon number \(\langle\hat n\rangle\) & \(1.61\) & \(0.49\) \\
Overlap scale \(|\alpha_{N,r}|\) & \(1.298\) & \(2.476\) \\
\(\varepsilon_F\) (D/N/P) & \(0.699\%\,/\,0.565\%\,/\,0.469\%\) & \(0.280\%\,/\,0.247\%\,/\,0.094\%\) \\
Benchmark \(p_{\rm succ}\) (D/N/P) & \(10.47\%\,/\,16.48\%\,/\,15.36\%\) & \(2.90\%\,/\,4.54\%\,/\,4.23\%\) \\
LE synthesis & \(31\) JC + \(31\) R pulses & \(31\) JC + \(31\) R pulses \\
\bottomrule
\end{tabular}}
\end{table}

At the shared point, the two kernels exhibit a clear accuracy--sampling trade-off because the Huang--An kernel has a smaller coefficient-projection tail and lower circuit-level fixed-scale errors, whereas \(g_\beta\) has higher postselection probabilities and lower direct-repetition overhead. The Huang--An state has a coefficient-projection tail smaller by a factor of \(4.5\), approximately half the non-Gaussianity, and approximately one third of the mean photon number. Its circuit-level fixed-scale errors are lower by factors of \(2.3\)--\(5.0\) at equal Law--Eberly synthesis counts. Across the Dirichlet, Neumann, and periodic benchmarks, the measured \(g_\beta\) postselection probabilities are \(3.61\)--\(3.63\) times larger. Under direct repetition, the expected raw-execution count needed for a fixed number of accepted samples is larger by the same factors for the Huang--An route.

The factor \(1/|\alpha_{N,r}|^2\) in the finite-resource reference probability of Section~\ref{ss:postselection} accounts for nearly all of the measured postselection-probability ratio. The values \(|\alpha_{N,r}|=1.298\) for \(g_\beta\) and \(2.476\) for Huang--An predict a \(g_\beta\)-to-Huang--An reference-probability ratio of \((2.476/1.298)^2=3.64\), within \(0.9\%\) of each circuit-level ratio. At this shared point, the postselection state \(\ket{\phi_r}\) has even parity, so the odd-\(n\) terms in \(\ket{\psi_N}\) do not contribute to \(\braket{\phi_r|\psi_N}\). The Huang--An state places more probability mass in this orthogonal sector, and within the even sector its \(n=2\) overlap term partially cancels its \(n=0\) term. Together, these factorization-specific contributions give \(\left|\braket{\phi_r|\psi_N}\right|=0.404\) for Huang--An and \(0.770\) for \(g_\beta\). Lower non-Gaussianity and mean photon number do not imply a larger overlap with the fixed postselection state.

We use \(g_\beta\) for the main-text experiments because it has lower postselection overhead and its coefficient decay obeys the bound in Corollary~\ref{cor:near-optimal-kernel}. Appendix~\ref{app:ha-retune} shows that the trade-off persists when the Huang--An kernel uses its grid-selected \((r,r')\). The Weyl-calculus LCHS formulas~\cite{ni2026quantumeigenvaluetransformation} use discrete constructions outside the continuous factorization analyzed here. Their oscillator comparison requires a separate study.


\subsection{State Preparation and Fock Truncation}
\label{ss:state-prep}

This subsection specifies the oscillator state-preparation block $U_\psi$ in Fig.~\ref{fig:main-circuit}.
Throughout, the \textit{kernel state} is the preparation-side oscillator state whose pairing with the squeezed-vacuum postselection state realizes the LCHS kernel, \(g=\alpha_g\phi^*\psi\) ideally and \(g_{N,r}=\alpha_{N,r}\phi_r^*\psi_{N,r}\) in the finite construction of Lemma~\ref{lem:finite-r-map}. In the infinite-squeezing limit this state is the normalized kernel itself, and we define \(\psi_\infty(x)=\mathcal N g(x)\) as the ideal kernel state, associated with the common-prefactor-free limiting shape of the formal construction. Since infinite squeezing is not physically realizable, the implemented oracle instead prepares the finite kernel state, a normalized finite squeezed-Fock state
\begin{align*}
    \ket{\psi_{N,r,r',\beta}} = S(r')\sum_{n=0}^{N-1} C_{n}\ket{n} = \sum_{n=0}^{N-1} C_n \ket{\phi_{n,r'}},
\end{align*}
where \(S(r')\) is the single-mode squeezing operator, \(r\) is the postselection squeezing parameter, \(r'\) is the preparation-basis squeezing parameter with \(r'<r\), and \(N\) is the number of squeezed-Fock coefficients. We use the independent cutoff \(N_{\rm Fock}\) for the dimension of the ordinary number-basis representation used to implement oscillator states and operators.
The coefficients $\{C_{n}\}$ are generated by pairing the formal finite-$r$ target wavefunction $\psi_r^{\rm form}$ with the squeezed Fock basis functions $\{\phi_{n,r'}(x)\}$. In position representation, the squeezed Fock basis functions are given by 
\begin{align*}
    \phi_{n,r'}(x)=\frac{1}{\sqrt{\sqrt{2}\sigma'}}h_n\left(\frac{x}{\sqrt{2}\sigma'}\right),\qquad \sigma'=e^{r'},
\end{align*}
where 
\begin{align*}
    h_n(y) = \pi^{-1/4}(2^n n!)^{-1/2} H_n(y)e^{-y^2/2}
\end{align*}
is the $n^{\rm th}$ normalized Hermite function and $H_n$ denotes the physicists' Hermite polynomial. 
For finite \(r\), \(\psi_r^{\rm form}\) is generally not square-integrable. Accordingly, \(\widetilde C_n\) denotes the following coefficient-generating integral:
\begin{align*}
\widetilde C_n=\braket{\phi_{n,r'}|\psi^{\rm form}_r}=\int_{\mathbb{R}} \phi^*_{n,r'}(x) \psi_r^{\rm form}(x)\dd x
\end{align*}
The physical coefficients are then obtained by normalization,
\begin{align*}
C_n = \frac{\widetilde C_n} {\left(\sum_{m=0}^{N-1}|\widetilde C_m|^2\right)^{1/2}}.
\end{align*}
Using the explicit forms of $\phi_{n,r'}(x)$ and $\psi_r^{\rm form}(x)$, this expression reduces to
\begin{align*}
    \widetilde C_{n}
    &= \sqrt{\frac{\sigma}{\sigma'}}
    \frac{1}{\sqrt{2^{n}{n}!}}
    \int_{\mathbb R}
    H_{n}\!\left(\frac{x}{\sqrt{2}\sigma'}\right)
    g(x) e^{-\gamma x^2} \,\dd x
\end{align*}
where \(r'<r\), so \(\sigma'<\sigma\) and
\begin{align*}
\gamma &= \frac14\left(\frac{1}{\sigma'^2}-\frac{1}{\sigma^2}\right)
= \frac14\left(e^{-2r'}-e^{-2r}\right) >0, \qquad \sigma=e^r,\quad \sigma'=e^{r'} .
\end{align*}
For the specific kernel from~\eqref{eq:kernel_def}, the coefficients take the explicit form
\begin{align}
    \widetilde C_n &= \sqrt{\frac{\sigma}{\sigma'}}
    \frac{1}{\sqrt{2^{n}{n}!}}
    \int_{\mathbb R}
    \frac{ e^{2^\beta}
    H_{n}\!\left(\frac{x}{\sqrt{2}\sigma'}\right)
    e^{-\gamma x^2}  }{ 2\pi(1-ix)e^{(1+ix)^\beta} }
    \,\dd x.
    \label{eq:coefficients}
\end{align}

In general, this integral does not admit a closed-form expression for arbitrary $\beta$.
Two examples of closed analytic form for the limit values of $\beta=\{0,1\}$ are given in Appendix~\ref{ss:analytic}.
Theorem~\ref{thm:state-prep} controls the orthogonal coefficient projection \(\Pi^F_{N,r'}\psi_\infty\) of the ideal state. Section~\ref{sec:clean_experiments} includes the discrepancy from the finite-\(r\) pairing, the finite-\(N_{\rm Fock}\) embedding, and the truncated joint dynamics in the measured finite-model error \(\epsilon_{\rm model}\). The numerical results evaluate these finite-resource effects. The theorem supplies the asymptotic convergence statement for the ideal projection.

The finite-\(r\) construction also admits a direct map-level bound that does not reference the ideal projection.
\begin{lemma}[Map-level error of the implemented finite-\(r\) kernel]
\label{lem:finite-r-map}
Let \(g_{N,r}(x):=\alpha_{N,r}\phi_r^*(x)\psi_{N,r}(x)\) denote the implemented finite kernel and let \(\mathcal K_{N,r}(t):=\int_{\Rbb} g_{N,r}(x)e^{-it(xL+H)}\dd x\) be the corresponding ideal-oscillator postselected map. Then, for all \(t\geq0\) and all admissible \((L,H)\),
\begin{align*}
    \left\|\mathcal K_{N,r}(t)-e^{-At}\right\|_2 \leq \left\|g_{N,r}-g\right\|_{L^1(\Rbb)}.
\end{align*}
\end{lemma}
\begin{proof}
By the exact identity \(e^{-At}=\int_{\Rbb} g(x)e^{-it(xL+H)}\dd x\), the difference equals \(\int_{\Rbb}\left[g_{N,r}(x)-g(x)\right]e^{-it(xL+H)}\dd x\), and the bound follows from the triangle inequality because each \(e^{-it(xL+H)}\) is unitary.
\end{proof}
At the kernel parameters selected in Section~\ref{sec:clean_experiments}, numerical quadrature over the recorded coefficients gives \(\|g_{N,r}-g\|_{L^1}=8.3\times10^{-2}\), a \(7.6\%\) relative \(L^1\) deviation from \(\|g\|_{L^1}=1.10\). The lemma turns this deviation into a map-error estimate of \(8.3\times10^{-2}\), up to the quadrature accuracy of the \(L^1\) evaluation. The measured finite-model errors are reported there for comparison. The finite-\(N_{\rm Fock}\) embedding, synthesis, and product-formula contributions remain the separately reported components, and the inequality does not by itself establish an asymptotic convergence rate in \(N\).


\subsection{Error Bound Estimation for State Preparation}
\label{ss:error-bound-sp}

We now quantify the approximation error incurred by truncating the squeezed-Fock expansion of the ideal kernel state $\psi_\infty(x)=\mathcal N g(x)$.
In prior work~\cite{Jin_2024}, the sub-optimal Lorentzian kernel $f(k)=\frac{1}{\pi(1+ik)}$ was realized using a single squeezed vacuum state $\frac{1}{\sqrt{s}\pi^{1/4}} e^{-\frac{k^2}{2s^2}}$, with a reported maximum kernel-state fidelity of $98.6\%$ at squeezing factor $s=0.925$. 
For the generalized exact-LCHS kernel used here, we instead employ a superposition of squeezed Fock states.
As shown below, the generalized exact-LCHS kernel of~\cite{LCHS2} yields a squeezed-Fock coefficient-projection error that \textit{decays superalgebraically with the cutoff $N$}.

Part (ii) of Theorem~\ref{thm:state-prep} sharpens this to the rate \(\exp[-cN^{1/(2\mu)}]\) for targets satisfying its joint decay and smoothness condition, the symmetric Gelfand--Shilov class \(S_\mu^\mu\). For \(g_\beta\), the rescaled target lies in \(S_{1/\beta}^{1/\beta}\) (Lemma~\ref{lem:kernel-GS}), which yields the stretched-exponential rate in Corollary~\ref{cor:near-optimal-kernel}. Since \(0<\beta<1\), the present kernel family has \(\mu=1/\beta>1\) and exponent \(\beta/2<1/2\), and this exponent approaches the root-exponential value \(1/2\) only as \(\beta\to1^{-}\). More generally, the endpoint \(\mu=\tfrac12\) in part~(ii) yields exponential decay \(\exp[-cN]\), whereas part~(iii) yields root-exponential decay \(\exp[-c\sqrt{N-1}]\) under its Gaussian-weighted strip hypotheses.

The truncated state $\psi_{N,r'}$ corresponds to the orthogonal projection of $\psi_\infty$ onto the first $N$ squeezed Fock basis functions. 
With the $\hbar=2$ convention used and under the change of variables
\(x=\sqrt{2}\sigma'y\), this projection is equivalent to a truncated Hermite
expansion of the rescaled function
\begin{align*}
F_{r'}(y) = \sqrt{\sqrt{2}\sigma'}\,\psi_\infty(\sqrt{2}\sigma'y).
\end{align*}
The approximation error reduces to bounding the decay of Hermite coefficients of $F_{r'}$. The following theorem and corollary make this precise.

\begin{theorem}[Error bound for the finite squeezed-Fock truncation]
\label{thm:state-prep}
    Let $\psi_\infty(x)=\mathcal{N}g(x)$ denote the ideal normalized kernel state, where the normalization constant $\mathcal N$ is chosen so that $\|\psi_\infty\|_2=1$. For a fixed squeezing parameter $r'$, let $\sigma'=e^{r'}$, and let $\mathbb H_{N,r'}$ be the $N$-dimensional subspace spanned by the squeezed Fock basis functions $\{\phi_{n,r'}(x)\}_{n=0}^{N-1}$. Let $\Pi^F_{N,r'}:L^2(\mathbb R)\rightarrow\mathbb H_{N,r'}$ denote the orthogonal projection onto this subspace, and set
    \begin{align*}
        \psi_{N,r'}=\Pi^F_{N,r'}\psi_\infty .
    \end{align*}
    Let the rescaled target $F_{r'}(y) = \sqrt{\sqrt{2}\,\sigma'}\, \psi_\infty(\sqrt{2}\,\sigma' y)$. Then, the following bounds hold.
    \begin{enumerate}
        \item[(i)] Let $\mathcal S$ be the Schwartz space of rapidly decaying $C^{\infty}$ functions on $\mathbb{R}$. If $F_{r'}\in\mathcal S(\mathbb R)$, then for every fixed integer
        $s\geq1$ and $N\geq s$,
        \begin{align*}
            \|\psi_\infty-\psi_{N,r'}\|_2 &\le
            \frac{\|\mathcal{A}^sF_{r'}\|_2}
            {\sqrt{N(N-1)\cdots(N-s+1)}},
            &&\mathcal{A}=\frac{1}{\sqrt2}\left(y+\frac{\dd}{\dd y}\right).
        \end{align*}
        Consequently, for fixed $r'$ and every fixed $s$,
        \begin{align*}
            \|\psi_\infty-\psi_{N,r'}\|_2 = \mathcal O(N^{-s/2}).
        \end{align*}
        \item[(ii)] Let \(\mu\geq \tfrac12\), and suppose that
\(F_{r'}\) belongs to the symmetric Gelfand--Shilov class
\(S_\mu^\mu(\mathbb R)\)~\cite{vanEijndhoven1987GS,lozanovperisic2007hermite}. Equivalently, suppose that there exist
constants \(A_{r'},K_{r'}>0\) such that
\begin{align*}
    \left\|  y^p\frac{\dd^q}{\dd y^q}F_{r'}  \right\|_2
    \leq  K_{r'}A_{r'}^{p+q} (p!)^\mu(q!)^\mu,
    \qquad  p,q\in\mathbb N_0.
\end{align*}
Then there exist constants \(C_{\mu,r'},c_{\mu,r'}>0\),
independent of \(N\), such that
\begin{align*}
    \left\| \psi_\infty-\psi_{N,r'} \right\|_2
    \leq C_{\mu,r'} \exp\!\left[-c_{\mu,r'}N^{1/(2\mu)} \right].
\end{align*}
Consequently, for fixed \(r'\) and \(\mu\), it is sufficient to choose
\begin{align*}
    N = \mathcal O\!\left( \log^{2\mu}\frac{1}{\epsilon_s} \right)
\end{align*}
to obtain a state-projection error at most \(\epsilon_s\).
    \item[(iii)] If, in addition, $F_{r'}$ is analytic in the strip
    $|\operatorname{Im}(y)| < \frac{\rho}{\sqrt{2}\,\sigma'}$ and there exist constants
    $\mathcal K>0$ and $\sigma_0\in\mathbb R$ such that $|e^{y^2/2}F_{r'}(y)| \leq \mathcal K |y|^{\sigma_0}$ as $|y|\rightarrow\infty$ within this strip, and
    \begin{align*}
        \widehat V_{r'} := \int_{\partial\mathcal S_{\rho/(\sqrt{2}\sigma')}}  |e^{y^2/2}F_{r'}(y)|\,|\dd y| < \infty,
    \end{align*}
    then there is a constant $C_{r'}$, independent of $N$, such that
    \begin{align*}
    \|\psi_\infty-\psi_{N,r'}\|_2
    \leq  C_{r'} \exp\!\left[ -\frac{\rho}{\sqrt{2}\sigma'} \sqrt{2(N-1)}  \right]
    = C_{r'} \exp\!\left[ -\frac{\rho}{\sigma'} \sqrt{N-1} \right].
    \end{align*}
    For a target state-projection error \(\epsilon_s>0\), and for fixed \(r'\) and \(\rho\), it is sufficient to choose
    \begin{align*}
        N=\mathcal O\!\left(e^{2r'}\log^2\frac{1}{\epsilon_s}\right).
    \end{align*}
\end{enumerate}
\end{theorem}

\begin{proof}
    See Appendix~\ref{proof:state-prep}.
\end{proof}

\begin{corollary}[Normalized finite projection]
\label{cor:normalized-projection}
Let
\begin{align*}
    \epsilon_{\rm tr}(N)  :=  \|\psi_\infty-\psi_{N,r'}\|_2,  \qquad \psi_{N,r'}  = \Pi^F_{N,r'}\psi_\infty .
\end{align*}
If \(\epsilon_{\rm tr}(N)<1\), define the normalized projection
\begin{align*}
    \widehat{\psi}_{N,r'} := \frac{\psi_{N,r'}}{\|\psi_{N,r'}\|_2}.
\end{align*}
Then
\begin{align*}
    \|\psi_\infty-\widehat{\psi}_{N,r'}\|_2^2 =  2\left(1-\sqrt{1-\epsilon_{\rm tr}(N)^2}\right)  \leq  2\epsilon_{\rm tr}(N)^2.
\end{align*}
Consequently,
\begin{align*}
    \|\psi_\infty-\widehat{\psi}_{N,r'}\|_2
    \leq  \sqrt{2}\,\epsilon_{\rm tr}(N),
\end{align*}
so the convergence rates in Theorem~\ref{thm:state-prep} also hold for the normalized finite projection, up to the constant factor \(\sqrt2\).
\end{corollary}

\begin{proof}
Orthogonality of the projection gives $\|\psi_{N,r'}\|_2^2 =  1-\epsilon_{\rm tr}(N)^2$ and $\langle\psi_\infty,\widehat{\psi}_{N,r'}\rangle =  \sqrt{1-\epsilon_{\rm tr}(N)^2}$. The stated identity follows directly, and \(1-\sqrt{1-\epsilon_{\rm tr}(N)^2}\leq\epsilon_{\rm tr}(N)^2\).
\end{proof}

\begin{corollary}[Stretched-exponential convergence for the LCHS kernel]
\label{cor:near-optimal-kernel}
Let \(g_\beta\) be the LCHS kernel in Eq.~\eqref{eq:kernel_def}, and let \(\psi_\infty=\mathcal N g_\beta\) be the corresponding normalized ideal kernel state. Then, for every fixed squeezing parameter \(r'\), there exist constants \(C_{\beta,r'},c_{\beta,r'}>0\), independent of \(N\), such that
\begin{align}
    \left\| \psi_\infty-\Pi^F_{N,r'}\psi_\infty \right\|_2
    \leq C_{\beta,r'}  \exp\!\left[-c_{\beta,r'}N^{\beta/2}\right].
    \label{eq:kernel-GS-projection-bound}
\end{align}
Consequently, for fixed \(0<\beta<1\) and fixed \(r'\), as \(\epsilon_s\downarrow0\), it is sufficient to choose
\begin{align*}
    N = \mathcal O\!\left( \log^{2/\beta}\frac{1}{\epsilon_s} \right).
\end{align*}
\end{corollary}

\begin{proof}
    See Appendix~\ref{proof:corollary}.
\end{proof}

\paragraph{Two independent oscillator cutoffs.}
The coefficient cutoff \(N\) and the oscillator Hilbert-space cutoff \(N_{\rm Fock}\) represent different approximations. Let \(\ket{n;r'}=S(r')\ket{n}\) and define
\begin{align*}
    \Pi^F_{N,r'}
    &= \sum_{n=0}^{N-1}\ket{n;r'}\!\bra{n;r'},
    & \Pi_{N_{\rm Fock}}
    &= \sum_{m=0}^{N_{\rm Fock}-1}\ket{m}\!\bra{m}.
\end{align*}
The projector \(\Pi^F_{N,r'}\) truncates the squeezed-Fock coefficient expansion, whereas \(\Pi_{N_{\rm Fock}}\) truncates the ordinary number basis used by the finite-dimensional simulator. For nonzero squeezing, these projectors generally do not commute, and a state containing only \(N\) squeezed-Fock components can have support on arbitrarily high ordinary Fock levels.

For the normalized coefficient-truncated state \(\ket{\psi_N}\) and postselection state \(\ket{\phi_r}\), define the ordinary-Fock tails
\begin{align*}
    \tau_{\psi,N_{\rm Fock}}
    &= \left\|  (\mathbb I-\Pi_{N_{\rm Fock}})\ket{\psi_N}  \right\|_2,
    & \tau_{\phi,N_{\rm Fock}}
    &= \left\| (\mathbb I-\Pi_{N_{\rm Fock}})\ket{\phi_r} \right\|_2.
\end{align*}
Whenever the projected norms are nonzero, the error introduced by normalized projection into the finite oscillator space satisfies
\begin{align}
    \epsilon_{\rm Fock}
    := \left\|\ket{\psi_N}  - \frac{\Pi_{N_{\rm Fock}}\ket{\psi_N}}{\|\Pi_{N_{\rm Fock}}\ket{\psi_N}\|_2}\right\|_2
    +  \left\|  \ket{\phi_r} -  \frac{\Pi_{N_{\rm Fock}}\ket{\phi_r}}{\|\Pi_{N_{\rm Fock}}\ket{\phi_r}\|_2}  \right\|_2
    \leq  \sqrt{2} \left( \tau_{\psi,N_{\rm Fock}}  + \tau_{\phi,N_{\rm Fock}} \right). \label{eq:fock-embedding-bound}
\end{align}

Theorem~\ref{thm:state-prep} controls the ideal coefficient-projection error as \(N\) increases. The preceding \(\epsilon_{\rm Fock}\) estimate concerns normalized projection of infinite-dimensional oscillator states. The numerical model uses a different, explicitly finite construction: it truncates the ladder operator first and then exponentiates the finite squeezing generator. For \(1\leq N\leq N_{\rm Fock}\), define
\begin{align*}
    \hat a_{N_{\rm Fock}} &:= \sum_{n=1}^{N_{\rm Fock}-1}\sqrt{n}\ket{n-1}\!\bra{n},
    & \mathsf S_{N_{\rm Fock}}(s)
    &:=\exp\!\left[ \frac{s}{2}\left\{(\hat a_{N_{\rm Fock}}^\dagger)^2-\hat a_{N_{\rm Fock}}^2\right\}\right],\\
    \ket{\chi_N^{(N_{\rm Fock})}} &:=\frac{\sum_{n=0}^{N-1}C_n\ket n}{\left\|\sum_{n=0}^{N-1}C_n\ket n\right\|_2},
    & \hat x_{N_{\rm Fock}} &:= \hat a_{N_{\rm Fock}}+\hat a_{N_{\rm Fock}}^\dagger.
\end{align*}
The finite prepared and postselection states used in the saved numerical computation are
\begin{align}
    \ket{\psi_N^{(N_{\rm Fock})}}
    &:= \frac{\mathsf S_{N_{\rm Fock}}(r')\ket{\chi_N^{(N_{\rm Fock})}}}
    {\|\mathsf S_{N_{\rm Fock}}(r')\ket{\chi_N^{(N_{\rm Fock})}}\|_2},
    &\ket{\phi_r^{(N_{\rm Fock})}}
    &:=\frac{\mathsf S_{N_{\rm Fock}}(r)\ket 0}
    {\|\mathsf S_{N_{\rm Fock}}(r)\ket 0\|_2}.
    \label{eq:finite-target-states}
\end{align}
Their finite overlap scale is
\begin{align*}
    \alpha_{N,r}^{(N_{\rm Fock})}
    &:=\frac{1}{\langle\phi_r^{(N_{\rm Fock})}|\psi_N^{(N_{\rm Fock})}\rangle},
\end{align*}
whenever the overlap is nonzero.
Throughout the finite-dimensional analysis and numerical reporting, we use the \(\hbar=2\) convention \(\hat x_{N_{\rm Fock}}=\hat a_{N_{\rm Fock}}+\hat a_{N_{\rm Fock}}^\dagger\).
To keep the finite-map formulas readable, below we write \(\ket{\psi_N}\), \(\ket{\phi_r}\), and \(\alpha_{N,r}\) for these finite numerical objects and suppress the displayed \(N_{\rm Fock}\) and \(r'\) dependence, while all finite-resource maps keep both cutoff indices. We define the finite target, prepared-state, and product-formula maps, together with their error decomposition, in Section~\ref{ss:postselection}.

\subsection{Non-Gaussianity and Stellar Rank of the Prepared States}
\label{ss:non-gauss}

Let 
\begin{align}
    \ket{\chi_N} =  \sum_{n=0}^{N-1}C_n\ket{n} \label{eq:fock}
\end{align}
denote the normalized finite superposition of unsqueezed Fock states. The associated stellar, or Bargmann, function is
\begin{align*}
    F_{\chi_N}(z) = \sum_{n=0}^{N-1}\frac{C_n}{\sqrt{n!}}z^n .
\end{align*}
The stellar rank \(r_\star(\chi_N)\) is defined as the number of zeros of \(F_{\chi_N}(z)\), counted with multiplicity~\cite{Chabaud_2020}. Let \(n_{\max}=\max \{n: C_n \neq 0\}.\) Then \(F_{\chi_N}(z)\) is a polynomial of degree \(n_{\max},\) and
\begin{align*}
    r_\star(\chi_N) = n_{\max}\leq N-1.
\end{align*}
In particular, if the highest Fock coefficient \(C_{N-1}\) is nonzero, then
\begin{align*}
r_\star(\chi_N) = N-1.
\end{align*}
Otherwise, if \( C_{N-1} = 0 \), the stellar rank satisfies
\begin{align*}
r_\star(\chi_N) < N-1.
\end{align*}
Since \(S(r')\) is a Gaussian unitary and Gaussian unitaries preserve stellar rank, 
\begin{align*}
    r_\star\!\left(\psi_{N,r,r',\beta}\right) = r_\star(\chi_N) = \max\{n:C_n\neq 0\}.
\end{align*}
The stellar rank provides a discrete witness of non-Gaussianity: pure Gaussian states have stellar rank zero, while finite Fock superpositions with positive stellar rank are non-Gaussian. Thus, the truncation cutoff \(N\) sets the maximum stellar-rank complexity available to the implementable oracle state. In the formal limit \(N\to\infty\), the ideal kernel state generally has infinite stellar rank, unless it reduces to a Gaussian state or to a Gaussian unitary applied to a finite superposition of Fock states.

The non-Gaussianity can also be quantified using the quantum relative-entropy (QRE) non-Gaussianity~\cite{PhysRevA.78.060303, Genoni_2010},
    \(\delta_{\rm nG}(\rho)=H(\tau_\rho)-H(\rho),\)
where \(\tau_\rho\) is the Gaussian state with the same first and second moments as \(\rho\), and $H(\cdot)$ is the von Neumann entropy of $(\cdot)$. Since the prepared state is pure and \(S(r')\) is a Gaussian unitary, this quantity can be computed from the unsqueezed core state \(\ket{\chi_N}\):
\begin{align*}
\delta_{\rm nG} \!\left( S(r')\ket{\chi_N} \right) = \delta_{\rm nG}(\ket{\chi_N}).
\end{align*}

Define \(\alpha=\braket{a}, \bar n=\braket{a^\dagger a}, m_2=\braket{a^2},\) where all expectation values are taken with respect to \(\ket{\chi_N}\). Using \(a\ket{n}=\sqrt{n}\ket{n-1}\) and \(a^2\ket{n+2}=\sqrt{(n+1)(n+2)}\ket{n}\), we obtain
\begin{align*}
    \alpha &= \braket{\chi_N|a|\chi_N}
    = \sum_{n=0}^{N-2}C_n^*C_{n+1}\sqrt{n+1},\,\ \bar n
    = \braket{\chi_N|a^\dagger a|\chi_N}
    = \sum_{n=0}^{N-1}n|C_n|^2, \\
    m_2 &= \braket{\chi_N|a^2|\chi_N}
    = \sum_{n=0}^{N-3}C_n^*C_{n+2}\sqrt{(n+1)(n+2)}.
\end{align*}
The centered second moments are
\(
N_c = \braket{\Delta a^\dagger \Delta a}
= \bar n-|\alpha|^2,\qquad
M_c = \braket{(\Delta a)^2}
= m_2-\alpha^2,
\)
with \(\Delta a=a-\braket{a}\).
The determinant of the single-mode covariance matrix is
\begin{align*}
    \det V = \left(N_c+\frac{1}{2}\right)^2-|M_c|^2.
\end{align*}
The corresponding symplectic eigenvalue is
\begin{align*}
    \nu = \sqrt{\det V} = \sqrt{\left(N_c+\frac12\right)^2 - |M_c|^2}.
\end{align*}
Therefore, for the pure state \(\ket{\chi_N}\)~\cite{Weedbrook_2012},
\begin{align*}
    \delta_{\rm nG} = \left(\nu+\frac12\right)\ln\left(\nu+\frac12\right)
    - \left(\nu-\frac12\right)\ln\left(\nu-\frac12\right).
\end{align*}
This gives the exact relative-entropy non-Gaussianity of the prepared state in terms of the expansion coefficients \(C_n\). All second moments here are mode-operator moments, and this normalization fixes the vacuum symplectic eigenvalue at \(\nu=\tfrac12\), so \(\delta_{\rm nG}\) is independent of the quadrature convention used elsewhere in the paper.

Finally, consider the hybrid interaction in Fig.~\ref{fig:main-circuit}. In the eigenbasis of \(L\),
\begin{align*}
e^{-it(\hat{x}\otimes L)} = \sum_{\lambda\in{\rm spec}(L)}
e^{-it\lambda\hat{x}}\otimes\ketbra{\lambda}{\lambda}.
\end{align*}
For a fixed \(L\)-eigenvalue \(\lambda\), the oscillator undergoes the Gaussian momentum displacement \(e^{-it\lambda\hat{x}}\). Within that branch, the interaction is Gaussian and preserves oscillator stellar rank.

Branchwise Gaussianity is insufficient to classify the complete hybrid operation. For a superposition of \(L\)-eigenstates, the interaction coherently correlates distinct oscillator displacements with the system components, and postselection induces
\begin{align*}
    \mathcal K_{\phi,\psi} = \sum_{\lambda} \bra{\phi}e^{-it\lambda\hat{x}}\ket{\psi} \ketbra{\lambda}{\lambda}.
\end{align*}
This map is generally a nontrivial spectral function of \(L\), so branchwise Gaussianity leaves the complete entangled evolution and postselected system map unclassified. We quantify continuous-variable non-Gaussianity of the prepared oscillator state. Discrete-variable magic lies outside this analysis.

\subsection{Trotter-Suzuki Approximation for Hamiltonian Evolution}
\label{sec:trotter}
\begin{figure}[h]
    \centering
    \begin{align*}
    \begin{aligned}
        &
        \Qcircuit @C=0.7em @R=1.0em {
            \lstick{\ket{\psi}_{\rm osc}}
                & \multigate{4}{e^{-i\Delta t A_1}}
                & \multigate{4}{e^{-i\Delta t A_2}}
                & \qw
                & \cdots
                & 
                & \multigate{4}{e^{-i\Delta t A_{N_L}}}
                & \push{\rule{0em}{1em}}{\qw}
                & \qw
                & \qw
                & \qw
                & \cdots
                & 
                & \qw
                & \qw
            \\
            \lstick{\ket{q_1}}
                & \ghost{e^{-i\Delta t A_1}}
                & \ghost{e^{-i\Delta t A_2}}
                & \qw
                & \cdots
                & 
                & \ghost{e^{-i\Delta t A_{N_L}}}
                & \push{\rule{0em}{1em}}{\qw}
                & \multigate{3}{e^{-i\Delta t B_1}}
                & \multigate{3}{e^{-i\Delta t B_2}}
                & \qw
                & \cdots
                & 
                & \multigate{3}{e^{-i\Delta t B_{N_H}}}
                & \qw
            \\
            \lstick{\ket{q_2}}
                & \ghost{e^{-i\Delta t A_1}}
                & \ghost{e^{-i\Delta t A_2}}
                & \qw
                & \cdots
                & 
                & \ghost{e^{-i\Delta t A_{N_L}}}
                & \push{\rule{0em}{1em}}{\qw}
                & \ghost{e^{-i\Delta t B_1}}
                & \ghost{e^{-i\Delta t B_2}}
                & \qw
                & \cdots
                & 
                & \ghost{e^{-i\Delta t B_{N_H}}}
                & \qw
            \\
            \lstick{{\vdots}}
                & \ghost{e^{-i\Delta t A_1}}
                & \ghost{e^{-i\Delta t A_2}}
                & \qw
                & \cdots
                & 
                & \ghost{e^{-i\Delta t A_{N_L}}}
                & \push{\rule{0em}{1em}}{\qw}
                & \ghost{e^{-i\Delta t B_1}}
                & \ghost{e^{-i\Delta t B_2}}
                & \qw
                & \cdots
                & 
                & \ghost{e^{-i\Delta t B_{N_H}}}
                & \qw
            \\
            \lstick{\ket{q_m}}
                & \ghost{e^{-i\Delta t A_1}}
                & \ghost{e^{-i\Delta t A_2}}
                & \qw
                & \cdots
                & 
                & \ghost{e^{-i\Delta t A_{N_L}}}
                & \push{\rule{0em}{1em}}{\qw}
                & \ghost{e^{-i\Delta t B_1}}
                & \ghost{e^{-i\Delta t B_2}}
                & \qw
                & \cdots
                & 
                & \ghost{e^{-i\Delta t B_{N_H}}}
                & \qw
        }\\
        &\underbrace{
        \hspace{6.4cm}
        }_{\text{Pauli decomposition of } \hat x \otimes L}
        \quad
        \underbrace{
        \hspace{6.4cm}
        }_{\text{Pauli decomposition of }H}
        \\[-1em]
        \end{aligned}
        \end{align*}
    \caption{Schematic of a single Suzuki--Trotter layer \(S_p(\Delta t)\), with \(\Delta t=t/n_t\), used to approximate the hybrid evolution. The full product-formula approximation is obtained by repeating this layer \(n_t\) times. Define $L=\sum_{i=1}^{N_L}\alpha_iP_i, H=\sum_{j=1}^{N_H}\beta_jQ_j$, and let $A_i:=\alpha_i\hat{x}_{N_{\rm Fock}}\otimes P_i,B_j:=\beta_j\mathbb I_{\rm osc}\otimes Q_j, \Delta t:=t/n_t$.}
    \label{fig:trotter_hybrid}
\end{figure}
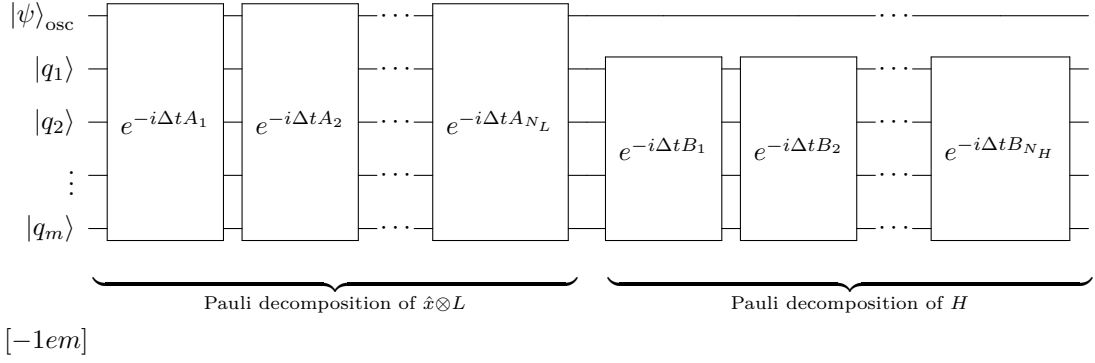

Figure~\ref{fig:trotter_hybrid} shows the Suzuki--Trotter implementation of the joint evolution block in Fig.~\ref{fig:main-circuit}. The hybrid Hamiltonian $H_{{\rm hyb},N_{\rm Fock}}=\hat{x}_{N_{\rm Fock}}\otimes L+\mathbb I_{\rm osc}\otimes H$ requires Suzuki--Trotter synthesis into a sequence of implementable operations. The number of Trotter steps $n_t$ required for target precision $\epsilon_t$ depends on the non-commutativity of the inner Pauli decomposition of $L$, whose exponentials carry the norm of the truncated oscillator quadrature $\hat{x}_{N_{\rm Fock}}$, the outer CV--DV split, and the inner Pauli decomposition of $H$.

Let
\begin{align*}
\Pi_{N_{\rm Fock}}
= \sum_{\ell=0}^{N_{\rm Fock}-1} \ket{\ell}\!\bra{\ell}, \qquad \hat{x}_{N_{\rm Fock}}
= \Pi_{N_{\rm Fock}} \hat{x} \Pi_{N_{\rm Fock}}.
\end{align*}
The finite hybrid Hamiltonian is
\begin{align*}
H_{{\rm hyb},N_{\rm Fock}} = \hat{x}_{N_{\rm Fock}}\otimes L + \mathbb I_{\rm osc}\otimes H.
\end{align*}
Under the convention \(\hat{x}=\hat a+\hat a^\dagger\),
\begin{align}
    \|\hat{x}_{N_{\rm Fock}}\| &\leq 2\sqrt{N_{\rm Fock}-1} = \mathcal O(\sqrt{N_{\rm Fock}}). \label{eq:x-norm-asymptotic}
\end{align}
The following theorem bounds the product-formula error for this finite-oscillator representation. The error of the finite oscillator representation itself is accounted for separately, within $\epsilon_{\rm model}$, in Section~\ref{ss:postselection}.

\begin{theorem}[Resource Complexity for Hybrid Hamiltonian Evolution]
\label{thm:trotter}
    Consider the finite-dimensional hybrid Hamiltonian 
    $H_{{\rm hyb},N_{\rm Fock}}=\hat{x}_{N_{\rm Fock}}\otimes L+\mathbb I_{\rm osc}\otimes H$ 
    with Pauli decompositions 
    \begin{align*} 
    L=\sum_{i=1}^{N_L} \alpha_i P_i, \quad H=\sum_{j=1}^{N_H} \beta_j Q_j.
    \end{align*}
    For a $p^{\rm th}$-order Suzuki--Trotter formula $S_p(t/n_t)$, define the mixed Pauli-level commutator sums $\Gamma_{p,a}^{(L,H)}$ for words containing exactly $a$ $L$-sector summands and $(p+1-a)$ $H$-sector summands, and define $\Gamma_p^{(L)}$ and $\Gamma_p^{(H)}$ for the pure inner Pauli splits of $L$ and $H$:
    \begin{align*}
        \Gamma_{p,a}^{(L,H)} 
        &\equiv   \sum_{\substack{R_1,\cdots,R_{p+1} \in  \{\alpha_iP_i\}_{i=1}^{N_L}\cup\{\beta_jQ_j\}_{j=1}^{N_H} \\ 
        \#_L(R_1,\cdots,R_{p+1})=a  \\
        \#_H(R_1,\cdots, R_{p+1})=p+1-a}} 
        \|[R_{p+1},[R_p,\cdots,[R_2,R_1]\cdots]]\|\\
        \Gamma_p^{(L)} 
        &\equiv  \sum_{i_1,\cdots,i_{p+1}} |\alpha_{i_1}\cdots\alpha_{i_{p+1}}| 
        \|[P_{i_{p+1}},[P_{i_p},\cdots,[P_{i_2},P_{i_1}]\cdots]]\|\\
        \Gamma_p^{(H)} 
        &\equiv  \sum_{j_1,\cdots,j_{p+1}} |\beta_{j_1}\cdots\beta_{j_{p+1}}| 
        \|[Q_{j_{p+1}},[Q_{j_p},\cdots,[Q_{j_2},Q_{j_1}]\cdots]]\|\\
    \end{align*}
    Then it is sufficient to choose the number of Trotter steps as
    \begin{align*}
        n_t=\Ocal \left[t^{1+1/p} \left( \frac{\Lambda_p(N_{\rm Fock})} {\epsilon_t} \right)^{1/p} \right],        
    \end{align*}
    where
        \begin{align}
            \Lambda_p(N_{\rm Fock}) :=
            \left[ \sum_{a=1}^{p} 2^a(N_{\rm Fock}-1)^{a/2} \Gamma_{p,a}^{(L,H)} + 2^{p+1}(N_{\rm Fock}-1)^{(p+1)/2} \Gamma_p^{(L)} + \Gamma_p^{(H)} \right], \label{eq:Lambda_p_definition}
        \end{align}
    to ensure
    \begin{align*}
        \left\| e^{-iH_{{\rm hyb},N_{\rm Fock}}t} -
        \left[ S_p\left(\frac{t}{n_t}\right) \right]^{n_t} \right\| \leq \epsilon_t.
    \end{align*}
    The corresponding gate counts scale as
    \begin{align*}
        N_{\rm CX}=\mathcal{O}\left(2n_tm_p\left[\sum_{i=1}^{N_L}\max\!\bigl(w(P_i)-1,0\bigr)+\sum_{j=1}^{N_H}\max\!\bigl(w(Q_j)-1,0\bigr)\right]\right),
    \end{align*}
    and
    \begin{align*}
        N_{\rm hyb}=\mathcal{O}(n_tm_pN_L),
    \end{align*}
    where $m_p$ is a formula- and convention-dependent Suzuki traversal multiplier. For standard recursive formulas of even order $p=2k$, this multiplier grows as $\mathcal O(5^{k-1})$. The total number of exponentials additionally depends on the number of implemented Pauli summands, which is accounted for by the explicit sums over the $N_L$ and $N_H$ terms above. Here, $N_L$ and $N_H$ are the numbers of Pauli terms in $L$ and $H$, and $w(\cdot)$ denotes Pauli weight.
\end{theorem}
\begin{proof}
    See Appendix~\ref{proof:trotter}.
\end{proof}

\begin{remark}[Pure qubit Hamiltonian]
For \(L=0\), all mixed and pure-\(L\) commutator contributions vanish. Thus, the bound reduces to the standard commutator-based Suzuki--Trotter scaling for a qubit Hamiltonian,
\begin{align*}
n_t = \Ocal\!\left[ t^{1+1/p} \left(\frac{\Gamma_p^{(H)}}{\epsilon_t}\right)^{1/p}\right].
\end{align*}
\end{remark}

\begin{remark}[Pure oscillator--qubit interaction]
A mixed nested commutator containing exactly \(a\) \(L\)-sector summands contributes an explicit cutoff factor \(\mathcal O(N_{\rm Fock}^{a/2})\), whereas the pure-\(L\) contribution scales as \(\mathcal O(N_{\rm Fock}^{(p+1)/2})\). Therefore, when \(\Gamma_p^{(L)}\neq 0\), the pure-\(L\) sector determines the worst-case oscillator-cutoff scaling.
For \(H=0\), the mixed and pure-\(H\) commutator contributions vanish, while the pure-\(L\) inner-splitting contribution remains. Hence,
\begin{align*}
n_t = \Ocal\!\left[ t^{1+1/p} \left(\frac{N_{\rm Fock}^{(p+1)/2}\Gamma_p^{(L)}} {\epsilon_t} \right)^{1/p} \right],
\end{align*}
where the constant factor \(2^{(p+1)/p}\) has been absorbed into the asymptotic notation.

This case applies to spatially semi-discretized linear parabolic PDEs, whose real symmetric positive-semidefinite coefficient matrices have \(H=0\) in the Cartesian decomposition, as in the heat-equation examples of Section~\ref{sec:examples}. For the Dirichlet heat-equation benchmark, the additional structure of the finite-difference Laplacian yields the sharper first-order estimate in Eq.~\eqref{eq:heat-trotter-error}.
\end{remark}

\subsection{Joint Cutoff--Depth Trade-off under a Global Error Budget}
\label{subsec:joint_tradeoff}

The coefficient cutoff \(N\) and the product-formula depth \(n_t\) become coupled once a global accuracy target is fixed. Increasing \(N\) reduces the coefficient-projection error \(\epsilon_{\rm tr}(N)=\|\psi_\infty-\Pi^F_{N,r'}\psi_\infty\|_2\) controlled by Theorem~\ref{thm:state-prep}. For a fixed kernel and \(r'\), this coefficient-projection error depends on \(N\) but not on \(N_{\rm Fock}\). Representing an \(N\)-component squeezed-Fock state on the finite simulator requires \(N_{\rm Fock}\geq N\), and in practice we use a fixed embedding ratio \(N_{\rm Fock}=\lceil\kappa N\rceil\) with \(\kappa>1\). The numerical study uses \(\kappa=2\) for \((N,N_{\rm Fock})=(32,64)\) and \(\kappa=4/3\) in the scaling sweeps. Through the embedding constraint, a larger coefficient cutoff increases \(\|\hat{x}_{N_{\rm Fock}}\|\leq2\sqrt{N_{\rm Fock}-1}\) and the commutator factor \(\Lambda_p(N_{\rm Fock})\) in Theorem~\ref{thm:trotter}, so the certified product-formula depth grows with \(N\).

To quantify it, we allocate a global budget \(\epsilon\) between the two certified error components,
\begin{align}
    \epsilon_{\rm tr}(N)+\epsilon_t\leq\epsilon, \label{eq:joint-budget-split}
\end{align}
where \(\epsilon_{\rm tr}(N)\) is the squeezed-Fock truncation error of the ideal kernel state at coefficient cutoff \(N\), that is, the coefficient-projection error defined in Corollary~\ref{cor:normalized-projection} and controlled by Theorem~\ref{thm:state-prep}, and \(\epsilon_t\) is the Trotter error tolerance of the product-formula simulation of the finite joint evolution in Theorem~\ref{thm:trotter}.

The finite-\(r\) coefficient-generating pairing, the finite-\(N_{\rm Fock}\) embedding and truncated joint dynamics, state synthesis, and the postselection scale are accounted for separately through \(\epsilon_{\rm model}\) and \(\epsilon_{\rm synth}\) in Section~\ref{ss:postselection} and are not bounded by Eq.~\eqref{eq:joint-budget-split}. 

For any admissible cutoff with \(\epsilon_{\rm tr}(N)<\epsilon\), if the \(p\)th-order product-formula error of Theorem~\ref{thm:trotter} is written as \(\epsilon_t\leq c_pt^{p+1}\Lambda_p(N_{\rm Fock})/n_t^p\), the smallest certified step count at the embedding ratio \(\kappa\) is
\begin{align}
    n_t(N;\epsilon)
    := \left\lceil\left[ \frac{c_pt^{p+1}\Lambda_p(\lceil\kappa N\rceil)}{\epsilon-\epsilon_{\rm tr}(N)}\right]^{1/p}\right\rceil.
    \label{eq:nt-joint}
\end{align}
In the compiled circuit, the preparation cost grows linearly in \(N\) (the Law--Eberly synthesis uses \(N-1\) pulse pairs), while the evolution cost grows linearly in \(n_t\). The product
\begin{align}
    \mathcal R(N;\epsilon)
    := N\,n_t(N;\epsilon), \qquad N_\star(\epsilon) \in \underset{N:\,\epsilon_{\rm tr}(N)<\epsilon}{\operatorname{argmin}} \;\mathcal R(N;\epsilon), \label{eq:joint-proxy}
\end{align}
is the joint cutoff--depth proxy suggested by this coupling. It upper-bounds the additive circuit-cost structure up to a factor.

For fixed \(\kappa\), product-formula order \(p\), theorem prefactors, and error budget \(\epsilon\), Eq.~\eqref{eq:joint-proxy} defines a conditional analytical resource proxy whose minimizer \(N_\star\) depends on those choices and on the cost proxy itself, a dependence suppressed in the notation.
When the highest-order pure-\(L\) commutator sum is nonzero, Eq.~\eqref{eq:Lambda_p_definition} gives \(\Lambda_p(\lceil\kappa N\rceil)=\mathcal O(N^{(p+1)/2})\), with \(\kappa\) absorbed into the constant. Hence, for any cutoff satisfying \(\epsilon_{\rm tr}(N)\leq\xi\epsilon\) with fixed \(0<\xi<1\),
\begin{align}
    n_t(N;\epsilon)=\mathcal O\!\left[t^{1+1/p}\epsilon^{-1/p}N^{(p+1)/(2p)}\right],
    \qquad
    \mathcal R(N;\epsilon)=\mathcal O\!\left[t^{1+1/p}\epsilon^{-1/p}N^{(3p+1)/(2p)}\right]. \label{eq:joint-asymptotic}
\end{align}
Corollary~\ref{cor:near-optimal-kernel} implies that admissible cutoffs exist for sufficiently large \(N\), since \(\epsilon_{\rm tr}(N)\) decays faster than any inverse power of \(N\). Moreover, \(n_t(N;\epsilon)\) is a positive integer, so \(\mathcal R(N;\epsilon)=Nn_t(N;\epsilon)\geq N\) and the proxy diverges as \(N\to\infty\). Its minimum over the admissible integer cutoffs is attained at a finite \(N_\star(\epsilon)\).

Fig.~\ref{fig:joint-cutoff-depth} evaluates this proxy for the selected kernel parameter set \((r,r',\beta)=(1.6,0.25,0.5)\). The truncation error \(\epsilon_{\rm tr}(N)\) is computed numerically from the squeezed-Fock coefficients of \(\psi_\infty\) up to \(n=223\), with the residual mass beyond the computed window recovered from Parseval's identity. 

\begin{figure}
    \centering
    \includegraphics[width=\linewidth]{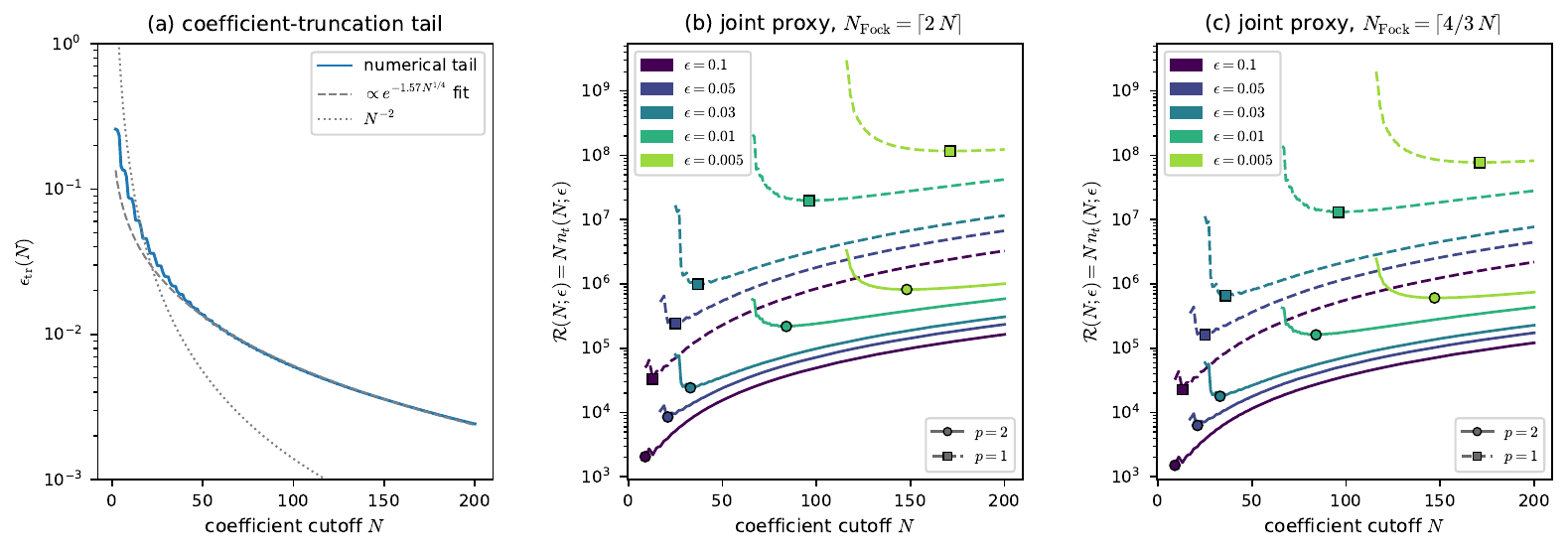}
    \caption{Illustrative normalized proxy for the joint cutoff--depth trade-off under a global error budget \(\epsilon\). (a) Numerically evaluated coefficient-truncation error \(\epsilon_{\rm tr}(N)=\|\psi_\infty-\Pi^F_{N,r'}\psi_\infty\|_2\) for the selected kernel parameters, together with a stretched-exponential fit and the power law $N^{-2}$ for comparison. (b,c) Resource proxy \(\mathcal R(N;\epsilon)=N\,n_t(N;\epsilon)\) from Eqs.~\eqref{eq:nt-joint} and~\eqref{eq:joint-proxy} with normalized theorem-dependent prefactors and embedding ratios \(N_{\rm Fock}=\lceil2N\rceil\) and \(\lceil4N/3\rceil\), for the second-order formula \(p=2\) (solid, circular markers) and the first-order formula \(p=1\) used in the numerical experiments (dashed, square markers). Markers denote the minimizing cutoff \(N_\star(\epsilon)\) of each curve.}
\label{fig:joint-cutoff-depth}
\end{figure}

In Fig.~\ref{fig:joint-cutoff-depth}(a), the measured tail is well described by the finite-range fit \(\epsilon_{\rm tr}(N)\propto e^{-1.57N^{1/4}}\) over \(40\leq N\leq200\). This behavior is consistent with the stretched-exponential upper-bound form \(\exp[-c_{\beta,r'}N^{\beta/2}]\) guaranteed by Corollary~\ref{cor:near-optimal-kernel} for \(\beta=1/2\). The fitted coefficient \(1.57\) is a descriptive parameter over the evaluated cutoff range and is not identified with the existence constant \(c_{\beta,r'}\) in the corollary.

Fig.~\ref{fig:joint-cutoff-depth}(b) and (c) show \(\mathcal R(N;\epsilon)\) with normalized theorem-dependent prefactors (\(c_pt^{p+1}=1\) and all nested-commutator sums set to one) at \(\kappa=2\) and \(\kappa=4/3\), for \(p=2\) and for the first-order formula \(p=1\) used in the numerical experiments. The minimizing cutoff increases from \(N_\star=9\) at \(\epsilon=10^{-1}\) to \(N_\star=148\) at \(\epsilon=5\times10^{-3}\) for \(p=2\) and from \(N_\star=13\) to \(N_\star=171\) for \(p=1\). Across both embedding ratios, \(N_\star\) changes little. First-order depths exceed second-order depths by one to two orders of magnitude. Both displayed orders show the same qualitative trade-off. For the first-order formula used in the numerical experiments, the same construction can be evaluated with the exact commutator structure of the Dirichlet benchmark, whose compiled inner splitting (Section~\ref{1d-heat}) contains a single noncommuting block pair with commutator norm one. At the operating point \((N,N_{\rm Fock},n_t)=(32,64,100)\), the certified first-order error bound is \(1.11\) with the exact \(\|\hat x_{64}\|=14.89\), about \(1.2\times10^{3}\) times the measured \(\epsilon_t=8.90\times10^{-4}\) of Section~\ref{sec:clean_experiments}, and certifying the measured accuracy would require \(n_t\approx1.2\times10^{5}\) steps. The measured product-formula error supports the operating cutoff. The kernel state concentrates at low Fock levels, so using the full truncated-quadrature norm produces the much larger worst-case certified bound.

\subsection{Postselection probability and sampling overhead}
\label{ss:postselection}

The finite-resource construction uses the exact finite-dimensional hybrid evolution
\begin{align*}
    U_{N_{\rm Fock}}(t) = e^{-iH_{{\rm hyb},N_{\rm Fock}}t}, \qquad H_{{\rm hyb},N_{\rm Fock}} = \hat{x}_{N_{\rm Fock}}\otimes L + \mathbb I_{\rm osc}\otimes H.
\end{align*}
Let \(U_{n_t}(t)\) denote its \(n_t\)-step product-formula approximation, and let \(\ket{\widetilde\psi_N}\) denote the state produced by the selected synthesis method as an approximation to the finite target state \(\ket{\psi_N}\). Define
\begin{align}
    \mathcal B^{\rm tar}_{N,N_{\rm Fock}}(t) &:= (\bra{\phi_r}_{\rm osc}\otimes\mathbb I_q) U_{N_{\rm Fock}}(t) (\ket{\psi_N}_{\rm osc}\otimes\mathbb I_q),\\
    \mathcal B^{\rm prep}_{N,N_{\rm Fock}}(t) &:= (\bra{\phi_r}_{\rm osc}\otimes\mathbb I_q) U_{N_{\rm Fock}}(t) (\ket{\widetilde\psi_N}_{\rm osc}\otimes\mathbb I_q),\\ \widetilde{\mathcal K}_{N,N_{\rm Fock},n_t}(t) &:= (\bra{\phi_r}_{\rm osc}\otimes\mathbb I_q) U_{n_t}(t) (\ket{\widetilde\psi_N}_{\rm osc}\otimes\mathbb I_q). \label{eq:evol-op}
\end{align}

For a normalized input state \(\ket{u_0}\), define the unnormalized finite postselected branch and its physical success probability by
\begin{align*}
    \ket{v_{N,N_{\rm Fock},n_t}(t)} &= \widetilde{\mathcal K}_{N,N_{\rm Fock},n_t}(t)\ket{u_0},\\
    p_{\rm succ}^{(N,N_{\rm Fock},n_t)}(t;u_0) &= \left\| \widetilde{\mathcal K}_{N,N_{\rm Fock},n_t}(t)\ket{u_0} \right\|_2^2.
\end{align*}
Conditioned on successful postselection, the normalized physical output state is
\begin{align*}
    \ket{u_{N,N_{\rm Fock},n_t}(t)} = \frac{\widetilde{\mathcal K}_{N,N_{\rm Fock},n_t}(t)\ket{u_0}}{\sqrt{p_{\rm succ}^{(N,N_{\rm Fock},n_t)}(t;u_0)}}.
\end{align*}

Recall that
\begin{align*}
    \alpha_{N,r} = \frac{1}{\langle\phi_r|\psi_N\rangle},
\end{align*}
where the overlap is evaluated using the normalized finite target oscillator states and is assumed to be nonzero. Define the total scaled-map error
\begin{align}
    \epsilon_{\rm tot} := \left\|\alpha_{N,r}\widetilde{\mathcal K}_{N,N_{\rm Fock},n_t}(t)-e^{-At}\right\|_2. \label{eq:total-scaled-map-error}
\end{align}
Because the physical probability is computed from the unscaled postselected branch, the corresponding finite-resource reference probability is
\begin{align*}
    p_{\rm ref}^{(N,r)}(t;u_0):= \frac{\left\|e^{-At}\ket{u_0}\right\|_2^2}{|\alpha_{N,r}|^2}.
\end{align*}

\begin{theorem}[Perturbation of the postselection probability]
\label{thm:postselection-success}
For every normalized input state \(\ket{u_0}\),
\begin{align}
    \left| p_{\rm succ}^{(N,N_{\rm Fock},n_t)}(t;u_0)- p_{\rm ref}^{(N,r)}(t;u_0)\right|
    \leq
    \frac{2\left\|e^{-At}\ket{u_0}\right\|_2\epsilon_{\rm tot}+\epsilon_{\rm tot}^2}{|\alpha_{N,r}|^2}.\label{eq:psucc-perturbation}
\end{align}
Since \(L\succeq0\) implies \(\|e^{-At}\|\leq1\), this further gives
\begin{align*}
    \left| p_{\rm succ}^{(N,N_{\rm Fock},n_t)}(t;u_0) - p_{\rm ref}^{(N,r)}(t;u_0) \right|
    \leq \frac{2\epsilon_{\rm tot} + \epsilon_{\rm tot}^2}{|\alpha_{N,r}|^2}.
\end{align*}
\end{theorem}

\begin{proof}
See Appendix~\ref{proof:postselection-success}.
\end{proof}

Let
\begin{align*}
    \epsilon_{\rm model} &:= \left\|\alpha_{N,r}\mathcal B^{\rm tar}_{N,N_{\rm Fock}}(t)-e^{-At}\right\|_2,\\
    \epsilon_{\rm synth} &:=\left\| \mathcal B^{\rm prep}_{N,N_{\rm Fock}}(t) -\mathcal B^{\rm tar}_{N,N_{\rm Fock}}(t) \right\|_2,\\
    \epsilon_t(n_t) &:= \left\|\widetilde{\mathcal K}_{N,N_{\rm Fock},n_t}(t)-\mathcal B^{\rm prep}_{N,N_{\rm Fock}}(t) \right\|_2.
\end{align*}
The two-cutoff error decomposition gives
\begin{align*}
    \epsilon_{\rm tot} \leq \epsilon_{\rm model} + |\alpha_{N,r}|\left[ \epsilon_{\rm synth} + \epsilon_t(n_t) \right].
\end{align*}
Here, \(\epsilon_{\rm model}\) contains the coefficient approximation, finite-squeezing regularization, finite-\(N_{\rm Fock}\) embedding, and dynamical-truncation effects. \(\epsilon_{\rm synth}\) and \(\epsilon_t\) denote the unscaled map errors introduced by oscillator-state synthesis and product-formula simulation, respectively.
The numerical section reports the three components for the benchmark and breadth experiments and for the circuit-level scaling rows, so the dominant error source can be identified in each case.

The same total scaled-map error controls the normalized conditional output. If
\(
    \epsilon_{\rm tot} < \left\| e^{-At}\ket{u_0}\right\|_2,
\)
then
\begin{align*}
    \min_{\theta\in\mathbb R} \left\|e^{i\theta} \ket{u_{N,N_{\rm Fock},n_t}(t)}
    - \frac{e^{-At}\ket{u_0}}{\left\|e^{-At}\ket{u_0}\right\|_2}\right\|_2
    \leq \frac{2\epsilon_{\rm tot}}{\left\|e^{-At}\ket{u_0}\right\|_2}.
\end{align*}
The phase accounts for the generally complex normalization factor \(\alpha_{N,r}\).

For comparison, the exact normalized factorization in Theorem~\ref{thm:CV--DV-lchs} has physical probability
\begin{align*}
    p_{\rm succ}^{(\infty)}(t;u_0) = \left\| \mathcal B_g(t)\ket{u_0} \right\|_2^2 = \frac{\left\|e^{-At}\ket{u_0}\right\|_2^2}{|\alpha_g|^2}.
\end{align*}
Writing \(\ket{u(t)}=e^{-At}\ket{u_0}\), one obtains
\begin{align*}
    \frac{\dd}{\dd t}\|\ket{u(t)}\|_2^2 = -2\bra{u(t)}L\ket{u(t)}.
\end{align*}
For a normalized initial state, the Rayleigh--Ritz inequality gives
\begin{align*}
    e^{-2t\|L\|} \leq \left\| e^{-At}\ket{u_0} \right\|_2^2 \leq e^{-2t\lambda_{\min}(L)}.
\end{align*}
Consequently,
\begin{align}
    \frac{e^{-2t\|L\|}}{|\alpha_g|^2} \leq p_{\rm succ}^{(\infty)}(t;u_0) \leq \frac{e^{-2t\lambda_{\min}(L)}}{|\alpha_g|^2} \leq 1. \label{eq:psucc-time-bounds}
\end{align}
The final inequality uses \(|\alpha_g|\geq1\), which follows from \(1=\int_{\Rbb} g(x)\dd x=\alpha_g\braket{\phi|\psi}\) and \(|\braket{\phi|\psi}|\leq1\) for normalized states, whereas the finite-resource reference probability satisfies
\begin{align*}
    \frac{e^{-2t\|L\|}}{|\alpha_{N,r}|^2} \leq p_{\rm ref}^{(N,r)}(t;u_0) \leq \frac{e^{-2t\lambda_{\min}(L)}}{|\alpha_{N,r}|^2}.
\end{align*}

Define
\begin{align*}
    \delta_p := \frac{ 2\left\|e^{-At}\ket{u_0}\right\|_2\epsilon_{\rm tot} + \epsilon_{\rm tot}^2}{|\alpha_{N,r}|^2}.
\end{align*}
Theorem~\ref{thm:postselection-success} implies
\begin{align*}
    p_{\rm succ}^{(N,N_{\rm Fock},n_t)}(t;u_0) \geq p_{\rm ref}^{(N,r)}(t;u_0)-\delta_p.
\end{align*}
The expected number of raw circuit executions needed to obtain \(N_{\rm acc}\) accepted samples is
\begin{align*}
    \mathbb E[N_{\rm raw}] = \frac{N_{\rm acc}}{p_{\rm succ}^{(N,N_{\rm Fock},n_t)}(t;u_0)}.
\end{align*}
Whenever \(p_{\rm ref}^{(N,r)}(t;u_0)>\delta_p\), it follows that
\begin{align*}
    \mathbb E[N_{\rm raw}] \leq \frac{N_{\rm acc}}{p_{\rm ref}^{(N,r)}(t;u_0)-\delta_p}.
\end{align*}

Amplitude amplification would reduce this direct-repetition overhead from $\mathcal O(1/p_{\rm succ})$ to $\mathcal O(1/\sqrt{p_{\rm succ}})$ per accepted sample. The present circuits omit the required coherent reflections about the prepared oscillator state and the postselection outcome, and their cost lies outside the present resource accounting. The reported output-state fidelities are conditioned on successful postselection. The additional raw-sampling overhead is determined by \(p_{\rm succ}^{(N,N_{\rm Fock},n_t)}\) and its inverse and must be reported separately from state-preparation and gate costs.

The remaining implementation-dependent component is the problem-specific basis-gate synthesis of the hybrid evolutions generated by \(L\) and \(H\), which follows the Pauli-decomposition designs used by DV solvers~\cite{PhysRevResearch.6.033246, alipanah2025quantum} and the hybrid circuit synthesis of~\cite{bell2025co}.


\section{Simulating linear ODE and PDE examples}
\label{sec:examples}

\subsection{Example: Damped Harmonic Oscillator}

Consider a damped harmonic oscillator with position $x(t)$, damping rate $\zeta \geq 0$, angular frequency $\omega > 0$, and underdamped angular frequency $\omega_d:=\sqrt{\omega^2 - \zeta^2/4}>0$, governed by $\ddot{x} + \zeta \dot{x} + \omega^2 x = 0$. Introducing $v:= \frac{1}{\omega_d} (\dot{x}+\frac{\zeta}{2}x)$ turns this second-order equation into the first-order linear system $\frac{\dd}{\dd t}(x,v)^\top = -A\,(x,v)^\top$ with
\begin{align*}
A = \begin{bmatrix} \zeta/2 & -\omega_d \\ \omega_d & \zeta/2 \end{bmatrix}, \qquad L = \frac{A + A^\dagger}{2} = \frac{\zeta}{2} I \succeq 0, \qquad H = \frac{A - A^\dagger}{2i} = -\omega_d Y,
\end{align*}
where $I$ is the identity matrix and $Y$ is the Pauli-$Y$ matrix. Since \(L\propto I\), the two terms commute and $e^{-it(kL+H)} = e^{-i\zeta tk/2}\,e^{i\omega_d tY}$. The dissipative contribution is thus encoded in each LCHS branch as the $k$-dependent phase $e^{-i\zeta tk/2}$, the Hamiltonian contribution generates the qubit rotation $e^{i\omega_d tY}$, and the physical damping factor is recovered only after the CV-weighted integral and postselection.

\subsection{Example: Heat Equation}
\subsubsection{One-dimensional Dirichlet Case}
\label{1d-heat}

Consider the one-dimensional heat equation
\begin{align*}
    \frac{\partial u(x,t)}{\partial t} = \alpha \frac{\partial^2 u(x,t)}{\partial x^2},
    \qquad 0 < x < \mathcal{L},
\end{align*}
with initial condition $u(x,0)=u_0(x)$ and homogeneous Dirichlet boundary conditions $u(0,t)=u(\mathcal{L},t)=0$.
For the binary encoding used here, let $m$ be the number of qubits in the DV register and take $M=2^m$ interior grid points, let $h=\mathcal{L}/(M+1)$, and write $u_j(t)=u(x_j,t)$ for the interior values. The second-order central difference approximation gives the ODE
\begin{align*}
    \frac{\dd}{\dd t}u(t) = -A_m^{(D)} u(t), \quad A_m^{(D)} := \frac{\alpha}{h^2}T_m^{(D)},
\end{align*}
where $u(t)=(u_1(t),\dots,u_M(t))^\top$ and the Dirichlet Laplacian is
\begin{align*}
    T_m^{(D)} = 2I^{\otimes m} -  \sum_{x=0}^{M-2}\left(\ketbra{x}{x+1}+\ketbra{x+1}{x}\right).
\end{align*}

We order the qubits from least to most significant as $q_0,\dots,q_{m-1}$. For $m=2$,
\begin{align*}
    T_2^{(D)} :=
    \begin{bmatrix}
        2 & -1 & 0 & 0 \\
        -1 & 2 & -1 & 0 \\
        0 & -1 & 2 & -1 \\
        0 & 0 & -1 & 2
    \end{bmatrix}  = 2 I_1 I_0 - I_1 X_0 - \frac{1}{2}\left(X_1X_0 + Y_1Y_0\right).
\end{align*}

Throughout this and the following compilation subsections, \(\hat x\) denotes the implemented truncated quadrature \(\hat x_{N_{\rm Fock}}\) of Section~\ref{sec:trotter}, with the subscript suppressed for readability.
Thus, one hybrid CV--DV product-formula step takes the form
\begin{align*}
    \mathcal U_m^{(D)}(\Delta t) &:=
    \exp\!\left(-i\theta\,\hat x \otimes T_m^{(D)}\right),
    &&\theta := \frac{\alpha\,\Delta t}{h^2},
\end{align*}
where $\Delta t=t/n_t$ for $n_t$ product-formula steps over total runtime $t$.

Let \( \sigma_{\pm}^{(j)}  \)
denote the ladder operator \(\sigma_{\pm}=(X\pm iY)/2\) acting on qubit \(j\) and as the identity on all other qubits. The parenthesized superscript \((j)\) is a qubit label, whereas the subscript \(\pm\) identifies the raising or lowering operator. We define
\begin{align}
    R_c^{(m)} &:= \sigma_-^{(c)} \prod_{j=0}^{c-1}\sigma_+^{(j)}
    + \sigma_+^{(c)} \prod_{j=0}^{c-1}\sigma_-^{(j)},
    \quad c=0,\ldots,m-1, \label{eq:Rc-definition}
\end{align}
where the empty product for \(c=0\) is understood as the identity. Thus,
\begin{align*}
    T_m^{(D)} = 2I^{\otimes m} - \sum_{c=0}^{m-1}R_c^{(m)}.
\end{align*}
The first few terms are $R_0^{(m)} = X_0 \text{ and } R_1^{(m)} = \frac{1}{2}\left(X_1X_0 + Y_1Y_0\right)$, so that $T_2^{(D)} = 2 I_0 I_1 -  R_0^{(2)} - R_1^{(2)}$. Using a first-order product formula, one Trotter step is compiled as
\begin{align*}
    \mathcal U_m^{(D)}(\Delta t) \approx e^{-i2\theta\,\hat x \otimes I^{\otimes m}} \prod_{c=0}^{m-1} e^{i\theta\,\hat x \otimes R_c^{(m)}}.
\end{align*}
Since $A_m^{(D)}$ is real symmetric, the only product-formula error comes from the inner compilation of $T_m^{(D)}$ into the blocks $\{2I^{\otimes m}, R_0^{(m)},\dots,R_{m-1}^{(m)}\}$. The identity block commutes with every $R_c^{(m)}$, so only distinct pairs $R_c^{(m)}$ and $R_{c'}^{(m)}$ contribute:
\begin{align*}
    [\hat{x}\otimes R_c^{(m)},\;\hat{x}\otimes R_{c'}^{(m)}]=\hat{x}^2\otimes[R_c^{(m)}, R_{c'}^{(m)}].
\end{align*}
Because $\|R_c^{(m)}\|=1$, each commutator obeys $\|[R_c^{(m)}, R_{c'}^{(m)}]\|\leq 2$, and there are $\binom{m}{2}=\mathcal{O}(\log^2 M)$ such pairs. Accumulating over $n_t$ product-formula steps with $\Delta t=t/n_t$ therefore gives
\begin{equation}
    \epsilon_{t,{\rm heat}}=\mathcal{O}\!\left(\frac{t^2}{n_t}\,\frac{\alpha^2}{h^4}\,N_{\rm Fock}\,\log^2 M\right).
    \label{eq:heat-trotter-error}
\end{equation}
The step count $n_t$ is then chosen so that $\epsilon_{t,{\rm heat}}$ stays below the desired tolerance.
Each $R_c^{(m)}$ contains $2^c$ Pauli strings with coefficients $\pm 2^{-c}$,
\begin{align*}
    R_c^{(m)} = 2^{-c} \sum_{\ell=1}^{2^c} s_{c,\ell}P_{c,\ell}^{(m)}
    \qquad s_{c,\ell}\in\{\pm 1\},
\end{align*}
where each $P_{c,\ell}^{(m)}$ has support on $q_0,\dots,q_c$ and an even number of $Y$ operators. Hence any two $P_{c,\ell}^{(m)}$ differ on an even number of qubits and commute, so
\begin{align*}
    e^{i\theta\,\hat x\otimes R_c^{(m)}} = \prod_{\ell=1}^{2^c}
    e^{i\theta 2^{-c}s_{c,\ell}\,\hat x\otimes P_{c,\ell}^{(m)}}.
\end{align*}
Each factor $e^{i\theta 2^{-c}s_{c,\ell}\,\hat x\otimes P_{c,\ell}^{(m)}}$ is implemented by basis changes, a parity CX ladder on the support of $P_{c,\ell}^{(m)}$, one hybrid interaction, and uncomputation. A representative hybrid Pauli block is shown in Fig.~\ref{fig:heat_h2_block}.

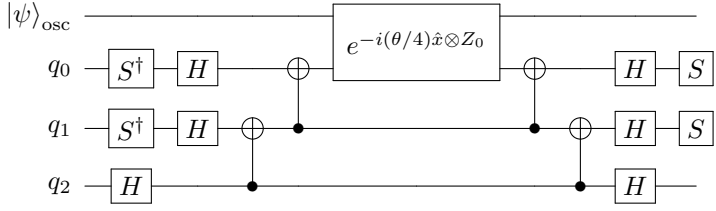
\begin{figure}[h]
    \centering
    \begin{align*}
    \Qcircuit @C=0.9em @R=0.8em {
        \lstick{\ket{\psi}_{\rm osc}} & \qw & \qw & \qw & \qw & \multigate{1}{e^{-i(\theta/4)\hat{x}\otimes Z_0}} & \qw & \qw & \qw & \qw \\
        \lstick{q_0} & \gate{S^\dagger} & \gate{H} & \qw & \targ & \ghost{e^{-i(\theta/4)\hat{x}\otimes Z_0}} & \targ & \qw & \gate{H} & \gate{S} \\
        \lstick{q_1} & \gate{S^\dagger} & \gate{H} & \targ & \ctrl{-1} & \qw & \ctrl{-1} & \targ & \gate{H} & \gate{S} \\
        \lstick{q_2} & \gate{H} & \qw & \ctrl{-1} & \qw & \qw & \qw & \ctrl{-1} & \gate{H} & \qw
    }
    \end{align*}
    \caption{Compilation of the hybrid Pauli factor $e^{-i(\theta/4)\hat{x}\otimes X_2Y_1Y_0}$, the member of $R_2^{(m)}$ that enters with sign $s_{2,\ell}=-1$. Basis changes map $X_2Y_1Y_0$ to $Z_2Z_1Z_0$, a CX ladder collects the parity on $q_0$, the oscillator couples through $e^{-i(\theta/4)\hat{x}\otimes Z_0}$, and the circuit is then uncomputed. The block $e^{i\theta \hat{x}\otimes R_2^{(m)}}$ is obtained by concatenating the four commuting Pauli factors in $R_2^{(m)}$.}
    \label{fig:heat_h2_block}
\end{figure}

\subsubsection{Other One-dimensional Boundary Conditions}
\label{1d-heat-bc}

Both remaining boundary conditions are obtained from the Dirichlet block by adding correction terms. For periodic boundary conditions, we have $T_m^{(P)} = T_m^{(D)} - T_m^{{\rm corr}, (P)}$, where $T_m^{{\rm corr}, (P)} := \ketbra{0}{M-1} + \ketbra{M-1}{0}$.
The Pauli decomposition of the correction term is
\begin{align*}
    T_m^{{\rm corr}, (P)} 
    &= \ketbra{0^m}{1^m} + \ketbra{1^m}{0^m} 
    = \prod_{j=0}^{m-1}\sigma^{(j)}_+ + \prod_{j=0}^{m-1}\sigma^{(j)}_- \\
    &= 2^{1-m}
    \sum_{\substack{S\subseteq\{0,\dots,m-1\}\\ |S|\ {\rm even}}}
    (-1)^{|S|/2} \left(\prod_{j\notin S}X_j\right)\left(\prod_{j\in S}Y_j\right).
\end{align*}
Thus $T_m^{{\rm corr}, (P)}$ contributes $2^{m-1}$ commuting Pauli strings, each of weight $m$.

For homogeneous Neumann boundary conditions, we have $T_m^{(N)} = T_m^{(D)} - T_m^{{\rm corr}, (N)}$, where $T_m^{{\rm corr}, (N)} := \ketbra{0}{0} + \ketbra{M-1}{M-1}$.
As in the periodic case,
\begin{align*}
    T_m^{{\rm corr}, (N)} 
    = \ketbra{0^m}{0^m} + \ketbra{1^m}{1^m} 
    = 2^{1-m} \sum_{\substack{S\subseteq\{0,\dots,m-1\}\\ |S|\ even}} \left(\prod_{j\in S}Z_j\right)\left(\prod_{j\notin S}I_j\right).
\end{align*}
Hence $T_m^{{\rm corr}, (N)}$ contributes $2^{m-1}$ commuting $Z$-type strings, one of which is the identity. In both cases one reuses $\mathcal U_m^{(bc)}(\Delta t):= \exp\!\left(-i\theta\,\hat x \otimes T_m^{(bc)}\right)$.

\subsubsection{\texorpdfstring{$d$}{d}-dimensional Extension}

Let $\Omega=\prod_{i=1}^d(0,L_i)$ and consider the $d$-dimensional heat equation
\begin{align*}
    \frac{\partial u(\mathbf{x},t)}{\partial t}
    &= \alpha \sum_{i=1}^d \frac{\partial^2 u(\mathbf{x},t)}{\partial x_i^2},
    && \mathbf{x}\in\Omega,
\end{align*}
with separable homogeneous boundary conditions $bc_i\in\{D,P,N\}$ along each axis. Let the $i^{\rm th}$ coordinate be discretized with $M_i=2^{m_i}$ grid points and spacing $h_i$, and denote by $I_{m_i}=I^{\otimes m_i}$ the identity on the corresponding coordinate subregister. The semidiscrete operator is the Kronecker sum
\begin{align*}
    A_d^{(\mathbf{bc})}
    = \alpha \sum_{i=1}^d \frac{1}{h_i^2}
    \left(I_{m_1}\otimes \cdots \otimes I_{m_{i-1}} \otimes T_{m_i}^{(bc_i)} \otimes I_{m_{i+1}} \otimes \cdots \otimes I_{m_d} \right),
\end{align*}
where $\mathbf{bc}=(bc_1,\dots,bc_d)$. The coordinate blocks commute because they act on disjoint qubit subregisters. Hence
\begin{align*}
    \mathcal U_d^{(\mathbf{bc})}(\Delta t)
    &:= \exp\!\left(-i\Delta t\,\hat x \otimes A_d^{(\mathbf{bc})}\right) \nonumber \\
    &= \prod_{i=1}^d
    \exp\!\left[
        -i\theta_i\,\hat x \otimes
        \left(
            I_{m_1}\otimes \cdots \otimes I_{m_{i-1}}
            \otimes T_{m_i}^{(bc_i)}
            \otimes I_{m_{i+1}} \otimes \cdots \otimes I_{m_d}
        \right)
    \right],
\end{align*}
with $\theta_i=\alpha\Delta t/h_i^2$. Each factor is compiled by the corresponding one-dimensional block.

We summarize exact gate counts of Trotter blocks in Table~\ref{tab:heat-gate-scaling}, with details provided in Appendix~\ref{ss:heat-count}.
\begin{table}[h]
    \centering
    \caption{Gate counts for the $d$-dimensional heat-equation compilation after $n_t$ first-order product-formula steps, assuming the same boundary condition on each axis. Here $M_i=2^{m_i}$ is the grid size along axis $i$. Details are provided in Appendix~\ref{ss:heat-count}.}
    \label{tab:heat-gate-scaling}{\small
    \setlength{\tabcolsep}{8pt}
    \begin{tabular}{@{}lccc@{}}
        \toprule
        Boundary & Hybrid gates & CX gates & 1-qubit gates \\
        \midrule
        Dirichlet &
        $n_t\sum\limits_{i=1}^d M_i$ &
        $n_t\sum\limits_{i=1}^d \left[2(m_i-2)M_i+4\right]$ &
        $n_t\sum\limits_{i=1}^d \left[3(m_i-1)M_i+2\right]$
        \\[0.6em]
        \midrule
        Periodic &
        $n_t\sum\limits_{i=1}^d \frac{3}{4}M_i$ &
        $n_t\sum\limits_{i=1}^d
        \left[\left(\frac{3}{2}m_i-\frac{7}{2}\right)M_i+4\right]$ &
        $n_t\sum\limits_{i=1}^d \left[
        \begin{cases}
            6, & m_i=2,\\
            2+\left(\frac{9}{4}m_i-\frac{13}{4}\right)M_i,
            & m_i\geq 3
        \end{cases}\right]$
        \\[0.6em]
        \midrule
        Neumann &
        $n_t\sum\limits_{i=1}^d \left(\frac{3}{2}M_i-1\right)$ &
        $n_t\sum\limits_{i=1}^d
        \left[\left(\frac{5}{2}m_i-5\right)M_i+6\right]$ &
        $n_t\sum\limits_{i=1}^d \left[3(m_i-1)M_i+2\right]$
        \\
        \bottomrule
    \end{tabular}}
\end{table}

\section{Circuit-Level Realization}
\label{sec:simulation}

In this section we describe the circuit-level realization used in our numerical study of the time-independent homogeneous CV--DV LCHS map introduced in Section~\ref{ss:lchs}.
Writing \(A=L+iH\) with \(L\succeq0\), the implemented unscaled postselected map at runtime \(T\) is \(\widetilde{\mathcal K}_{N,N_{\rm Fock},n_t}(T)\), defined in Eq.~\eqref{eq:evol-op}, where \(U_{n_t}(T)\) is the product-formula approximation discussed in Section~\ref{sec:trotter}.

\subsection{CV State Preparation}
\label{ss:cv-prep}

The truncated squeezed-Fock expansion coefficients $\{C_n\}_{n=0}^{N-1}$ are obtained by evaluating the integral in Eq.~\eqref{eq:coefficients} with adaptive numerical quadrature over a finite symmetric interval chosen so that the Gaussian-damped tails are below the numerical tolerance. Under the $\hbar=2$ convention adopted throughout this paper, the Gaussian damping factor in that integral is $\gamma=\frac14(e^{-2r'}-e^{-2r})$ with squeezing widths $\sigma=e^r$ and $\sigma'=e^{r'}$.

These coefficients define the normalized unsqueezed finite seed state $\ket{\chi_N}$ in Eq.~\eqref{eq:fock}. Both synthesis routes below target this seed state. Applying the outer squeezing operation $\mathsf S_{N_{\rm Fock}}(r')$ and normalizing then produces the final finite kernel state $\ket{\psi_N^{(N_{\rm Fock})}}$ defined in Eq.~\eqref{eq:finite-target-states}.

\paragraph{Law--Eberly protocol.}
This scheme uses the deterministic oscillator-state synthesis procedure of Law and Eberly~\cite{law1996arbitrary}, later implemented in a superconducting resonator by Hofheinz \emph{et al.}~\cite{hofheinz2009synthesizing}.  The target of this step is the state $\ket{\chi_N}$ prepared jointly with a two-level system in its ground state. The LE construction alternates two qubit--oscillator primitives. The first is a Jaynes--Cummings exchange pulse on the adjacent manifold $\{\ket{e,n-1},\ket{g,n}\}$,
\begin{align}
    S_n(\alpha,\phi) = \exp\!\left[ -i\frac{\alpha}{\sqrt n} \left( e^{i\phi}\sigma_- a^\dagger +e^{-i\phi}\sigma_+ a \right) \right], \qquad n\geq 1,
\label{eq:le-jc-pulse}
\end{align}
where the $1/\sqrt n$ factor compensates the oscillator matrix element. In the present Bosonic Qiskit (version 15.1) circuit implementation, this pulse is emitted as a $\code{cv\_jc}$ gate. The second is the qubit rotation,
\begin{align}
    R_n(\theta,\varphi) = \mathbb I_{\rm osc}\otimes \exp\!\left[ -\frac{i\theta}{2} \left(\cos\varphi\,X+\sin\varphi\,Y\right) \right],
\label{eq:le-qubit-rotation}
\end{align}
which is a standard Qiskit $\code{r}$ gate on the LE ancillary qubit.

The synthesis solves the time-reversed problem as suggested in~\cite{law1996arbitrary}. Starting from $\ket{g}\ket{\chi_N}$, alternating JC pulses and qubit rotations eliminate the amplitudes from the highest occupied Fock level downward until the state reaches $\ket{g,0}$ up to an irrelevant global phase, and the preparation circuit is the adjoint of this unpreparation sequence in reverse order, as summarized in Algorithm~\ref{alg:le-finite-fock}.

\begin{algorithm}[h]
\caption{Law--Eberly synthesis used for the CV state preparation oracle}
\label{alg:le-finite-fock}
\begin{algorithmic}[1]
\Require Normalized coefficients $\{C_n\}_{n=0}^{N-1}$
\Ensure Gate list preparing $\ket{\chi_N}$
\State Set the symbolic reverse state $\ket{\Psi_{\rm rev}}\gets \ket{g}\sum_{n=0}^{N-1}C_n\ket{n}$
\State Set $\mathcal P_{\rm unprep}\gets[\,]$
\For{$n=N-1,N-2,\ldots,1$}
    \State Set $a_n=\langle e,n-1|\Psi_{\rm rev}\rangle$ and $b_n=\langle g,n|\Psi_{\rm rev}\rangle$
    \If{$|b_n|>0$}
        \Comment{Choose $S_n$ so that the updated amplitude on $\ket{g,n}$ is zero.}
        \If{$|a_n|=0$}
            \State Set $\alpha_n=\pi/2$ and $\phi_n=0$
        \Else
            \State Set $\alpha_n=\arctan(|b_n|/|a_n|)$ and $\phi_n=\arg b_n-\arg a_n-\pi/2$
        \EndIf
        \State $\ket{\Psi_{\rm rev}}\gets S_n(\alpha_n,\phi_n)\ket{\Psi_{\rm rev}}$
        \State Append $S_n(\alpha_n,\phi_n)$ to $\mathcal P_{\rm unprep}$
    \EndIf
    \State Set $u_{n-1}=\langle e,n-1|\Psi_{\rm rev}\rangle$ and $v_{n-1}=\langle g,n-1|\Psi_{\rm rev}\rangle$
    \If{$|u_{n-1}|>0$}
        \Comment{Choose $R_{n-1}$ so that the updated amplitude on $\ket{e,n-1}$ is zero.}
        \If{$|v_{n-1}|=0$}
            \State Set $\theta_{n-1}=\pi$ and $\varphi_{n-1}=0$
        \Else
            \State Set $\theta_{n-1}=2\arctan(|u_{n-1}|/|v_{n-1}|)$ and $\varphi_{n-1}=-\arg u_{n-1}+\arg v_{n-1}+\pi/2$
        \EndIf
        \State $\ket{\Psi_{\rm rev}}\gets R_{n-1}(\theta_{n-1},\varphi_{n-1})\ket{\Psi_{\rm rev}}$
        \State Append $R_{n-1}(\theta_{n-1},\varphi_{n-1})$ to $\mathcal P_{\rm unprep}$
    \EndIf
\EndFor
\State Return $\mathcal P_{\rm prep}=(\mathcal P_{\rm unprep})^\dagger$ in reverse order
\end{algorithmic}
\end{algorithm}

The circuit-level realization of the returned preparation sequence is shown in Fig.~\ref{fig:le_state_prep_circuit}. For a target with support on all $N$ levels, this produces $N-1$ JC pulses and $N-1$ qubit rotations. Thus, with \(N=32\) used in Section~\ref{sec:clean_experiments}, the ideal Law--Eberly realization reports \(31\) JC pulses and \(31\) qubit rotations.

\begin{figure}[h]
\centering
\begin{align*}
\Qcircuit @C=1.05em @R=1.0em {
\lstick{\ket{0}_{\text{osc}}}
  & \qw
  & \qw
  & \multigate{1}{S_{1}}
  & \qw
  & \multigate{1}{S_{2}}
  & \qw
  & \cdots
  &
  & \qw
  & \multigate{1}{S_{N-1}}
  & \qw
  & \rstick{\ket{\chi_N}} \\
\lstick{\ket{g}_{\rm LE}}
  & \qw
  & \gate{R_{0}}
  & \ghost{S_{1}}
  & \gate{R_{1}}
  & \ghost{S_{2}}
  & \qw
  & \cdots
  &
  & \gate{R_{N-2}}
  & \ghost{S_{N-1}}
  & \qw
  & \rstick{\ket{g}_{\rm LE}}
}
\end{align*}
\caption{Circuit-level Law--Eberly state preparation of the finite oscillator seed state. Each $R_n$ block is a standard \texttt{r} gate in Qiskit and each $S_n$ block is emitted as a \texttt{cv\_jc} gate.}
\label{fig:le_state_prep_circuit}
\end{figure}
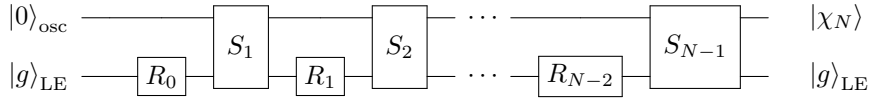

\paragraph{Variational SNAP-plus-displacement (SNAP+$\Dcal$) synthesis.}
Another implementation uses universal oscillator control generated by selective number-dependent arbitrary phase (SNAP) gates together with displacements~\cite{heeres2015cavity, Krastanov2015}. In this picture, we seek a variational approximation of the form 
\begin{align*}
    U_{\rm SNAP+\Dcal} = \prod_{\ell=1}^{N_{\rm layer}} \left[ \Dcal(\alpha_\ell)\, \mathrm{SNAP}(\bm{\theta}_\ell) \right],
\end{align*}
where $N_{\rm layer}$ is a user-defined parameter for the number of $\mathrm{SNAP}+\Dcal$ layers and $\mathrm{SNAP}(\bm{\theta}_\ell) = \sum_{n=0}^{N-1} e^{i\theta_{\ell,n}}\ketbra{n}{n}$.
The variational parameters are optimized to maximize the state-preparation fidelity,
\begin{align*}
    \max_{\{\alpha_\ell,\bm{\theta}_\ell\}}
    \left|\bra{\chi_N}
    U_{\rm SNAP+\Dcal}\ket{0}\right|^2.
\end{align*}
This route is implemented directly at the circuit level in Bosonic Qiskit~\cite{bosqis, bosonic-qiskit-repo} using $\code{cv\_snap}$ and $\code{cv\_d}$ gates.
Previous numerical studies report high-fidelity oscillator-state preparation with relatively short optimized SNAP+$\Dcal$ sequences~\cite{fosel2020efficient, Kudra2022}. A separate study applies the ansatz to non-trivial electronic ground states from quantum chemistry~\cite{Dutta2024EST}.

\subsection{Gate-level Compilation of Hybrid Evolution}
\label{sec:impl-trotter}
Building on Section~\ref{sec:trotter}, the compiled evolution consists of hybrid factors \(e^{-i\Delta t\,\alpha_i\hat{x}\otimes P_i}\) and qubit-only factors \(e^{-i\Delta t\,\beta_jQ_j}\).
The latter are standard DV Pauli rotations, so the only new CV--DV ingredient is the compilation of a single hybrid factor \(e^{-i\Delta t\,\alpha_i\hat{x}\otimes P_i}\).
Let \(V_i\) denote the single-qubit basis changes for \(P_i\), and let \(W_i\) denote a parity-mapping circuit with target qubit \(q_\star\), chosen such that
\begin{align*}
W_iV_iP_iV_i^\dagger W_i^\dagger=Z_{q_\star}.
\end{align*}
The hybrid factor then becomes
\begin{align*}
    e^{-i\Delta t\,\alpha_i\hat{x}\otimes P_i} = V_i^\dagger W_i^\dagger e^{-i\Delta t\,\alpha_i\hat{x}\otimes Z_{q_\star}} W_iV_i.
\end{align*}
The primitive hybrid interaction admits the conditional-displacement form
\begin{align*}
    c\Dcal(-i\Delta t\,\alpha_i)
    &= e^{-i\Delta t\,\alpha_i\hat{x}\otimes Z_{q_\star}}\\
    &= \Dcal(-i\Delta t\,\alpha_i)\otimes\ket{0}\!\bra{0}_{q_\star} + \Dcal(i\Delta t\,\alpha_i)\otimes\ket{1}\!\bra{1}_{q_\star},
\end{align*}
where \(\Dcal(\zeta)=e^{\zeta a^\dagger-\zeta^*a}\) and \(\hat{x}=a+a^\dagger\) under the \(\hbar=2\) convention.

\section{Numerical Experiments}
\label{sec:clean_experiments}

The primary benchmark in this section is the one-dimensional heat equation with Dirichlet, Neumann, and periodic boundary conditions on a two-qubit register. The later subsections broaden the study to larger one-dimensional grids, a two-dimensional heat stress case, and non-normal advection--diffusion generators, together with a scaling study, a photon-loss sensitivity study, and a comparison with DV LCHS.

The implementations in this section are based on Bosonic Qiskit (version 15.1)~\cite{bosqis, bosonic-qiskit-repo}.
Unless stated otherwise, the target is the exact discretized solution $u(T)=e^{-AT}u(0)$ at $T=1$, with oscillator truncation $N_{\rm Fock}=64$ and first-order Trotterization with $n_t=100$ steps. The starting benchmark instance is a four-dimensional semidiscrete system with $M=4$, $\alpha=1$, $h=1$, $T=1$, and initial state $u(0)=\ket{01}$. Since $\alpha/h^2=1$ in all three cases, the semidiscrete generators are
\begin{align*}
    T^{(D)}_2&= 2 I_1I_0 - I_1X_0 - \frac{1}{2}(X_1X_0 + Y_1Y_0), \\
    T^{(N)}_2&= \frac{3}{2} I_1I_0 - I_1X_0 - \frac{1}{2}(X_1X_0 + Y_1Y_0 + Z_1Z_0), \\
    T^{(P)}_2&= 2 I_1I_0 - I_1X_0 - X_1X_0 .
\end{align*}
For normalized states, fidelities are reported as
\begin{align*}
F(\psi,\phi)=|\langle \psi | \phi\rangle|^2
\end{align*}
and the exact target solution vector is $u_{\rm exact}(T)=e^{-AT}u(0)$.
We use
\begin{align*}
    F_{\rm LE}=F(u_{\rm LE},u_{\rm exact}),
\end{align*}
where $u_{\rm LE}$ is the postselected system-register solution vector produced by the CV--DV LCHS circuit after postselection using the circuit-level LE route. The LE route deterministically prepares the target truncated oscillator state in the ideal qubit--oscillator circuit model, while SNAP+$\Dcal$ gives a variational approximation. Thus, $F_{\rm LE}$ reports the corresponding end-to-end PDE fidelities.
The cutoff $N$ denotes the number of squeezed-Fock coefficients in the truncated CV oracle state and is distinct from the ordinary-Fock simulation dimension $N_{\rm Fock}$.

Beyond these conditional normalized fidelities, the primary error metric in this section is the fixed-target-scale map error
\begin{align*}
    \varepsilon_F = \frac{\left\|\alpha_{N,r}\widetilde{\mathcal K}_{N,N_{\rm Fock},n_t}(T)-e^{-AT}\right\|_F}{\left\|e^{-AT}\right\|_F}.
\end{align*}
Here, \(\widetilde{\mathcal K}_{N,N_{\rm Fock},n_t}\) is the unscaled physical postselected map, and \(\alpha_{N,r}=\braket{\phi_r|\psi_N}^{-1}\) is fixed by the normalized target preparation and postselection states before evaluating any system input. The evaluation applies this single scale uniformly across all inputs and columns, so \(\varepsilon_F\) includes amplitude and global-phase errors relative to the fixed physical convention. Conditional fidelities omit these two error components and serve as a secondary diagnostic throughout this section.

For an input \(u_0\), the scaled relative error and the conditional fidelity of the postselected circuit-level map are
\begin{align*}
\varepsilon_{\rm rel}(u_0) &= \frac{\left\|\alpha_{N,r}\widetilde{\mathcal K}_{N,N_{\rm Fock},n_t}(T)u_0-e^{-AT}u_0\right\|_2}{\left\|e^{-AT}u_0\right\|_2}, \\
F_{\rm cond}(u_0) &= F\!\left(u_{N,N_{\rm Fock},n_t}(T),\ \frac{e^{-AT}u_0}{\|e^{-AT}u_0\|_2}\right),
\end{align*}
where \(u_{N,N_{\rm Fock},n_t}(T)\) is the normalized conditional output state of Section~\ref{ss:postselection}, so for the circuit-level LE route \(F_{\rm LE}\) is \(F_{\rm cond}\) evaluated at the benchmark input. For the benchmark and breadth finite-map experiments, and for the circuit-level scaling rows where applicable, we record the worst values of \(\varepsilon_{\rm rel}\) and \(1-F_{\rm cond}\) over the computational-basis inputs, the fixed diagnostic input when applicable, and the physical postselection probability \(p_{\rm succ}\) with its reciprocal \(1/p_{\rm succ}\). The prepared oracle state is fixed to the gauge \(\braket{\psi_N|\widetilde\psi_N}>0\), applied once at classical readout.

We also provide the measured components of the two-cutoff error decomposition: \(\epsilon_{\rm model}\) for the finite kernel encoding under ideal state injection, \(\epsilon_{\rm synth}\) for oscillator-state preparation, and \(\epsilon_t\) for product-formula discretization. Selected values are reported in the text and tables to identify the dominant error source.

\subsection{Parameter Selection}

Kernel parameters are selected by a coarse grid search at the exact-truncated-map level, so no circuits enter the selection, with $r$ from $0.8$ to $2.2$, $r'$ from $0.1$ to $1.0$, $\beta$ from $0.2$ to $0.9$, and $N\in\{16,24,32\}$, all at $N_{\rm Fock}=64$ and $T=1$. Fig.~\ref{fig:clean_sensitivity_exact} shows the resulting landscape of the fixed-scale error $\varepsilon_F$. From this grid we select one kernel parameter set for the entire paper,
\begin{align*}
    \theta^\star:=(r,\,r',\,\beta,\,N)=(1.6,\ 0.25,\ 0.5,\ 32),
\end{align*}
used for all three boundary conditions and all subsequent experiments, except that the scaling sweeps keep $(r,r',\beta)$ fixed and vary $(N,N_{\rm Fock})$ explicitly and the SNAP+$\Dcal$ state-preparation study also evaluates $N=16$ at the same $(r,r',\beta)$.
Appendix~\ref{app:ha-retune} additionally performs a common local \((r,r')\) scan for both kernels and evaluates the Huang--An kernel at the grid-selected point from that scan.
The selection criterion combines the worst-boundary map error with the benchmark postselection probability. Five grid points have worst-boundary errors below $0.7\%$. Among them, $\theta^\star$ ranks third in error and gives the largest worst-boundary benchmark success probability, $10.5\%$, compared with at most $8.9\%$ for the other four. Across the full grid, $\theta^\star$ is undominated in these two quantities. 

\begin{figure}[ht!]
\centering
\includegraphics[width=0.93\linewidth]{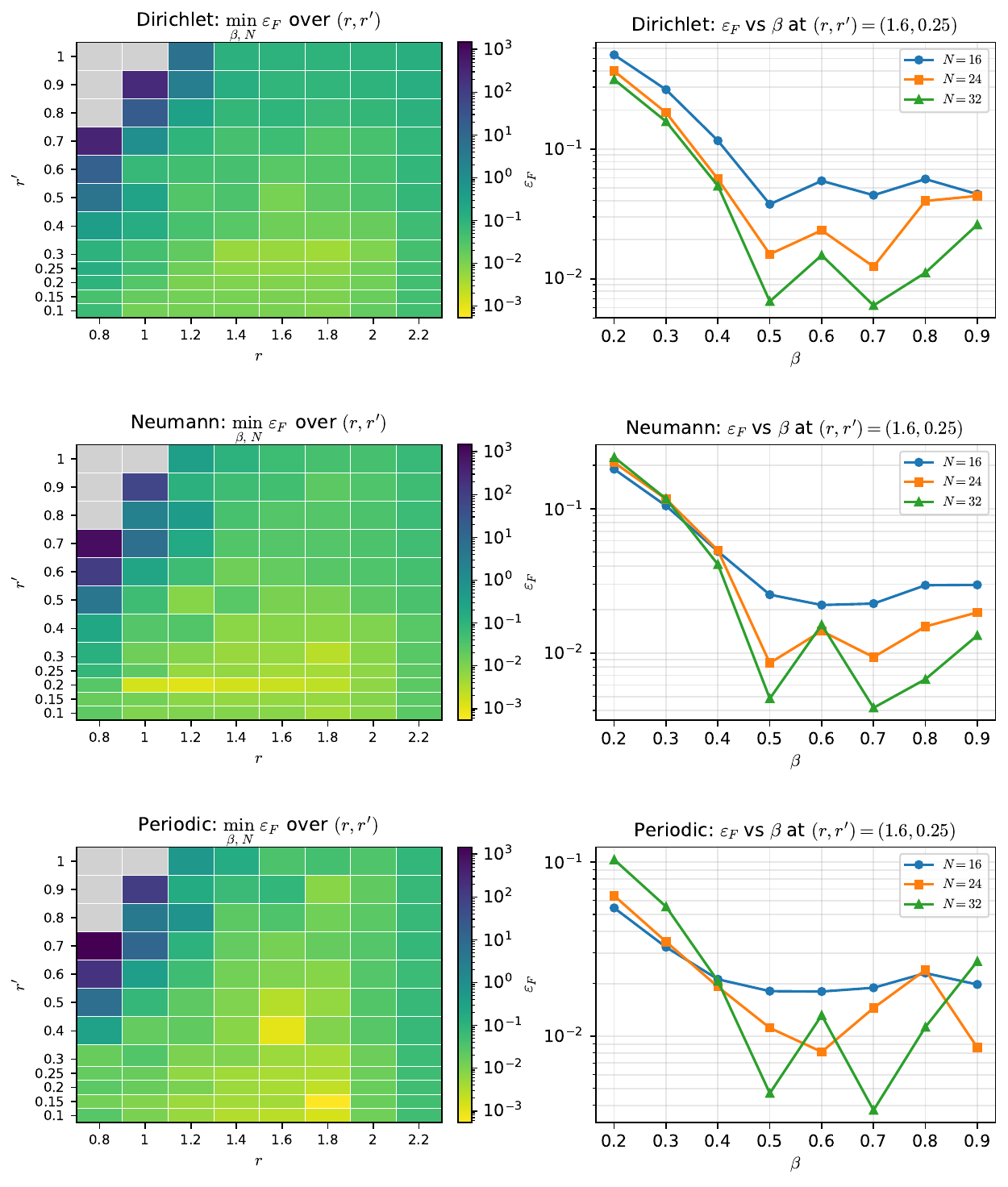}
\caption{Fixed-scale map error $\varepsilon_F$ over the coarse selection grid, evaluated on the exact truncated map at $N_{\rm Fock}=64$, $T=1$ under ideal CV state injection. Left column: profile of $\varepsilon_F$ over $(r,r')$, minimized over $(\beta,N)$. Grey cells are invalid ($r'\ge r$) or have numerically vanishing zeroth-moment overlap. Right column: $\varepsilon_F$ versus $\beta$ at $(r,r')=(1.6,0.25)$ for $N\in\{16,24,32\}$.}
\label{fig:clean_sensitivity_exact}
\end{figure}

At $\theta^\star$ the exact-map errors are $\varepsilon_F=0.67\%$, $0.48\%$, and $0.47\%$ for Dirichlet, Neumann, and periodic, respectively. The landscape shows a nonmonotonic dependence of $\varepsilon_F$ on $N$ at fixed squeezing parameters because the finite overlap scale $\alpha_{N,r}=\braket{\phi_r|\psi_N}^{-1}$ enters the fixed-scale error and the success probability jointly. This dependence requires selecting the cutoffs and squeezing parameters together.

The parameter set $\theta^\star$ uses $r=1.6$, corresponding to about $13.9$ dB of squeezing with mean photon number $\sinh^2r\approx5.6$, and the ordinary-Fock tails beyond the embedding cutoff $N_{\rm Fock}=64$ are small for every state involved. The largest is the tail norm $\tau_{\phi,N_{\rm Fock}}=3.59\times10^{-2}$ of the ideal postselection squeezed vacuum (tail probability $1.29\times10^{-3}$), which enters the $\epsilon_{\rm Fock}$ bound in Eq.~\eqref{eq:fock-embedding-bound}. The preparation squeezed vacuum at $r'=0.25$ has tail probability below $10^{-10}$. For the exactly squeezed finite kernel state, the $\tau_{\psi,N_{\rm Fock}}$ of the same bound equals $\|(\mathbb I-\Pi_{N_{\rm Fock}})\,S(r')\ket{\chi_N}\|_2=3.75\times10^{-4}$ at tail probability weight $1.41\times10^{-7}$. Consistently, the finite construction of Eq.~\eqref{eq:finite-target-states} matches the exactly squeezed finite kernel state to infidelity $1.70\times10^{-7}$. The coefficient-integration window was increased with Hermite order, and the reported quantities were stable under a further enlargement of the window and tighter quadrature tolerances.

For the variational SNAP+$\Dcal$ alternative we report the classical optimization cost at the selected $(r, r', \beta)$ values, for $N=16$ and $N=32$, using a fixed depth of $30$ ansatz layers and L-BFGS-B with $9$ random starts of $26.7$ minutes each on an Apple M3 Max chip. For $N=16$, a start first reaches a state-preparation infidelity below $10^{-6}$ after $14.9$ minutes, and the best observed infidelity over all starts is $2.78\times10^{-7}$. For $N=32$, the best infidelity reached within the $26.7$-minute starts is $3.03\times10^{-4}$, and continued optimization of the best start reaches $8.67\times10^{-7}$ after $3.4$ hours of optimization on that start. The optimization records along with seeds are provided in~\cite{cv-dv-repo}.

The diagnostic input, constructed orthogonal to the slowest-decaying eigenvector of \(L\), probes a fast-decaying direction beyond the computational basis. It indeed gives the largest scaled relative error \(\varepsilon_{\rm rel}\) in all three boundary cases, $1.17\%$, $1.13\%$, and $2.48\%$ for Dirichlet, Neumann, and periodic boundaries, with conditional infidelities \(1-F_{\rm cond}\) of at most $1.31\times10^{-4}$. Only for the periodic boundary does the benchmark input $u(0)=\ket{01}$ have the larger conditional infidelity, as Table~\ref{tab:clean_oracle_baseline} reports.

\begin{table}[ht!]
\centering 
\caption{Circuit-level LE benchmark results at $(r,r',\beta,N)=(1.6,0.25,0.5,32)$, $N_{\rm Fock}=64$, $n_t=100$, $T=1$. The shared kernel state has stellar rank $N-1=31$ and non-Gaussianity $\delta_{\rm nG}=1.33$. LE prepares it exactly in the ideal circuit model ($\epsilon_{\rm synth}\le3\times10^{-16}$). $\varepsilon_F$ is the fixed-scale map error of the postselected map.
$1-F_{\rm LE}$ is the conditional infidelity of the benchmark input $u(0)=\ket{01}$. Here, $p_{\rm succ}$ is the physical postselection probability of the benchmark input, and $1/p_{\rm succ}$ is the expected number of attempts.}
\label{tab:clean_oracle_baseline}
\begin{tabular}{lcccc}
\toprule
Boundary & $\varepsilon_F$ & $1-F_{\rm LE}$ & $p_{\rm succ}$ & $1/p_{\rm succ}$ \\
\midrule
Dirichlet & $0.70\%$ & $1.80\times10^{-5}$ & $10.47\%$ & $9.55$ \\
Neumann   & $0.56\%$ & $1.39\times10^{-5}$ & $16.48\%$ & $6.07$ \\
Periodic  & $0.47\%$ & $2.13\times10^{-5}$ & $15.36\%$ & $6.51$ \\
\bottomrule
\end{tabular}
\end{table}

The full-map spectral-norm decomposition identifies the finite kernel encoding as the dominant measured error source. Across the three boundary conditions, $\epsilon_{\rm model}$ ranges from $3.37\times10^{-3}$ to $3.78\times10^{-3}$ and $\epsilon_{\rm synth}\le3\times10^{-16}$. The product-formula errors are $\epsilon_t=8.90\times10^{-4}$ and $1.76\times10^{-3}$ for Dirichlet and Neumann and $1.48\times10^{-14}$ for periodic. The periodic Pauli terms commute, making the product formula exact up to numerical precision. At $n_t=100$, $\epsilon_{\rm model}$ exceeds $|\alpha_{N,r}|(\epsilon_{\rm synth}+\epsilon_t)$ in all three cases. These values lie more than an order of magnitude below the uniform \(L^1\) estimate of Lemma~\ref{lem:finite-r-map} evaluated in Section~\ref{ss:state-prep}. That estimate bounds the ideal-oscillator map \(\mathcal K_{N,r}\) for every admissible generator and evolution time. Its scope differs from the composite \(\epsilon_{\rm model}\), and it is accordingly loose for the specific benchmark generators. Increasing $n_t$ can reduce $\epsilon_t$ and may lower the total error further. The finite-model contribution remains.

For a variationally prepared normalized state $\widetilde\psi_N$ with preparation fidelity $F(\psi_N,\widetilde\psi_N)$, the phase-aligned statevector error is $\sqrt{2-2\sqrt{F(\psi_N,\widetilde\psi_N)}}$. Contractivity of the postselection block propagates this statevector error into a bound on $\epsilon_{\rm synth}$. The state infidelity $1-F(\psi_N,\widetilde\psi_N)$ and $\epsilon_{\rm synth}$ are distinct quantities. Figure~\ref{fig:clean_scaling} varies $\epsilon_t$ through $n_t$. Changes in $\epsilon_{\rm model}$ follow from the encoding parameters $(N,N_{\rm Fock},r,r',\beta)$.

Postselection probabilities across all evaluated inputs, including the diagnostic input, span $1.72\%$ to $10.47\%$, $10.64\%$ to $22.82\%$, and $0.70\%$ to $15.36\%$ for Dirichlet, Neumann, and periodic boundaries, respectively, reflecting the input dependence of the decayed-solution norm $\|e^{-At}\ket{u_0}\|_2^2$ that controls the success probability. At the benchmark point, the measured deviations \(|p_{\rm succ}^{(N,N_{\rm Fock},n_t)}-p_{\rm ref}^{(N,r)}|\) lie within the perturbation bound of Theorem~\ref{thm:postselection-success} for all three boundary conditions, and the largest deviation, \(1.14\times10^{-3}\) for the Dirichlet boundary, reaches \(62\%\) of its bound of \(1.84\times10^{-3}\).

\subsection{Resource Costs}

The $M=4$ benchmark uses $N_{\rm Fock}=64$, represented by six oscillator-emulation qubits in Bosonic Qiskit, together with two system qubits and one additional ancillary qubit for Law--Eberly synthesis. The breadth cases use three to five system qubits according to their dimension, while keeping $N_{\rm Fock}=64$. The scaling study explicitly varies $N_{\rm Fock}\in\{16,32,64\}$, corresponding to four, five, and six oscillator-emulation qubits.

Law--Eberly synthesis provides the circuit-level state-loading baseline, with the pulse counts of Table~\ref{tab:clean_resource_counts} identical for all three boundary conditions.

The benchmark cost divides into input loading, oscillator kernel-state preparation, hybrid evolution, and oscillator postselection. Table~\ref{tab:clean_resource_counts} displays the two nontrivial compiled blocks, Table~\ref{tab:clean_oracle_baseline} reports the postselection probability and direct-repetition overhead, input loading costs one $X$ gate for the common input $\ket{01}$, and the slow-mode-orthogonal diagnostic uses simulator initialization and is excluded from hardware gate counts.

The periodic coefficient matrix $A^{(P)}$ contains only two non-identity Pauli terms and no Pauli-$Y$ term, whose basis changes would require both $R_Z$ and $H$ gates. This yields the lower Trotter-block counts of Table~\ref{tab:clean_resource_counts}, although $A^{(P)}$ has more non-zero matrix entries than $A^{(D)}$.

\begin{table}[h]
\centering 
\caption{Resource counts for the circuit-level LE route at $\theta^\star$. The Trotter columns refer only to the 100-step hybrid evolution block and aggregate single-qubit DV gates, CXs, unconditional CV displacements ($\Dcal$), and controlled displacements ($c\Dcal$). The two outer squeezing gates are excluded. LE synthesis uses one ancillary qubit.}
\label{tab:clean_resource_counts}
\begin{tabular}{lccc cccc}
\toprule
& \multicolumn{3}{c}{Law--Eberly} & \multicolumn{4}{c}{Trotter Block} \\
\cmidrule(lr){2-4}\cmidrule(lr){5-8}
Boundary & JC pulses & R gates & Ancillary qubit & 1Q gates & CXs & $\Dcal$ & $c\Dcal$ \\
\midrule
Dirichlet & 31 & 31 & 1 & 1400 & 400 & 100 & 300 \\
Neumann   & 31 & 31 & 1 & 1400 & 600 & 100 & 400 \\
Periodic  & 31 & 31 & 1 & 600  & 200 & 100 & 200 \\
\bottomrule
\end{tabular}
\end{table}

\subsection{Breadth and Scaling Experiments}

To probe behavior beyond the $M=4$ benchmark, we evaluate seven additional instances at the same parameter set $\theta^\star$ (Table~\ref{tab:clean_breadth}). All instances are evaluated with full circuit-level maps using Law--Eberly preparation, so every row reports the fixed-scale map error, worst-input conditional infidelity \(1-F_{\rm cond}\), and postselection-probability range from circuit maps over all \(D\) computational-basis inputs.

The one-dimensional heat instances with $M=8$, $M=16$, and $M=32$ and Dirichlet boundaries extend the benchmark family in system dimension at fixed grid spacing \(h=1\), so the physical domain grows with \(M\).
The advection--diffusion (AD) instances discretize \(u_t=\nu u_{xx}-c\,u_x\) at \(M\) interior grid points with homogeneous Dirichlet boundaries and centered finite differences. Writing \(\dot{\bm u}=-A_M\bm u\), the resulting coefficient matrix is
\begin{align}
A_M &= \operatorname{tridiag}\!\left( -\frac{\nu}{h^2}-\frac{c}{2h}, \frac{2\nu}{h^2}, -\frac{\nu}{h^2}+\frac{c}{2h} \right).
\label{eq:advection-diffusion-matrix}
\end{align}
Here \(\operatorname{tridiag}(\ell,d,u)\) denotes a matrix with lower, diagonal, and upper entries \(\ell,d,u\), respectively. For the parameters used in all AD experiments, \(\nu=c=h=1\) and hence
\begin{align*}
A_M=\operatorname{tridiag}\!\left(-\frac32,2,-\frac12\right), \qquad M\in\{8,16,32\}.
\end{align*}
Its Cartesian decomposition is
\begin{align*}
L_M=\frac{A_M+A_M^\dagger}{2} =\operatorname{tridiag}(-1,2,-1)\succ0,
\qquad
H_M=\frac{A_M-A_M^\dagger}{2i} =-\frac{i}{2}\operatorname{tridiag}(-1,0,1).
\end{align*}
Since \([L_M,H_M]\neq0\), \(A_M=L_M+iH_M\) is non-normal. The normalized Henrici departure \(\sqrt{\|A_M\|_F^2-\sum_j|\lambda_j(A_M)|^2}/\|A_M\|_F\) is \(0.376\), \(0.384\), and \(0.388\) at \(M=8\), \(16\), and \(32\), respectively. The corresponding mesh P\'eclet number is \(ch/(2\nu)=0.5\). Across the same instances, $\kappa(V)$ grows by more than five orders of magnitude, from about $48.5$ at $M=8$ to $2.63\times10^7$ at $M=32$. The fixed-scale error stays within $1.15\%$--$1.48\%$, and the worst-input conditional infidelity \(1-F_{\rm cond}\) is at most $4.68\times10^{-4}$ (Table~\ref{tab:clean_breadth}).

In each of the three AD rows, $\epsilon_{\rm model}$ dominates $|\alpha_{N,r}|(\epsilon_{\rm synth}+\epsilon_t)$. From $M=8$ to $M=32$, the $\epsilon_t$ column of Table~\ref{tab:clean_breadth} rises from $1.83\times10^{-3}$ to $2.07\times10^{-3}$ as $\kappa(V)$ grows by more than five orders of magnitude. All three instances discretize the same PDE family. Simultaneous changes in dimension, spectral properties, and non-normality preclude a general conclusion about sensitivity to system size or non-normality.

The final case, the $4\times4$ two-dimensional heat instance with $D=16$ and Dirichlet boundaries on both axes, stresses the encoding at the shared parameter set. Its generator spans approximately twice the spectral range of the one-dimensional family, with $\|L\|_2=7.24$ compared with at most $3.99$ in the one-dimensional family. Its smallest eigenvalue is $0.764$, compared with $0.0341$ at $M=16$. At $T=1$, this stronger damping suppresses the postselection probability below $2\%$ and shrinks the reference norm $\|e^{-AT}\|_F$ from $1.74$ at the one-dimensional $M=16$ instance to $0.53$, and the fixed-scale map error grows to $7.40\%$ (Table~\ref{tab:clean_breadth}), of which a factor of $3.2$ comes from the smaller reference norm and a factor of $1.6$ from the growth of the absolute error $\|\alpha_{N,r}\widetilde{\mathcal K}_{N,N_{\rm Fock},n_t}(T)-e^{-AT}\|_F$ from $0.025$ to $0.040$. Each circuit column agrees with an independently constructed dense product-formula column to relative error at most $1.41\times10^{-13}$, and the measured product-formula error is $8.07\times10^{-4}$, below $10^{-3}$ and the smallest $\epsilon_t$ among the seven breadth instances. So, the growth comes from the finite kernel encoding, whose parameters are held fixed at $\theta^\star$ across the whole table.
Appendix~\ref{app:ha-retune-2d} scans \((r,r')\) directly on the two-dimensional generator to test sensitivity to the shared parameter set. At the exact-truncated-map level, the lowest fine-grid errors are \(1.537\%\) for \(g_\beta\) and \(0.669\%\) for the Huang--An kernel (Table~\ref{tab:ha-retune-2d-front}). The \(7.40\%\) error mainly reflects kernel parameters selected on the one-dimensional family. The lowest observed \(g_\beta\) error stays above the error cut of the one-dimensional rule, so the rule selects no two-dimensional operating point for \(g_\beta\). A fixed-error two-dimensional assessment requires kernel parameters matched to the wider spectral range of this generator.

\begin{table}[ht!]
\centering 
{\small
\caption{Breadth experiments using parameter set $\theta^\star$. All seven rows are evaluated with full circuit-level maps using Law--Eberly preparation, and \(\epsilon_t\) is the measured deviation between the circuit-level and exact finite-dimensional maps. The worst \(1-F_{\rm cond}\) column reports the largest conditional infidelity over the computational-basis inputs. Here, $\kappa(V)$ is the eigenvector condition number of the generator, and the heat generators are Hermitian, so $\kappa(V)=1$. The $p_{\rm succ}$ column is in percent and gives the benchmark-input (computational-basis index $1$) value with the range over all computational-basis inputs in parentheses. AD denotes advection--diffusion.}
\label{tab:clean_breadth}
\begin{tabular}{lcccccc}
\toprule
Instance & $D$ & $\kappa(V)$ & $\varepsilon_F$ & worst $1-F_{\rm cond}$ & $p_{\rm succ}$ [\%] & $\epsilon_t$ \\
\midrule
1D heat, $M=8$, Dirichlet               & 8  & $1$ & $1.04\%$ & $1.23\times10^{-4}$ & $10.86$ ($5.36$--$12.46$) & $1.48\times10^{-3}$ \\
1D heat, $M=16$, Dirichlet              & 16 & $1$ & $1.46\%$ & $3.08\times10^{-4}$ & $10.63$ ($5.26$--$12.48$) & $1.72\times10^{-3}$ \\
1D heat, $M=32$, Dirichlet & 32 & $1$ & $1.46\%$ & $2.58\times10^{-4}$ & $10.63$ ($5.27$--$12.47$) & $1.78\times10^{-3}$ \\
AD, $M=8$, Dirichlet  & 8  & $48.5$ & $1.15\%$ & $4.12\times10^{-4}$ & $12.10$ ($2.73$--$12.46$) & $1.83\times10^{-3}$ \\
AD, $M=16$, Dirichlet & 16 & $3.98\times10^{3}$ & $1.48\%$ & $4.68\times10^{-4}$ & $11.86$ ($2.68$--$12.48$) & $2.02\times10^{-3}$ \\
AD, $M=32$, Dirichlet & 32 & $2.63\times10^{7}$ & $1.47\%$ & $3.51\times10^{-4}$ & $11.86$ ($2.68$--$12.46$) & $2.07\times10^{-3}$ \\
2D heat, $4\times4$, Dirichlet          & 16 & $1$ & $7.40\%$ & $6.20\times10^{-3}$ & $0.94$ ($0.48$--$1.85$) & $8.07\times10^{-4}$ \\
\bottomrule
\end{tabular}}
\end{table}

Fig.~\ref{fig:clean_scaling} reports the scaling study, combining exact-truncated-map evaluations at eight evolution times $T\in\{0.25,0.5,0.75,1,1.25,1.5,1.75,2\}$ for $(N,N_{\rm Fock})\in\{(12,16),(24,32),(48,64)\}$ with circuit-level evaluations at $T=1$ over $n_t\in\{25,35,50,71,100,141,200,283\}$. Panel (a) shows a reversal between the two largest cutoff pairs at $T=0.5$, so increasing $(N,N_{\rm Fock})$ at fixed squeezing parameters can increase $\varepsilon_F$ at a sampled time. In panel (b), the measured product-formula error shows first-order scaling over the evaluated $n_t$ range, at log--log slopes $-1.012$, $-1.053$, and $-1.044$ for the three cutoff pairs. In panel (c), the benchmark-input postselection probability decreases from $16.60\%$--$19.28\%$ across the cutoff pairs at $T=0.5$ to $4.51\%$--$5.57\%$ at $T=2$, with nonmonotonic changes at some adjacent samples. A joint cost optimization must assess these quantities together. The present sweeps leave the resource-optimal point undetermined.

\begin{figure}[ht!]
\centering
\includegraphics[width=0.98\linewidth]{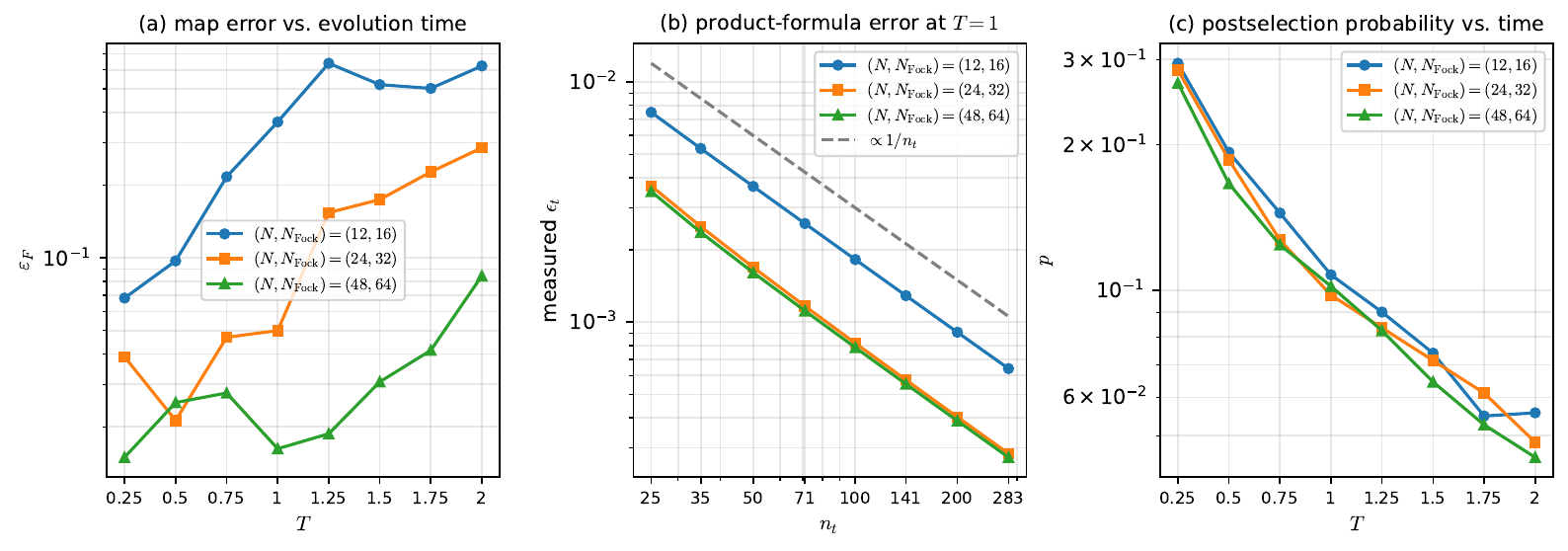}
\caption{Scaling study at the fixed kernel parameters $(r,r',\beta)=(1.6,0.25,0.5)$. (a) Fixed-scale map error $\varepsilon_F$ versus evolution time $T$ for three cutoff pairs (exact-truncated-map evaluations at eight times). (b) Measured product-formula error $\epsilon_t$ versus $n_t$ at $T=1$ (circuit-level, eight Trotter counts). The dashed line indicates the first-order scaling $\propto1/n_t$. (c) Postselection probability of the benchmark input versus $T$.}
\label{fig:clean_scaling}
\end{figure}

\subsection{Oscillator Photon Loss at the Map Level}

We evaluate a bounded sensitivity model for oscillator photon loss. A single lumped pure-loss channel $\Lambda_\eta$, with transmissivity $\eta$ and loss probability $1-\eta$, acts on the prepared finite kernel state before the $H=0$ joint evolution. Its Fock-basis Kraus operators are
\begin{align*}
A_k = \sqrt{\frac{(1-\eta)^k}{k!}}\, \eta^{\hat n/2}\hat a^k.
\end{align*}
For an input $u_0$, the postselected output is
\begin{align*}
    \rho_{\rm out}(u_0) = \sum_k \widetilde{\mathcal K}_k u_0u_0^\dagger \widetilde{\mathcal K}_k^\dagger, \qquad \widetilde{\mathcal K}_k 
    = (\bra{\phi_r}\otimes\mathbb I_D) U_{N_{\rm Fock}}(T) (A_k\ket{\psi_N}\otimes\mathbb I_D).
\end{align*}
The branch vectors $A_k\ket{\psi_N}$ are unnormalized, so their Kraus weights remain in the mixture, and the discarded branch weight is below $10^{-12}$. The conditional fidelity extends to this mixed output as
\begin {align*}
F_{\rm cond}(u_0)=\frac{\bra{u_{\rm exact}}\rho_{\rm out}(u_0)\ket{u_{\rm exact}}}{\|u_{\rm exact}\|_2^2\operatorname{tr}\rho_{\rm out}(u_0)}, 
\end{align*}
and Fig.~\ref{fig:noise_loss}(a) reports \(1-F_{\rm cond}\) at the benchmark input. At $\eta=1$, the branch-zero map reproduces the clean exact truncated map to machine precision for every boundary condition. 
This calculation uses a lumped finite-dimensional Kraus channel.

For $0<\eta\leq1$, contraction of the postselection block gives
\begin{align}
    \|\widetilde{\mathcal K}_0-\widetilde{\mathcal K}\|
    \leq \|(\eta^{\hat n/2}-\mathbb I)\ket{\psi_N}\|_2
    \leq -\frac{1}{2}\ln(\eta) \sqrt{\langle\hat n^2\rangle_{\psi_N}}. \label{eq:loss-distortion-bound}
\end{align}
Because $-\frac12\ln\eta=\frac{1-\eta}2+\mathcal O[(1-\eta)^2]$, the plotted quantity $[(1-\eta)/2]\sqrt{\langle\hat n^2\rangle}$ in Fig.~\ref{fig:noise_loss}(b) is a first-order approximation to the bound in Eq.~\eqref{eq:loss-distortion-bound}, not a rigorous upper bound at finite loss. The total input-state weight transferred to incoherent Kraus branches satisfies
\begin{align*}
    1-\langle\eta^{\hat n}\rangle_{\psi_N} \leq (1-\eta)\langle\hat n\rangle_{\psi_N}.
\end{align*}
At the selected parameter set $\theta^\star$, direct evaluation of the finite kernel state gives mean photon number $\langle\hat n\rangle=1.61$ and root-mean-square photon number $\sqrt{\langle\hat n^2\rangle}=6.72$. The formula $\langle\hat n\rangle=\sinh^2r'$ applies to an infinite-dimensional squeezed vacuum. 

The first two photon-number moments of the finite kernel state \(\ket{\psi_N}\) enter complementary analytic bounds for the measured loss response in Fig.~\ref{fig:noise_loss}. At $1-\eta=1\%$, the second moment gives the first-order distortion constant $[(1-\eta)/2]\sqrt{\langle\hat n^2\rangle}=3.4\times10^{-2}$ shown in Fig.~\ref{fig:noise_loss}(b). The measured kernel-state deviation is $3.1\times10^{-2}$, within $9\%$ of this value, and the coherent-branch map distortion, which measures loss-induced deformation of the implemented recovery filter, is $4.6\times10^{-3}$--$1.2\times10^{-2}$ across the three boundary conditions. At $1-\eta=5\%$, the first moment bounds the input-state weight transferred to incoherent Kraus branches by $(1-\eta)\langle\hat n\rangle$, equal to $8.0\%$. At this loss, the postselection probability in Fig.~\ref{fig:noise_loss}(c) shifts by less than $10\%$, and the lossy branches carry below $1\%$ of the accepted probability mass. The benchmark-input conditional infidelity in Fig.~\ref{fig:noise_loss}(a) is between $5.94\times10^{-4}$ and $9.20\times10^{-4}$ at $1-\eta=1\%$. These moments are independent of the simulated system, so photon-lean kernel parameters suppress the accuracy and sampling effects of photon loss.

The numerical study places a single pure-loss channel before the $H=0$ joint evolution. Displacement covariance of ideal pure loss motivates future studies of other placements within these branches. This model excludes distributed continuous-time loss, a Lindblad master equation, endpoint compensation, the $[L,H]\neq0$ case, qubit noise, oscillator dephasing, thermal noise, finite measurement efficiency, leakage, and calibration errors in the JC, SNAP, $\Dcal$, and $c\Dcal$ primitives.

\begin{figure}[ht!]
\centering
\includegraphics[width=0.98\linewidth]{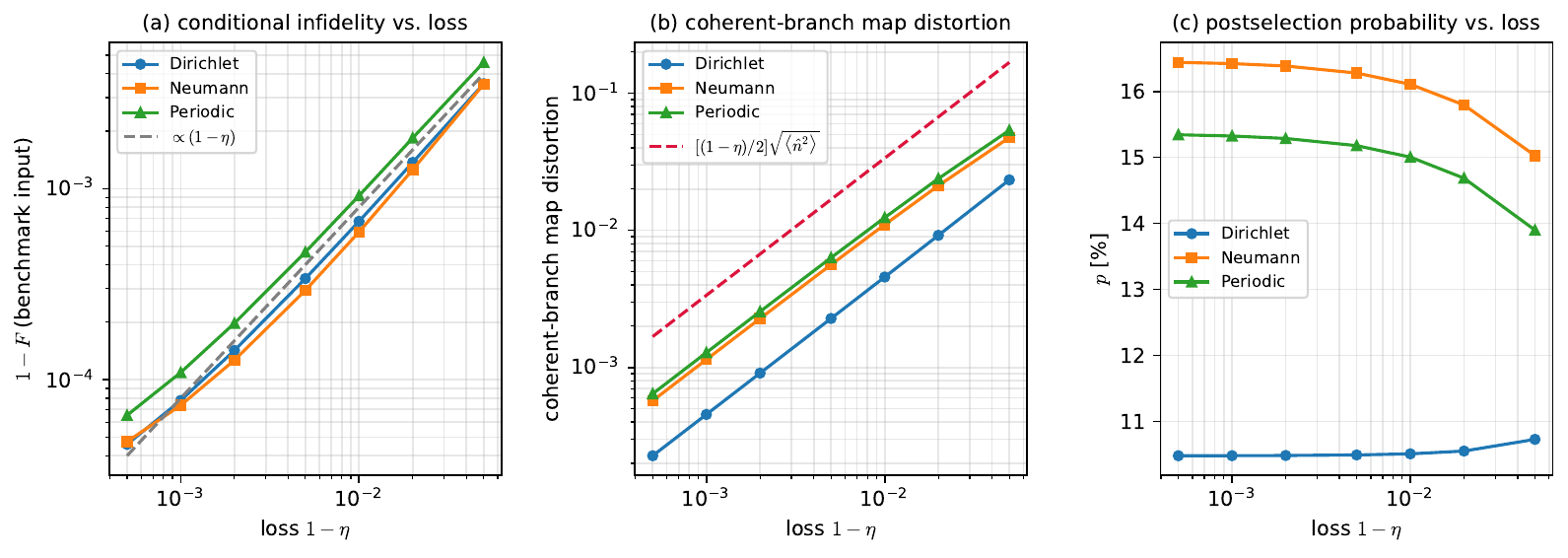}
\caption{Lumped pre-evolution oscillator photon loss with the parameter set $\theta^\star$. (a) Benchmark-input conditional infidelity \(1-F_{\rm cond}\) versus loss $1-\eta$, evaluated by exact Kraus summation. The dashed line indicates linear scaling in $1-\eta$. (b) Coherent-branch map distortion versus $1-\eta$, together with the first-order approximation $[(1-\eta)/2]\sqrt{\langle\hat n^2\rangle}$ to the distortion bound in Eq.~\eqref{eq:loss-distortion-bound}. (c) Physical postselection probability versus $1-\eta$.}
\label{fig:noise_loss}
\end{figure}

\subsection{Comparison with DV LCHS}
\label{ss:dv-compare}

For the DV LCHS baseline, the first question is the size of the quadrature discretization. In the homogeneous heat-equation benchmark considered here, $H=0$ and $L=A^{(bc)}$, where $bc$ corresponds to three boundary conditions, so the standard DV LCHS formulas in~\cite{LCHS1,LCHS2} give
\begin{align*}
    h_1&=\frac{1}{eT\|L\|_2},\,\
    K=\eta_{\rm DV}\left\lceil \frac{(\log(1/\epsilon))^{1/\beta}}{h_1}\right\rceil h_1,\,\
    Q= \left\lceil\frac{1}{\log 4}\log\!\left(\frac{8}{3C_\beta}\frac{K}{\epsilon}\right)\right\rceil, \\
    C_\beta &= 2\pi e^{-2^\beta}, \,\
    M_{\rm DV} = 2\left\lfloor \frac{K}{h_1}\right\rfloor Q,\,\
    m_c = \left\lceil \log_2 M_{\rm DV}\right\rceil.
\end{align*}
Here $h_1$ is the quadrature step size, $K$ is the truncation half-width of the $k$ interval, $Q$ is the Gauss--Legendre order on each subinterval, $M_{\rm DV}$ is the number of resulting LCU terms, $m_c$ is the control-register width, and $\eta_{\rm DV}$ is a safety factor on the truncation half-width. In Table~\ref{tab:dv_resource_compare}, we fix $\eta_{\rm DV}=1$ and $\epsilon=0.1$ and report parameter values. For each boundary condition, we then choose the value of $\beta$ from the scan $\beta\in[0.60,1)$ that maximizes the classical DV fidelity $F_{\rm DV} = F(u_{\rm DV},u_{\rm exact})$, where $u_{\rm DV}$ is the classical DV LCHS output for the same initial state. 
We compare auxiliary-register realizations that use the same \(g_\beta\) family with route-specific parameter selection. The DV route represents the continuous LCHS variable with a discretized quadrature register, and the CV--DV route uses an oscillator mode. A different kernel family for the DV baseline would also change the recovery weight. Optimization over DV kernel families lies outside this comparison.

Across the ten instances in Table~\ref{tab:dv_resource_compare}, the prescribed sizing $M_{\rm DV}$ requires between $192$ and $632$ quadrature terms on $8$ to $10$ ancillary qubits. The CV--DV oracle uses $32$ squeezed-Fock coefficients and $64$ Fock levels throughout, giving a compactness factor $M_{\rm DV}/N$ of $6.0$ to $19.8$. On the full $2^{9}$-node grid spanned by the Dirichlet ancilla register, the complex-kernel quadrature with exact branch propagators reproduces the Dirichlet benchmark solutions to infidelity at most $7.2\times10^{-7}$ at all four grid sizes ($M=4$, $8$, $16$, $32$). Because this calculation uses the full register grid, the remaining DV error at this sizing comes from the truncated quadrature.

At the prescribed $\epsilon=0.1$ truncation, with $\beta$ selected by the fidelity scan, the classical solution infidelity $1-F_{\rm DV}$ of Table~\ref{tab:dv_resource_compare} remains above the circuit-level conditional infidelity $1-F_{\rm LE}$ of the LE route in every instance. The fixed-scale metric of this section orders the two routes the same way in every row of Table~\ref{tab:dv_fixed_scale_compare}, where $\varepsilon_F^{\rm DV}$ reaches $16.48\%$ on the one-dimensional instances while the CV--DV route stays at or below $1.48\%$ there.

Per instance, the DV infidelity $1-F_{\rm DV}$ at this sizing exceeds the benchmark-input $1-F_{\rm LE}$ by a factor of $13$ to $163$ at the $M=4$ benchmarks, and the worst-input $1-F_{\rm cond}$ by a factor of $7.7$ to $47$ for the one-dimensional breadth instances and of $1.5$ for the two-dimensional case.
Reducing $\epsilon$ would lower these deviations at the cost of a larger $M_{\rm DV}$. The DV columns of Table~\ref{tab:dv_fixed_scale_compare} describe the prescribed sizing for the common \(g_\beta\) kernel family. The best attainable accuracy of the DV LCHS method, including optimization over DV kernels, remains undetermined.

The DV route has lower sampling overhead in every instance. At the one-dimensional benchmark inputs, the DV LCU block requires $3.8$ to $4.3$ expected attempts ($1/p_{\rm DV}$), compared with $6.07$ to $9.55$ ($1/p_{\rm succ}$) for the CV--DV route, a factor of $1.5$ to $2.4$. The difference arises because $\|c\|_1$ stays close to one and the CV--DV overhead carries the inverse-overlap scale $|\alpha_{N,r}|$ of the kernel encoding. In the two-dimensional case, the decayed solution norm limits both routes to $49.91$ and $106.89$ attempts, a factor of $2.1$.


\begin{table}[h]
\centering 
{\small
\setlength{\tabcolsep}{5pt}
\caption{Best DV LCHS quadrature parameters from the classical $\beta$ scan with $\epsilon=0.1$ and $\eta_{\rm DV}=1$. Here $\|c\|_1=\sum_j|c_j|$ is the $\ell^1$ norm of the complex quadrature coefficients, which sets the LCU normalization and postselection scale of the DV block, and the last two columns report the resulting classical solution infidelity $1-F_{\rm DV}$ and the fidelity gap $F_{\rm LE}-F_{\rm DV}$ relative to the Law--Eberly baseline. AD denotes advection--diffusion.}
\label{tab:dv_resource_compare}
\begin{tabular}{lccccccccc}
\toprule
Instance & $\beta_{\rm best}$ & $h_1$ & $K$ & $Q$ & $M_{\rm DV}$ & $m_c$ & $\|c\|_1$ & $1-F_{\rm DV}$ & $F_{\rm LE}-F_{\rm DV}$ \\
\midrule
1D heat, $M=4$, Dirichlet & 0.60 & 0.1017 & 4.067 & 4 & 320 & 9 & 0.9357 & $2.33\times10^{-3}$ & $2.32\times10^{-3}$ \\
1D heat, $M=4$, Neumann & 0.90 & 0.1077 & 2.586 & 4 & 192 & 8 & 1.2073 & $2.26\times10^{-3}$ & $2.25\times10^{-3}$ \\
1D heat, $M=4$, Periodic & 0.80 & 0.0920 & 2.851 & 4 & 248 & 8 & 1.0740 & $2.78\times10^{-4}$ & $2.56\times10^{-4}$ \\
\midrule
1D heat, $M=8$, Dirichlet               & 0.60 & 0.0948 & 4.078 & 4 & 344 & 9  & 0.9364 & $5.78\times10^{-3}$ & $5.75\times10^{-3}$ \\
1D heat, $M=16$, Dirichlet              & 0.60 & 0.0928 & 4.081 & 4 & 352 & 9  & 0.9367 & $4.94\times10^{-3}$ & $4.76\times10^{-3}$ \\
1D heat, $M=32$, Dirichlet & 0.60 & 0.0922 & 4.056 & 4 & 352 & 9 & 0.9350 & $5.00\times10^{-3}$ & $4.81\times10^{-3}$ \\
AD, $M=8$, Dirichlet  & 0.60 & 0.0948 & 4.078 & 4 & 344 & 9  & 0.9364 & $3.16\times10^{-3}$ & $3.14\times10^{-3}$ \\
AD, $M=16$, Dirichlet & 0.60 & 0.0928 & 4.081 & 4 & 352 & 9  & 0.9367 & $4.04\times10^{-3}$ & $3.91\times10^{-3}$ \\
AD, $M=32$, Dirichlet & 0.60 & 0.0922 & 4.056 & 4 & 352 & 9 & 0.9350 & $4.10\times10^{-3}$ & $3.91\times10^{-3}$ \\
2D heat, $4\times4$, Dirichlet          & 0.60 & 0.0508 & 4.016 & 4 & 632 & 10 & 0.9323 & $9.08\times10^{-3}$ & $3.58\times10^{-3}$ \\
\bottomrule
\end{tabular}}
\end{table}


\begin{table}[h]
\centering 
\caption{Fixed-scale map error and postselection probability of the DV LCHS baseline at the prescribed sizing of Table~\ref{tab:dv_resource_compare} (\(\epsilon=0.1\), \(\eta_{\rm DV}=1\), exact branch propagators), compared with the circuit-level CV--DV route. \(\varepsilon_F^{\rm DV}\) is the fixed-scale error of the unscaled quadrature map \(\sum_j c_je^{-iT(k_jL+H)}\) with no fitted rescaling, and \(p_{\rm DV}=\|\sum_j c_je^{-iT(k_jL+H)}\ket{u_0}\|_2^2/\|c\|_1^2\) is the acceptance probability of the LCU block. Probability columns are in percent and give the benchmark-input value with the range over all computational-basis inputs in parentheses. The CV--DV columns restate the fixed-scale errors of Tables~\ref{tab:clean_oracle_baseline} and~\ref{tab:clean_breadth} and the circuit-level probabilities of the same maps, and for the periodic instance both probabilities are input independent. AD denotes advection--diffusion.}
\label{tab:dv_fixed_scale_compare}
\begin{tabular}{lcccc}
\toprule
Instance & $\varepsilon_F^{\rm DV}$ & $\varepsilon_F$ (CV--DV) & $p_{\rm DV}$ [\%] & $p_{\rm succ}$ (CV--DV) [\%] \\
\midrule
1D heat, $M=4$, Dirichlet & $12.82\%$ & $0.70\%$ & $25.48$ ($11.71$--$25.48$) & $10.47$ ($5.28$--$10.47$) \\
1D heat, $M=4$, Neumann   & $16.48\%$ & $0.56\%$ & $24.70$ ($24.70$--$36.16$) & $16.48$ ($16.48$--$22.82$) \\
1D heat, $M=4$, Periodic  & $4.38\%$  & $0.47\%$ & $24.33$ & $15.36$ \\
1D heat, $M=8$, Dirichlet & $8.14\%$  & $1.04\%$ & $23.09$ ($11.00$--$26.14$) & $10.86$ ($5.36$--$12.46$) \\
1D heat, $M=16$, Dirichlet & $7.92\%$ & $1.46\%$ & $23.05$ ($10.98$--$26.16$) & $10.63$ ($5.26$--$12.48$) \\
1D heat, $M=32$, Dirichlet & $7.98\%$ & $1.46\%$ & $23.16$ ($11.03$--$26.24$) & $10.63$ ($5.27$--$12.47$) \\
AD, $M=8$, Dirichlet      & $8.10\%$  & $1.15\%$ & $26.09$ ($5.52$--$26.09$) & $12.10$ ($2.73$--$12.46$) \\
AD, $M=16$, Dirichlet     & $7.94\%$  & $1.48\%$ & $26.10$ ($5.54$--$26.10$) & $11.86$ ($2.68$--$12.48$) \\
AD, $M=32$, Dirichlet     & $7.99\%$  & $1.47\%$ & $26.21$ ($5.56$--$26.21$) & $11.86$ ($2.68$--$12.46$) \\
2D heat, $4\times4$, Dirichlet & $9.87\%$ & $7.40\%$ & $2.00$ ($1.00$--$3.82$) & $0.94$ ($0.48$--$1.85$) \\
\bottomrule
\end{tabular}
\end{table}


For circuit resource comparison, on the DV side, the ancilla amplitude vector can be compressed as a matrix product state (MPS), so ancilla preparation need not scale like generic amplitude loading. 
To realize the complex kernel in the DV circuits, the ancilla register is prepared with amplitudes proportional to $\sqrt{|g_\beta(k_j)|}$ over all binary-encoded quadrature nodes $k_j$, and a single diagonal layer $D_\Phi=\mathrm{diag}(e^{i\arg g_\beta(k_j)})$ is inserted before the inverse preparation to apply the kernel phases. The postselected block is proportional to the complex-kernel quadrature $\sum_j g_\beta(k_j)\,e^{-iTk_jL}$.

The two-layer MPS-to-circuit conversion used in the compiled circuits prepares the $2^{m_c}$-node amplitude profile approximately~\cite{ran2020mps}. For all Dirichlet experiments ($m_c = 9$), a $50$-draw reconstruction ensemble gives a state-preparation infidelity of $(1.51\pm0.13)\times10^{-3}$, reported as the mean and standard deviation over the draws, with extremes from $1.16\times10^{-3}$ to $1.93\times10^{-3}$. The DV loading counts describe an approximate preparation routine. Law--Eberly synthesis prepares the oscillator kernel state to numerical precision in the ideal circuit model. Exact isometry-based state preparation loads the same profile at a cost of $511$ one-qubit and $502$ CX gates in the same reporting basis~\cite{iten2016isometries}, with the CX count unchanged at transpiler optimization level $2$ (Table~\ref{tab:dv_circuit_compare}). Each cost is stated in its platform's native operations.

For the Dirichlet benchmark, the practical DV circuit considered here uses a nine-qubit control register and applies the ancilla preparation and unpreparation only once, while repeating the controlled LCHS evolution block:
\begin{align*}
    \langle 0_{\mathrm{anc}}| U_{\mathrm{prep}}^\dagger D_\Phi \bigl[U_{\mathrm{ctrl}}(T/n)\bigr]^n U_{\mathrm{prep}} |0_{\mathrm{anc}}\rangle, \qquad n=100, \qquad T=1.
\end{align*}
As summarized in Table~\ref{tab:dv_circuit_compare}, the one-time blocks are size independent because the $m_c=9$ quadrature register does not change with $M$, and the one-time overhead totals $772$ one-qubit and $606$ CX gates including the single input-loading $X$ gate. A single controlled-evolution slice instead grows from $56$ to $776$ CX gates between $M=4$ and $M=32$.

By comparison, the Dirichlet CV--DV implementation prepares the shared unsqueezed seed state $\ket{\chi_N}$ deterministically with $31$ JC pulses and $31$ qubit rotations independently of $M$, with the outer squeezing $\mathsf S_{N_{\rm Fock}}(r')$ that completes the finite kernel state excluded from these counts as in Table~\ref{tab:clean_resource_counts}, and its Trotterized middle block grows from $1400$ single-qubit gates, $400$ CXs, and $300$ $\rm c\Dcal$ at $M=4$ to $38600$, $19600$, and $3100$ at $M=32$, with $100$ $\Dcal$ throughout, shifting part of the oracle cost into oscillator operations.

\begin{table}[h]
\centering 
\caption{Operator counts for the Dirichlet DV LCHS circuits at heat-equation sizes $M=4,8,16,32$, each with $m_c=9$ control qubits, MPS-compressed ancilla PREP, the complex-kernel phase layer, and $n_t=100$ repeated controlled-evolution slices, in the \{1Q, CX\} reporting basis with each CRZ decomposed into two CX and two single-qubit rotations. The one-time rows are size independent, unPREP is the inverse of PREP at the same counts, and the exact isometry PREP rows are a numerically exact loading alternative~\cite{iten2016isometries} that is not part of the compiled circuits. Full-circuit depths are measured on the fully instantiated $100$-slice circuits, and the transpiled rows use Qiskit optimization level $2$ in a $\{U,\mathrm{CX}\}$ basis, which leaves every CX count unchanged.}
\label{tab:dv_circuit_compare}
\begin{tabular}{lcccc}
\toprule
& $M=4$ & $M=8$ & $M=16$ & $M=32$\\
\midrule
1Q gates, ancilla MPS~\cite{ran2020mps} PREP & \multicolumn{4}{c}{130} \\
CX gates, ancilla MPS~\cite{ran2020mps} PREP & \multicolumn{4}{c}{48} \\
\midrule
1Q gates, ancilla exact isometry~\cite{iten2016isometries} PREP & \multicolumn{4}{c}{511} \\
CX gates, ancilla exact isometry~\cite{iten2016isometries} PREP & \multicolumn{4}{c}{502} \\
\midrule
1Q gates, kernel phase layer & \multicolumn{4}{c}{511} \\
CX gates, kernel phase layer & \multicolumn{4}{c}{510} \\
\midrule
1Q gates per slice & 71 & 273 & 637 & 1145 \\
CX gates per slice & 56 & 186 & 408 & 776 \\
\midrule
1Q gates, full circuit & 7872 & 28072 & 64472 & 115272\\
CX gates, full circuit & 6206 & 19206 & 41406 & 78206 \\
Depth, full circuit & 12874 & 36975 & 77466 & 132964\\
\midrule
1Q gates, full circuit, transpiled & 7211 & 21612 & 46912 & 84814\\
CX gates, full circuit, transpiled & 6206 & 19206 & 41406 & 78206 \\
Depth, full circuit, transpiled & 12263 & 30567 & 60159 & 106049\\
\bottomrule
\end{tabular}
\end{table}

The DV slice cost grows with grid size because the binary shift operators that encode the discretized Laplacian use carry chains of multi-controlled NOT gates on the system register~\cite{douglas2009walks}. Decomposition into elementary gates introduces a factor of $\lceil\log_2 M\rceil$~\cite{barenco1995elementary}. The $k$-dependent part contributes a fixed count of $2m_c$ CX gates per slice through one controlled rotation per quadrature-register bit. The carry network accounts for the remaining grid-size dependence. The hybrid Trotter block directly compiles the Pauli decomposition of $L$, whose term count grows linearly with $M$. A ladder-string compilation of the controlled-displacement blocks could reduce this cost in future work. Across the studied instances, the register compactness persists and the Trotter-block CX ratio narrows from $14$ at $M=4$ to $4.0$ at $M=32$. The comparison covers the studied sizes. Asymptotic scaling remains open.

\section{Conclusion and Future Work}
\label{sec:conc}

We developed a finite-resource realization of the LCHS recovery filter on a hybrid oscillator-qubit architecture. The construction shares its one-mode unitary dilation with Schr\"odingerisation~\cite{Jin2024} and qumodisation~\cite{https://doi.org/10.1111/sapm.70047}. Its distinguishing element is the encoding of the LCHS kernel into physically prepared and postselected oscillator states. Theorem~\ref{thm:state-prep} and Corollary~\ref{cor:near-optimal-kernel} control the coefficient-projection error of the ideal kernel state. Theorem~\ref{thm:trotter} controls the product-formula cost through the independent oscillator cutoff \(N_{\rm Fock}\), and Theorem~\ref{thm:postselection-success} propagates the total scaled-map error $\epsilon_{\rm tot}$ to the postselection probability. A single error budget thus governs output accuracy and sampling overhead. We measure the remaining finite-model effects, and the measured benchmark error components are consistent with their corresponding bounds. For the Dirichlet benchmark, the exact-prefactor first-order certificate lies about three orders of magnitude above the measured product-formula error. Direct measurement supports the operating parameters where the worst-case bound is loose.

Across full circuit-level maps for the one-dimensional heat and non-normal advection--diffusion instances up to $D=32$, the fixed-scale map error is at most $1.48\%$ and the worst-input conditional infidelity is $4.68\times10^{-4}$. At the same shared parameter set, the $4\times4$ two-dimensional stress case reaches a fixed-scale error of $7.40\%$. Section~\ref{sec:clean_experiments} separates the contributions behind this growth. Against the DV LCHS baseline at its prescribed sizing, the hybrid oracle uses a single oscillator mode in place of the 8-to-10-qubit quadrature register and attains the smaller fixed-scale error in all ten instances. This compactness requires oscillator state preparation, $\Dcal$ and $c\Dcal$ operations, and postselection repetition, for which no CX-equivalent cost is assigned. The DV route also requires fewer expected attempts in every instance. In two dimensions, the decayed solution norm limits both routes.

The photon-loss study covers a single lumped pure-loss channel acting on the prepared kernel state before the $H=0$ evolution. Section~\ref{sec:clean_experiments} defines the effects outside this model. Device-calibrated noise simulation and logical-oscillator encodings~\cite{noh2020,guo2026} are next steps. Another direction is to map recent optimal and discrete LCHS constructions, including the Weyl-calculus formulas~\cite{ni2026quantumeigenvaluetransformation}, onto normalized oscillator states with a matched finite-resource analysis. Extensions to nonlinear PDEs through Carleman linearization and to inhomogeneous systems require noise, state-synthesis, and problem-specific oracle analyses before their practical advantage can be assessed~\cite{gan2025provably}. Within the present scope, the results support hybrid CV--DV LCHS as an alternative to qubit-only formulations when quadrature complexity dominates resource costs.


\section*{Acknowledgments}
ERD acknowledges fruitful discussions with Kaustubh Chandramouli. ERD, MZ, AL, TS, and YL acknowledge support by the U.S. Department of Energy, Office of Science, Advanced Scientific Computing Research, under contract number DE-SC0025384. MZ and RD acknowledge support from Pacific Northwest National Laboratory's Quantum Algorithms and Architecture for Domain Science (QuAADS) Laboratory Directed Research and Development (LDRD) Initiative. MZ and RD also acknowledge support from the U.S. Department of Energy, Office of Science, National Quantum Information Science Research Centers, Quantum Science Center (under FWP 76213) at Pacific Northwest National Laboratory. The Pacific Northwest National Laboratory is operated by Battelle for the U.S. Department of Energy under Contract No. DE-AC05-76RL01830. This research used resources of the National Energy Research Scientific Computing Center, a DOE Office of Science User Facility supported by the Office of Science of the U.S. Department of Energy under Contract No. DE-AC02-05CH11231 using NERSC award ASCR-ERCAP0037555.

\section*{Author Contributions}
ERD, MZ, RD, AL, and YL contributed to conceptualization. ERD, MZ, RD, and YL contributed to methodology. ERD, MZ, and TS contributed to software. ERD, MZ, and YL contributed to validation and formal analysis. ERD and MZ contributed to investigation and resources. MZ contributed to data curation and visualization. ERD, MZ, RD, AL, TS, and YL contributed to writing the original draft. ERD, MZ, RD, and YL contributed to writing--review and editing. TS and YL contributed to supervision and project administration. YL contributed to funding acquisition.

\section*{Competing Interests}
All authors declare no competing interests.

\section*{Data and Code Availability}
All data and scripts are available in the GitHub repository~\cite{cv-dv-repo}.

\bibliography{refs}
\bibliographystyle{unsrt}

\clearpage
\appendix
\section{Appendix}

\subsection{Notation used in the main text}
\label{app:notation}

\begin{table}[!htbp]
\centering 
\caption{Recurring notation in the main text. A relationship between an ideal quantity and its finite counterpart is stated in the finite quantity's row. Symbols local to one derivation or used only in the Appendix are omitted.}
\label{tab:notation}
\scriptsize
\setlength{\tabcolsep}{5pt}
\renewcommand{\arraystretch}{1.4}
\begin{tabular}{p{0.11\textwidth}p{0.66\textwidth}p{0.13\textwidth}} 
\toprule
Notation & Meaning and relation & Defined in \\
\midrule
\(A=L+iH,\newline D\) & Generator of the target nonunitary evolution, with dissipative Hermitian part \(L\), coherent Hermitian part \(H\), and system dimension \(D\). & Sec~\ref{sec:methodology}, Thm~\ref{thm:CV--DV-lchs} \newline Eq.~\eqref{eq:cart} \\ \hdashline
\(u(t),\ \ket{u_0},\ T\) & Statevector of the target ODE, the normalized initial state, and the total evolution time. \(\ket{u(t)}=e^{-At}\ket{u_0}\) (homogeneous time-independent). & Sec~\ref{sec:background}\newline Eq.~\eqref{eq:ode} \\ \hdashline
\(g(k),\ g_\beta(k),\ \beta\) & \(g\) is an admissible exact-LCHS kernel. The selected family \(g_\beta\) uses \(0<\beta<1\), which controls its tail and regularity. & Sec~\ref{sec:background}, Eq.~\eqref{eq:kernel_def} \newline Thm~\ref{thm:CV--DV-lchs} \\ \hdashline
\(\ket{\psi},\ \ket{\phi},\ \alpha_g\) & Normalized ideal oscillator preparation and postselection states, with \(g=\alpha_g\phi^*\psi\). The factor \(\alpha_g\) is the ideal normalization scale. & Thm~\ref{thm:CV--DV-lchs} \\ \hdashline
\(\mathcal B_g(t),\ \mathcal K_g(t)\) & Unscaled ideal postselection block and scaled exact map, with \(\mathcal K_g(t)=\alpha_g\mathcal B_g(t)=e^{-At}\). & Thm~\ref{thm:CV--DV-lchs}\newline Eq.~\eqref{eq:operator} \\ \hdashline
\addlinespace[2pt]
\(\hat x,\ \hat a,\ \hat a^\dagger,\ S(r)\) & Oscillator quadrature, ladder operators, and the single-mode squeezing operator, under the \(\hbar=2\) convention \([\hat x,\hat p]=2i\) and \(\hat x=\hat a+\hat a^\dagger\). The truncated-generator counterpart of \(S\) is \(\mathsf S_{N_{\rm Fock}}\). & Sec~\ref{sec:methodology},~\ref{ss:state-prep},~\ref{ss:error-bound-sp} \\ \hdashline
\(r,\ r'\), \(\sigma=e^r, \sigma'=e^{r'}\) & Postselection squeezing and preparation-basis squeezing, respectively, with \(r'<r\), and the corresponding Gaussian widths. & Sec~\ref{ss:state-prep} \\ \hdashline
\(\psi_\infty=\mathcal N g_\beta\) & Ideal normalized kernel state before coefficient or oscillator-space truncation. & Sec~\ref{ss:state-prep}, Thm~\ref{thm:state-prep} \\ \hdashline
\(N\) & Number of squeezed-Fock coefficients kept after truncation. & Sec~\ref{ss:state-prep}, Thm~\ref{thm:state-prep} \\ \hdashline
\(\psi_r^{\rm form}\),\newline \(\phi_{n,r'}(x)\),\newline \(\widetilde C_n,\ C_n\) & Formal finite-\(r\) coefficient-generation target, squeezed Fock basis functions, and the unnormalized and normalized expansion coefficients, with \(\widetilde C_n=\braket{\phi_{n,r'}|\psi^{\rm form}_r}\). The target \(\psi_r^{\rm form}=(2\pi\sigma^2)^{1/4}g(x)e^{x^2/(4\sigma^2)}\) is generally not square-integrable and enters only this pairing. & Sec~\ref{sec:methodology}, Sec~\ref{ss:state-prep}\newline Eq.~\eqref{eq:coefficients} \\ \hdashline
\(\Pi^F_{N,r'},\ \psi_{N,r'},\newline\widehat\psi_{N,r'}\) & Squeezed-Fock coefficient projector, the unnormalized ideal projection \(\psi_{N,r'}=\Pi^F_{N,r'}\psi_\infty\), and its normalized version. & Thm~\ref{thm:state-prep}\newline Cor~\ref{cor:normalized-projection} \\ \hdashline
\(N_{\rm Fock}\),\newline \(\Pi_{N_{\rm Fock}}\),\newline \(\hat x_{N_{\rm Fock}}\) & Ordinary-number-basis dimension, its projector, and the finite quadrature. This cutoff is independent of \(N\), although the implemented embedding requires \(N_{\rm Fock}\geq N\). & Sec~\ref{ss:error-bound-sp}\newline Eq.~\eqref{eq:x-norm-asymptotic} \\ \hdashline
\(\ket{\chi_N}\) & Normalized unsqueezed seed state \(\sum_{n=0}^{N-1}C_n\ket{n}\), the common target of both synthesis routes. Squeezing and renormalizing it gives \(\ket{\psi_N}\). & Eq.~\eqref{eq:fock}\newline Sec~\ref{ss:cv-prep} \\ \hdashline
\(\ket{\psi_N},\ \ket{\widetilde\psi_N}\),\newline \(\ket{\phi_r},\ \alpha_{N,r}\) & Finite kernel state serving as the numerical target, its synthesized approximation, the postselection state, and the inverse-overlap scale \(\alpha_{N,r}=\langle\phi_r|\psi_N\rangle^{-1}\). The full \(N_{\rm Fock}\) and \(r'\) dependence is defined before it is suppressed. & Sec~\ref{ss:error-bound-sp},~\ref{ss:postselection}\newline Eq.~\eqref{eq:finite-target-states},~\eqref{eq:evol-op} \\ \hdashline
\(\epsilon_{\rm tr}(N)\) & Ideal squeezed-Fock coefficient-projection error. It depends on \(N\), not \(N_{\rm Fock}\). & Cor~\ref{cor:normalized-projection}, Eq.~\eqref{eq:joint-budget-split} \\ \hdashline
\(r_\star,\ \delta_{\rm nG}\) & Stellar rank \(r_\star=\max\{n:C_n\neq0\}\leq N-1\) and quantum relative-entropy non-Gaussianity of the prepared state. & Sec~\ref{ss:non-gauss} \\ \hdashline
\addlinespace[2pt]
\(g_{N,r},\ \mathcal K_{N,r}\) & Implemented finite-\(r\) kernel and its ideal-oscillator postselected map. They are finite counterparts of \(g\) and the exact map and remain distinct from the fully finite map below. & Lem~\ref{lem:finite-r-map} \\ \hdashline
\(H_{{\rm hyb},N_{\rm Fock}}\) & Finite hybrid Hamiltonian obtained by replacing \(\hat x\) with \(\hat x_{N_{\rm Fock}}\). & Sec~\ref{sec:trotter}, Thm~\ref{thm:trotter} \\ \hdashline
\(p,\ n_t,\newline \Lambda_p(N_{\rm Fock})\) & Product-formula order, number of product-formula steps, and cutoff-dependent commutator factor governing the step bound. & Sec~\ref{sec:trotter}, Thm~\ref{thm:trotter}\newline Eq.~\eqref{eq:Lambda_p_definition} \\ \hdashline
\(U_{N_{\rm Fock}}(t)\),\newline \(U_{n_t}(t)\) & Exact finite hybrid evolution \(e^{-iH_{{\rm hyb},N_{\rm Fock}}t}\) and its \(n_t\)-step product-formula approximation. & Sec~\ref{ss:postselection} \\ \hdashline
\(\mathcal B^{\rm tar}_{N,N_{\rm Fock}}(t)\),\newline \(\mathcal B^{\rm prep}_{N,N_{\rm Fock}}(t)\),\newline \(\widetilde{\mathcal K}_{N,N_{\rm Fock},n_t}(t)\) & Unscaled finite maps obtained from the target state, the synthesized state, and the synthesized state with product-formula evolution, respectively. & Sec~\ref{ss:postselection}\newline Eq.~\eqref{eq:evol-op} \\ \hdashline
\(\epsilon_{\rm model},\ \epsilon_{\rm synth},\newline\epsilon_t,\ \epsilon_{\rm tot}\) & Finite-model, state-synthesis, product-formula, and total scaled-map errors. They obey \(\epsilon_{\rm tot}\leq\epsilon_{\rm model}+|\alpha_{N,r}|(\epsilon_{\rm synth}+\epsilon_t)\). & Sec~\ref{ss:postselection}\newline Eq.~\eqref{eq:total-scaled-map-error} \\ \hdashline
\(p_{\rm succ}^{(\infty)}\), \(p_{\rm ref}^{(N,r)}\),\newline \(p_{\rm succ}^{(N,N_{\rm Fock},n_t)}\) & Ideal physical success probability, comparison probability using the same finite overlap scale, and finite physical success probability, respectively. & Sec~\ref{ss:postselection}, Thm~\ref{thm:postselection-success}\newline Eqs.~\eqref{eq:psucc-perturbation},~\eqref{eq:psucc-time-bounds} \\ \hdashline
\(\varepsilon_F\) & Fixed-scale relative Frobenius-norm map error reported in the numerical study. It differs from the spectral-norm \(\epsilon_{\rm tot}\) and uses no per-input rescaling. & Sec~\ref{sec:clean_experiments} \\ \hdashline
\(\varepsilon_{\rm rel}(u_0),\newline F_{\rm cond}(u_0)\) & Per-input scaled relative error at the fixed scale \(\alpha_{N,r}\) and conditional fidelity of the postselected circuit-level map. Under photon loss, \(F_{\rm cond}\) is evaluated on the trace-normalized mixed output. & Sec~\ref{sec:clean_experiments} \\ \hdashline
\(F,\ F_{\rm LE},\ F_{\rm DV}\) & State fidelity \(F(\psi,\phi)=|\braket{\psi|\phi}|^2\), the conditional end-to-end fidelity of the LE route, and the classical DV LCHS fidelity, both against \(u_{\rm exact}\). & Sec~\ref{sec:clean_experiments},~\ref{ss:dv-compare} \\ \hdashline
\(\theta^\star\) & Kernel parameter set \((r,r',\beta,N)=(1.6,\,0.25,\,0.5,\,32)\) selected by the coarse grid search and used in all experiments except the scaling sweeps (\((r,r',\beta)\) fixed and \((N,N_{\rm Fock})\) varied), the SNAP+\(\Dcal\) study (also evaluates \(N=16\)), the common local \((r,r')\) scan of Appendix~\ref{app:ha-retune}, including the Huang--An evaluations at its grid-selected point, and the two-dimensional \((r,r')\) scan of Appendix~\ref{app:ha-retune-2d}. & Sec~\ref{sec:clean_experiments} \\
\bottomrule
\end{tabular}
\end{table}

\subsection{Moment-matching interpretation of the oscillator oracle}
\label{app:moment-matching}

The oscillator construction can be placed within the moment-matching dilation framework of~\cite{xm61-ytf7}. Consider
\begin{align*}
    \frac{d}{dt}\ket{u(t)} = -(L+iH)\ket{u(t)} = \bigl(-iH+K_{\rm mm}\bigr)\ket{u(t)},
    \qquad
    K_{\rm mm}:=-L.
\end{align*}
In the moment-matching framework, a skew-Hermitian ancillary operator \(F\), a right encoding vector \(\ket{r)}\), and a left evaluation functional \((l|\) obey the algebraic moment conditions
\begin{align}
    (l|F^m|r)=1, \qquad m=0,1,2,\ldots . \label{eq:moment-condition}
\end{align}
The exact dilation theorem in~\cite{xm61-ytf7} additionally requires its stated function-space, domain, and convergence hypotheses.
The corresponding dilated Hamiltonian is
\begin{align*}
    \mathcal H_{\rm mm} = \mathbb I_{\rm osc}\otimes H + iF\otimes K_{\rm mm}.
\end{align*}
For the oscillator realization, choose
\begin{align*}
    F=i\hat{x}, \qquad K_{\rm mm}=-L.
\end{align*}
Since \(\hat{x}\) is Hermitian, \(F=i\hat{x}\) is skew-Hermitian, and
\begin{align*}
    \mathcal H_{\rm mm} = \mathbb I_{\rm osc}\otimes H + i(i\hat{x})\otimes(-L)
    = \hat{x}\otimes L + \mathbb I_{\rm osc}\otimes H,
\end{align*}
which is precisely the hybrid oscillator-qubit Hamiltonian used in our construction.

\begin{proposition}[Physical oscillator realization of the LCHS moment triple]
Let \(g(x)\) be an LCHS kernel satisfying
\begin{align}
    \widehat g(y) := \int_{\mathbb R}g(x)e^{-ixy}\,\dd x = e^{-y}, \qquad y\geq 0. \label{eq:lchs-fourier-property}
\end{align}
Suppose normalized oscillator states \(\ket{\psi}\) and \(\ket{\phi}\), together with a known nonzero scalar \(\alpha_{\phi,\psi}\), satisfy
\begin{align}
    g(x) = \alpha_{\phi,\psi}\phi^*(x)\psi(x). \label{eq:physical-kernel-factorization}
\end{align}
Assume additionally that
\begin{align*}
    x^m g(x)\in L^1(\mathbb R), \qquad m=0,1,2,\ldots,
\end{align*}
and that the oscillator matrix elements \(\bra{\phi}\hat{x}^m\ket{\psi}\) are well defined. For example, \(\ket{\psi}\in\bigcap_{m\geq0}\operatorname{Dom}(\hat{x}^m)\). These assumptions justify differentiation of the Fourier identity under the integral sign.
Define
\(
    |r) = c_r\ket{\psi}, \,\ (l| = c_l\bra{\phi}, \,\ c_lc_r=\alpha_{\phi,\psi}.
\)
Then the triple
\(
\left( F=i\hat{x}, \ket{r)}, (l| \right)
\)
satisfies the algebraic moment conditions in Eq.~\eqref{eq:moment-condition}.
\end{proposition}

\begin{proof}
Using Eq.~\eqref{eq:physical-kernel-factorization},
\begin{align*}
    (l|F^m|r) = \alpha_{\phi,\psi} \bra{\phi}(i\hat{x})^m\ket{\psi} = \int_{\mathbb R} g(x)(ix)^m\,\dd x.
\end{align*}
Differentiating Eq.~\eqref{eq:lchs-fourier-property} \(m\) times from the right at \(y=0\) gives
\begin{align*}
    \int_{\mathbb R}g(x)(-ix)^m\,\dd x = (-1)^m.
\end{align*}
It follows that
\begin{align*}
    \int_{\mathbb R} g(x)(ix)^m\, \dd x = 1,
\end{align*}
and hence \((l|F^m|r)=1\) for every \(m\geq0\), which proves the stated algebraic moment conditions.
\end{proof}
A normalized physical factorization always exists when \(g\in L^1(\mathbb R)\). For example, setting
\begin{align*}
    G := \int_{\mathbb R}|g(x)|\,\dd x,\,\ \phi(x) := \sqrt{\frac{|g(x)|}{G}}, \,\
    \psi(x) := e^{i\arg g(x)} \sqrt{\frac{|g(x)|}{G}},
\end{align*}
gives normalized states satisfying
\(
    g(x)=G\phi^*(x)\psi(x),
\)
so that \(\alpha_{\phi,\psi}=G\).

The exact moment conditions apply to the ideal continuous-variable construction. In the finite-resource implementation, the squeezed-vacuum postselection state \(\ket{\phi_r}\), finite squeezed-Fock kernel state \(\ket{\psi_N}\), and truncated quadrature operator \(\hat{x}_{N_{\rm Fock}}\) generally satisfy the conditions only approximately. Defining
\(
    \alpha_{N,r} := \frac{1}{\braket{\phi_r|\psi_N}},
\)
ensures the zeroth moment and the identity at \(t=0\). The higher-order moment defects are
\begin{align*}
    \delta_m^{(N,r)} &:= \alpha_{N,r} \bra{\phi_r} (i\hat{x}_{N_{\rm Fock}})^m \ket{\psi_N} -1, \qquad m \geq 1.
\end{align*}
Equivalently, the finite-resource oracle may be viewed as an approximate physical realization of the ideal moment-matching triple. Our operator-norm error analysis in Section~\ref{ss:postselection} controls the resulting approximation without requiring exact matching of every finite-dimensional moment.

\subsection{Proof of Theorem~\ref{thm:CV--DV-lchs}}
\label{proof:CV--DV-lchs}

Since $A=L+iH$, the solution of the initial-value problem is
\begin{align*}
    \ket{u(t)} = e^{-At}\ket{u_0} = e^{-t(L+iH)}\ket{u_0}.
\end{align*}
By the exact-LCHS hypothesis in Eq.~\eqref{eq:assumed-exact-lchs-identity},
\begin{align*}
e^{-t(L+iH)} = \int_{\mathbb R} g(x)e^{-it(xL+H)} \,\dd x, \qquad t\geq0.
\end{align*}
Now expand the oscillator states in the position basis:
\begin{align*}
    \ket{\psi}_{\rm osc}
    = \int_{\mathbb R}\psi(x)\ket{x}\,\dd x,
    \qquad
    \bra{\phi}_{\rm osc}
    = \int_{\mathbb R}\phi^*(x)\bra{x}\,\dd x .
\end{align*}
The position operator has the spectral decomposition
\begin{align*}
    \hat{x} = \int_{\mathbb R} x\ket{x}\bra{x} \,\dd x .
\end{align*}
Therefore, for any qubit-register state $\ket{v}$,
\begin{align*}
    \left( \hat{x}\otimes L+\mathbb I_{\rm osc}\otimes H \right)
    \left( \ket{x}\otimes\ket{v} \right)
    = \ket{x}\otimes (xL+H)\ket{v}.
\end{align*}
Thus the joint Hamiltonian is diagonal with respect to the oscillator position basis. Hence, by functional calculus,
\begin{align*}
    e^{-it(\hat{x}\otimes L+\mathbb I_{\rm osc}\otimes H)} \left( \ket{x}\otimes\ket{v} \right)
    = \ket{x}\otimes e^{-it(xL+H)}\ket{v}.
\end{align*}
Applying this identity gives
\begin{align*}
    e^{-it(\hat{x}\otimes L+\mathbb I_{\rm osc}\otimes H)}
    \left( \ket{\psi}_{\rm osc}\otimes\mathbb{I}_q \right)
    = \int_{\mathbb R} \psi(x) \ket{x}\otimes e^{-it(xL+H)} \,\dd x .
\end{align*}
Projecting the oscillator register onto \(\bra{\phi}_{\rm osc}\) gives
\begin{align*}
    \mathcal B_g(t) 
    &= (\bra{\phi}_{\rm osc}\otimes\mathbb I_q) e^{-it(\hat{x}\otimes L+\mathbb I_{\rm osc}\otimes H)} (\ket{\psi}_{\rm osc}\otimes\mathbb I_q)\\
    &= \int_{\mathbb R}\phi^*(x)\psi(x)e^{-it(xL+H)}\,\dd x\\
    &= \frac{1}{\alpha_g}\int_{\mathbb R}g(x)e^{-it(xL+H)}\,\dd x\\
    &= \frac{e^{-At}}{\alpha_g}.
\end{align*}
Thus,
\begin{align*}
    \mathcal K_g(t) := \alpha_g\mathcal B_g(t) = e^{-At}.
\end{align*}
Applying this operator identity to an arbitrary initial state
\(\ket{u_0}\in\mathbb C^D\) gives
\begin{align*}
    \mathcal K_g(t)\ket{u_0} = e^{-At}\ket{u_0} = \ket{u(t)},
\end{align*}
which proves the claimed kernel-modular continuous--discrete variable LCHS representation.

\subsection{Analytical Values of the Unnormalized Expansion Coefficients \texorpdfstring{$\widetilde C_n$}{Cn} at Limit Values of \texorpdfstring{$\beta$}{beta}}
\label{ss:analytic}
The unnormalized coefficients are defined by the coefficient-generating pairing of the formal signal with the squeezed Hermite basis in the $\hbar=2$ convention:

\begin{align*}
    \widetilde C_n(\beta) = \mathcal{K}_n \int_{-\infty}^{\infty} \frac{H_n\!\left(x/(\sqrt{2}\sigma')\right)e^{-\gamma x^2}} {2\pi e^{-2^\beta}e^{(1+ix)^\beta}(1-ix)} \,\dd x .
\end{align*}

Here, $\mathcal{K}_n=\sqrt{\frac{\sigma}{\sigma'}}\frac{1}{\sqrt{2^n n!}}$ and $\gamma=\frac14(\sigma'^{-2}-\sigma^{-2})$. We denote the limit-point functions by $\widetilde C_n^{(0)}$ and $\widetilde C_n^{(1)}$, corresponding to $\beta=0$ and $\beta=1$, respectively. The normalized physical coefficients $C_n$ are obtained from $\widetilde C_n$ by the finite-$N$ normalization defined in Section~\ref{ss:state-prep}.

We employ an explicit polynomial decomposition of the Hermite--Lorentzian kernel. Define
\begin{align*}
    Q_{n-1}(x) =\frac{H_n\!\left(x/(\sqrt{2}\sigma')\right)-H_n\!\left(-i/(\sqrt{2}\sigma')\right)}{x+i}.
\end{align*}
This isolates the singular contribution of the complex pole at $x=-i$.

At $\beta=0$, the exact closed form is
\begin{align*}
    \widetilde C_n^{(0)}
    =\mathcal{K}_n\left[\frac12H_n\!\left(-\frac{i}{\sqrt{2}\sigma'}\right)e^\gamma \operatorname{erfc}(\sqrt\gamma) +\frac{i}{2\pi}\int_{-\infty}^{\infty}Q_{n-1}(x)e^{-\gamma x^2}\,\dd x
\right].
\end{align*}

For $\beta=1$, completing the square in $-\gamma x^2-ix$ gives
\begin{align*}
    \widetilde C_n^{(1)}
    = \mathcal{K}_n\left[ \frac12H_n\!\left(-\frac{i}{\sqrt{2}\sigma'}\right)e^\gamma \operatorname{erfc}\!\left(\sqrt\gamma-\frac{1}{2\sqrt\gamma}\right) 
    + \frac{ie}{2\pi}\int_{-\infty}^{\infty}Q_{n-1}(x)e^{-\gamma x^2-ix}\,\dd x \right].
\end{align*}
The transition from $\beta=0$ to $\beta=1$ shifts the error-function argument by $(2\sqrt\gamma)^{-1}$.

The analytical results were compared against direct numerical quadrature using a double-exponential integration scheme with a working precision of 30 decimal places. For the parameter set $\sigma' = 1.2$ and $\sigma = 2.0$ ($\gamma = 1/9$), the agreement is presented in the Table~\ref{tab:lower_bounds} and Table~\ref{tab:upper_bounds}.

\begin{table}[h!]
\centering
\caption{Comparison of Numerical and Analytical Coefficients for the Lower Bound ($\beta = 0, \gamma = 1/9$)}
\label{tab:lower_bounds}
\begin{tabular}{rll}
\hline
$n$ & Numerical Result & Analytical (Explicit) \\
\hline
0 & 0.4597572626 & 0.4597572626 \\
1 & $0.5273259733i$ & $0.5273259733i$ \\
2 & $-0.0143676684$ & $-0.0143676684$ \\
3 & $0.3661994074i$ & $0.3661994074i$ \\
\hline
\end{tabular}
\end{table}

\begin{table}[h!]
\centering
\caption{Comparison of Numerical and Analytical Coefficients for the Upper Bound ($\beta = 1, \gamma = 1/9$)}
\label{tab:upper_bounds}
\begin{tabular}{rll}
\hline
$n$ & Numerical Result & Analytical (Explicit) \\
\hline
0 & 1.3713254670 & 1.3713254670 \\
1 & $-0.8819209173i$ & $-0.8819209173i$ \\
2 & $-0.7976665975$ & $-0.7976665975$ \\
3 & $-0.1673816287i$ & $-0.1673816287i$ \\
\hline
\end{tabular}
\end{table}

As demonstrated in Tables~\ref{tab:lower_bounds} and~\ref{tab:upper_bounds}, the explicit decomposition agrees with direct high-precision quadrature to all ten displayed digits.

\subsection{Proof of Theorem~\ref{thm:state-prep}}
\label{proof:state-prep}

Since the ordinary Hermite functions form a complete orthonormal basis of $L^2(\mathbb R)$, and the dilation $x=\sqrt{2}\,\sigma' y$ is unitary on $L^2(\mathbb R)$ after the normalization factor $1/\sqrt{\sqrt2\,\sigma'}$, the squeezed Fock functions also form a complete orthonormal basis of $L^2(\mathbb R)$. Therefore,
\begin{align*}
    \lim_{N\rightarrow\infty} \|\psi_\infty-\Pi^F_{N,r'}\psi_\infty\|_2=0,
\end{align*}
proving the existence of a cutoff $N$ for any prescribed $\epsilon_s$.

Under the change of variables $x=\sqrt{2}\,\sigma' y$, projection onto $\mathbb H_{N,r'}$ is equivalent to ordinary Hermite-function projection of $F_{r'}(y) = \sqrt{\sqrt2\,\sigma'}\, \psi_\infty(\sqrt2\,\sigma' y)$ onto $\mathrm{span}\{h_0,\ldots,h_{N-1}\}$. If $F_{r'}\in\mathcal S(\mathbb R)$, the Hermite projection estimate of~\cite[Chapter~III]{lubich2008quantum}, also recalled in~\cite[Remark 3.10]{wang2025}, gives, for every $s\leq N$,
\begin{align*}
    \|F_{r'}-\Pi_N^H F_{r'}\|_2  \leq  \frac{\|\mathcal{A}^sF_{r'}\|_2} {\sqrt{N(N-1)\cdots(N-s+1)}},
\end{align*}
where $\Pi_N^H$ denotes projection onto $\mathrm{span}\{h_0,\ldots,h_{N-1}\}$. The Hermite projection error computed for the rescaled function $F_{r'}(y)$ has the same $L^2$ norm as the squeezed-Fock projection error in the original variable $x$, which proves part (i). For fixed $s$, the denominator is $\Theta(N^{s/2})$, giving the stated $\mathcal O(N^{-s/2})$ rate.

For part~(ii), let
\begin{align*}
    F_{r'}(y)  = \sum_{n=0}^{\infty}  d_n^{(r')}h_n(y),
    \qquad   d_n^{(r')} = \langle h_n,F_{r'}\rangle,
\end{align*}
be the ordinary Hermite-function expansion of the rescaled target. By the Hermite-coefficient characterization of the symmetric Gelfand--Shilov space \(S_\mu^\mu(\mathbb R)\)~\cite{vanEijndhoven1987GS,lozanovperisic2007hermite}, the assumption \(F_{r'}\in S_\mu^\mu(\mathbb R)\) implies that there exists a constant \(\tau_{\mu,r'}>0\) such that
\begin{align}
    M_{\mu,r'}
    := \sum_{n=0}^{\infty}  |d_n^{(r')}|^2  \exp\!\left[  2\tau_{\mu,r'}n^{1/(2\mu)}  \right] <\infty . \label{eq:GS-weighted-Hermite}
\end{align}
Since \(\Pi_N^H F_{r'}=\sum_{n=0}^{N-1}d_n^{(r')}h_n\), orthonormality gives
\begin{align*}
    \|F_{r'}-\Pi_N^H F_{r'}\|_2^2
    = \sum_{n=N}^{\infty}|d_n^{(r')}|^2
    = \sum_{n=N}^{\infty}  |d_n^{(r')}|^2  e^{2\tau_{\mu,r'}n^{1/(2\mu)}}  e^{-2\tau_{\mu,r'}n^{1/(2\mu)}}.
\end{align*}
Because \(n\mapsto n^{1/(2\mu)}\) is increasing, for every \(n\geq N\),
\(
    e^{-2\tau_{\mu,r'}n^{1/(2\mu)}}
    \leq
    e^{-2\tau_{\mu,r'}N^{1/(2\mu)}}.
\)
It follows from~\eqref{eq:GS-weighted-Hermite} that
\begin{align*}
    \|F_{r'}-\Pi_N^H F_{r'}\|_2^2
    \leq  e^{-2\tau_{\mu,r'}N^{1/(2\mu)}}
    \sum_{n=N}^{\infty} |d_n^{(r')}|^2  e^{2\tau_{\mu,r'}n^{1/(2\mu)}}
    \leq M_{\mu,r'} e^{-2\tau_{\mu,r'}N^{1/(2\mu)}}.
\end{align*}
Taking the square root yields
\begin{align*}
    \|F_{r'}-\Pi_N^H F_{r'}\|_2
    \leq C_{\mu,r'}\exp\!\left[  -\tau_{\mu,r'}N^{1/(2\mu)}  \right],
    \qquad C_{\mu,r'}:=M_{\mu,r'}^{1/2}.
\end{align*}
The unitary coordinate dilation relating \(F_{r'}\) to \(\psi_\infty\) preserves the \(L^2\) norm and maps ordinary Hermite projection to squeezed-Fock projection. So,
\begin{align*}
    \|\psi_\infty-\psi_{N,r'}\|_2
    \leq  C_{\mu,r'}  \exp\!\left[  -\tau_{\mu,r'}N^{1/(2\mu)}  \right].
\end{align*}
To ensure that this error is at most \(\epsilon_s\), it is sufficient to choose
\begin{align*}
    N \geq  \left[ \frac{1}{\tau_{\mu,r'}}\log\!\left(  \frac{C_{\mu,r'}}{\epsilon_s}   \right)  \right]^{2\mu}.
\end{align*}
Hence, for fixed \(r'\) and \(\mu\),
\begin{align*}
    N = \mathcal O\!\left( \log^{2\mu}\frac{1}{\epsilon_s}  \right).
\end{align*}
This proves part~(ii).

For part (iii), the stated weighted growth and boundary-integrability assumptions are exactly the hypotheses needed to apply Theorem 3.8(ii) in~\cite{wang2025} to the rescaled target $F_{r'}$ in the strip $|\operatorname{Im}(y)|<\frac{\rho}{\sqrt2\,\sigma'}$. This gives
\begin{align*}
    \|F_{r'}-\Pi_N^H F_{r'}\|_2
    \leq
    C_{r'} \exp\!\left[ -\frac{\rho}{\sqrt2\,\sigma'} \sqrt{2(N-1)} \right],
\end{align*}
where $C_{r'}$ is independent of $N$. Returning to the $x$ variable again preserves the $L^2$ norm, so the same bound holds for $\|\psi_\infty-\psi_{N,r'}\|_2$. Solving the right-hand side for $\epsilon_s$ gives
\begin{align*}
    N=\mathcal O\!\left(e^{2r'}\log^2\frac{1}{\epsilon_s}\right),
\end{align*}
completing the proof.

\begin{lemma}[Gelfand--Shilov regularity of the LCHS kernel]
\label{lem:kernel-GS}
Let $g_\beta$ be the LCHS kernel from Eq. \eqref{eq:kernel_def}. Then $g_\beta\in S_{1/\beta}^{1/\beta}(\mathbb R)$. Moreover, for every fixed squeezing parameter \(r'\), the normalized and rescaled target
\begin{align*}
    F_{r'}(y) = \sqrt{\sqrt{2}\sigma'}\, \psi_\infty(\sqrt{2}\sigma' y),  \qquad  \sigma'=e^{r'},
\end{align*}
also belongs to \(S_{1/\beta}^{1/\beta}(\mathbb R)\).
\end{lemma}
\begin{proof}
Using the principal branch of \(z^\beta\), the function \(g_\beta\) is holomorphic in the open strip $|\operatorname{Im}z|<1$. Fix \(0<\delta<1\), and write \(z=x+iy\) with \(|y|\leq\delta\). Then
\begin{align*}
    1+iz=(1-y)+ix, \qquad 1-iz=(1+y)-ix.
\end{align*}
Since \(1-y>0\), one has $|\arg(1+iz)|<\pi/2$.
Therefore,
\begin{align*}
    \operatorname{Re}\!\left[(1+iz)^\beta\right] 
    = |1+iz|^\beta  \cos\!\left(\beta\arg(1+iz)\right) \geq \cos\!\left(\frac{\beta\pi}{2}\right)|x|^\beta .
\end{align*}
Moreover,
\begin{align*}
    |1-iz| =  |(1+y)-ix| \geq  1-\delta .
\end{align*}
Hence, after absorbing the fixed numerator and denominator factors into the constant, there exist
\(C_{\beta,\delta},a_{\beta,\delta}>0\) such that
\begin{align}
    |g_\beta(z)|
    \leq C_{\beta,\delta} \exp\!\left[ -a_{\beta,\delta}|x|^\beta \right],
    \qquad |\operatorname{Im}z|\leq\delta . \label{eq:kernel-strip-decay}
\end{align}

Set \(\rho=\delta/2\). For every real \(k\), Cauchy's derivative estimate on the circle \(|z-k|=\rho\) gives
\begin{align*}
    |g_\beta^{(q)}(k)|
    \leq q!\,\rho^{-q} \sup_{|z-k|=\rho}|g_\beta(z)|.
\end{align*}
On this circle,
\begin{align*}
    |\operatorname{Re}z| \geq (|k|-\rho)_+ .
\end{align*}
For \(|k|\geq2\rho\), this implies $ |\operatorname{Re}z|\geq|k|/2$. The bounded region \(|k|<2\rho\) can be absorbed into the
prefactor. Consequently, after decreasing \(a_{\beta,\delta}\) and increasing \(C_{\beta,\delta}\) if needed,
\begin{align}
    |g_\beta^{(q)}(k)|  \leq C_{\beta,\delta} B_\delta^q q!\,  \exp\!\left[ -a_{\beta,\delta}|k|^\beta \right],
    \qquad B_\delta:=\frac{2}{\delta}. \label{eq:kernel-derivative-decay}
\end{align}
Consequently,
\begin{align*}
    \left\|   k^p\frac{\dd^q}{\dd k^q}g_\beta  \right\|_2
    &\leq  C_{\beta,\delta}B_\delta^q q!
    \left(  \int_{\mathbb R}  |k|^{2p}e^{-2a_{\beta,\delta}|k|^\beta}  \,\dd k  \right)^{1/2}.
\end{align*}
The integral can be expressed using the Gamma function as
\begin{align*}
    \int_0^\infty
    k^{2p}e^{-2a_{\beta,\delta}k^\beta}\,\dd k
    = \frac{1}{\beta}   (2a_{\beta,\delta})^{-(2p+1)/\beta}  \Gamma\!\left(\frac{2p+1}{\beta}\right).
\end{align*}
Stirling's estimate then implies
\begin{align*}
    \left\|  k^p\frac{\dd^q}{\dd k^q}g_\beta  \right\|_2
    \leq  K_{\beta,\delta}  A_{\beta,\delta}^{p}  B_\delta^q  (p!)^{1/\beta}q!.
\end{align*}
Since \(0<\beta<1\), one has \(q!\leq(q!)^{1/\beta}\), and hence
\begin{align*}
    \left\|  k^p\frac{\dd^q}{\dd k^q}g_\beta  \right\|_2
    \leq  K_{\beta,\delta}  D_{\beta,\delta}^{p+q}  (p!)^{1/\beta}(q!)^{1/\beta}.
\end{align*}
This is the defining factorial estimate for \(g_\beta\in S_{1/\beta}^{1/\beta}(\mathbb R)\).

It remains to verify that multiplication by the normalization constant and an invertible coordinate dilation preserve the required Gelfand--Shilov estimate. Let
\(
    a_{r'}:=\sqrt{2}\sigma'>0.
\)
Since
\(
    F_{r'}(y) = \mathcal N a_{r'}^{1/2}g_\beta(a_{r'}y),
\)
a change of variables gives, for all \(p,q\in\mathbb N_0\),
\begin{align*}
    \left\|  y^p\frac{\dd^q}{\dd y^q}F_{r'}  \right\|_2
    &= |\mathcal N|a_{r'}^{q-p} \left\| k^p\frac{\dd^q}{\dd k^q}g_\beta \right\|_2 .
\end{align*}
Furthermore,
\begin{align*}
    a_{r'}^{q-p} \leq \max\{a_{r'},a_{r'}^{-1}\}^{p+q}.
\end{align*}
Combining this inequality with the preceding factorial estimate yields
\begin{align*}
    \left\| y^p\frac{\dd^q}{\dd y^q}F_{r'} \right\|_2
    &\leq |\mathcal N|K_{\beta,\delta} \left[ D_{\beta,\delta} \max\{a_{r'},a_{r'}^{-1}\}  \right]^{p+q}  (p!)^{1/\beta}(q!)^{1/\beta}.
\end{align*}
Thus, defining
\begin{align*}
    K_{\beta,r'}:=|\mathcal N|K_{\beta,\delta}, \qquad D_{\beta,r'} := D_{\beta,\delta}\max\{a_{r'},a_{r'}^{-1}\},
\end{align*}
we obtain
\begin{align*}
    \left\| y^p\frac{\dd^q}{\dd y^q}F_{r'}  \right\|_2
    \leq K_{\beta,r'} D_{\beta,r'}^{p+q} (p!)^{1/\beta}(q!)^{1/\beta}.
\end{align*}
Thus, \(F_{r'}\in S_{1/\beta}^{1/\beta}(\mathbb R)\) for every fixed \(r'\).
\end{proof}

\subsection{Proof of Corollary~\ref{cor:near-optimal-kernel}}
\label{proof:corollary}
\begin{proof}
By Lemma~\ref{lem:kernel-GS}, the rescaled target \(F_{r'}\) belongs to the symmetric Gelfand--Shilov class \(S_{1/\beta}^{1/\beta}(\mathbb R)\). Applying Theorem~\ref{thm:state-prep}(ii) with \(\mu=\frac{1}{\beta}>\frac12\) gives
\begin{align*}
    \left\|  \psi_\infty-\Pi^F_{N,r'}\psi_\infty  \right\|_2
    \leq  C_{\beta,r'}  \exp\!\left[   -c_{\beta,r'}N^{1/(2\mu)}  \right].
\end{align*}
Since $\frac{1}{2\mu}=\frac{\beta}{2}$, this proves Eq.~\eqref{eq:kernel-GS-projection-bound}. Solving the bound for \(N\) gives the stated cutoff scaling. 
\end{proof}

\subsection{Proof of Theorem~\ref{thm:trotter}}
\label{proof:trotter}
    Let $A_i \equiv \alpha_i\hat{x}_{N_{\rm Fock}}\otimes P_i,\,\ B_j\equiv \beta_j \mathbb I_{\rm osc}\otimes Q_j$. 
    Define the set of implemented Hamiltonian summands $\mathcal G \equiv \{A_i\}_{i=1}^{N_L}\cup\{B_j\}_{j=1}^{N_H}$.
    For a word $(G_1,\ldots,G_{p+1})$ with $G_k\in\mathcal G$, define
    \begin{align*}
    C(G_1,\ldots,G_{p+1})
    \equiv
    [G_{p+1},[G_p,\cdots,[G_2,G_1]\cdots]].
    \end{align*}
    By the commutator-scaling bound for a $p^{\rm th}$-order Suzuki--Trotter formula~\cite{Childs_2021}, the simulation error satisfies
    \begin{align*}
        \left\|
        e^{-iH_{{\rm hyb},N_{\rm Fock}}t} -
        \left[S_p\!\left(\frac{t}{n_t}\right)\right]^{n_t}
        \right\| =
        \mathcal O\left(
        \frac{t^{p+1}}{n_t^p}
        \sum_{G_1,\ldots,G_{p+1}\in\mathcal G}
        \left\|
        [G_{p+1},[G_p,\cdots,[G_2,G_1]\cdots]]
        \right\|
        \right).
    \end{align*}
    Partition the nested-commutator sum into three disjoint classes:
    \begin{enumerate}
        \item mixed commutators containing at least one $A_i$ and at least one $B_j$;
        \begin{align*}
            \Gamma_p^{\rm outer}
            &\equiv
            \sum_{\substack{
            G_1,\ldots,G_{p+1}\in\mathcal G\\
            \text{at least one }G_k\in\{A_i\}_i\\
            \text{and at least one }G_k\in\{B_j\}_j
            }}
            \|C(G_1,\ldots,G_{p+1})\|
        \end{align*}
        \item pure $A_i$ commutators;
        \begin{align*}
            \Gamma_p^L
            &\equiv
            \sum_{i_1,\ldots,i_{p+1}}
            \left\|
            C(A_{i_1},\ldots,A_{i_{p+1}})
            \right\|
        \end{align*}
        \item pure $B_j$ commutators.
        \begin{align*}
            \Gamma_p^H
            &\equiv
            \sum_{j_1,\ldots,j_{p+1}}
            \left\|
            C(B_{j_1},\ldots,B_{j_{p+1}})
            \right\|.
        \end{align*}
    \end{enumerate}
    For a mixed nested commutator containing exactly $a$ terms of type $A_i$ and $p+1-a$ terms of type $B_j$, submultiplicativity of the operator norm implies
    \begin{align*}
        \|C(G_1,\ldots,G_{p+1})\|
        \leq \|\hat x_{N_{\rm Fock}}\|^a
        \|C(R_1,\ldots,R_{p+1})\|,
    \end{align*}
    where each $R_k$ is the corresponding Pauli-level operator drawn from $\{\alpha_iP_i\}_i\cup\{\beta_jQ_j\}_j$. Summing over all such mixed words gives
        $\Gamma_p^{\rm outer}
        \leq \sum_{a=1}^p
        \|\hat x_{N_{\rm Fock}}\|^a
        \Gamma_{p,a}^{(L,H)}$. For the pure $L$ sector, every summand is of type $A_i=\alpha_i\hat x_{N_{\rm Fock}}\otimes P_i$, so every $(p+1)$-fold nested commutator contributes at most $p+1$ factors of $\hat x$, yielding $\Gamma_p^L \leq \|\hat x_{N_{\rm Fock}}\|^{p+1} \Gamma_p^{(L)}$.

    For the pure $H$ sector, each summand is $B_j=\beta_j \mathbb I_{\rm osc}\otimes Q_j$, so no oscillator factor appears, and $\Gamma_p^H = \Gamma_p^{(H)}$. Combining these contributions gives
    \begin{align*}
        \left\|
        e^{-iH_{{\rm hyb},N_{\rm Fock}}t} -
        \left[S_p\!\left(\frac{t}{n_t}\right)\right]^{n_t}
        \right\| =
        \mathcal O\left[
        \frac{t^{p+1}}{{n_t}^p}
        \left( \sum_{a=1}^{p}
        \|\hat x_{N_{\rm Fock}}\|^a\Gamma_{p,a}^{(L,H)} +
        \|\hat x_{N_{\rm Fock}}\|^{p+1}\Gamma_p^{(L)} +
        \Gamma_p^{(H)}
        \right)\right].
    \end{align*}
    To ensure the simulation error is at most $\epsilon_t$, it is sufficient that
    \begin{align*}
        \frac{t^{p+1}}{{n_t}^p}
        \left( \sum_{a=1}^{p}
        \|\hat x_{N_{\rm Fock}}\|^a\Gamma_{p,a}^{(L,H)} +
        \|\hat x_{N_{\rm Fock}}\|^{p+1}\Gamma_p^{(L)} +
        \Gamma_p^{(H)}
        \right) \leq \epsilon_t.
    \end{align*}
    Solving for ${n_t}$ and using Eq.~\eqref{eq:Lambda_p_definition} gives
    \begin{align*}
        {n_t} = \mathcal O\left[ t^{1+1/p}
        \left( \frac{ \displaystyle
        \sum_{a=1}^{p}2^a
        (N_{\rm Fock}-1)^{a/2}\Gamma_{p,a}^{(L,H)} +
        2^{p+1}(N_{\rm Fock}-1)^{(p+1)/2}\Gamma_p^{(L)} +
        \Gamma_p^{(H)} }{ \epsilon_t } \right)^{1/p} \right],
    \end{align*}
    which proves the Trotter-step scaling.
    
    For gate complexity, each Pauli exponential of weight $w(P_i)$ or $w(Q_j)$ requires
        $\mathcal O\!\left(2\max(w-1,0)\right)$
    entangling gates. Since each Trotter step contains $m_p$ Suzuki layers, the CX count scales as
    \begin{align*}
        N_{\rm CX} = \mathcal O\left( 2{n_t}m_p
        \left[ \sum_{i=1}^{N_L}\max(w(P_i)-1,0) + \sum_{j=1}^{N_H}\max(w(Q_j)-1,0) \right] \right).
    \end{align*}
    
    Finally, each $L$-sector term requires one hybrid oscillator-qubit interaction per Suzuki layer, implying
    $N_{\rm hyb} = \mathcal O(n_tm_pN_L)$. Thus the theorem follows.

\subsection{Proof of Theorem~\ref{thm:postselection-success}}
\label{proof:postselection-success}

\begin{proof}
Let
\begin{align*}
    E(t) = \alpha_{N,r} \widetilde{\mathcal K}_{N,N_{\rm Fock},n_t}(t) - e^{-At}.
\end{align*}
Then
\begin{align*}
    p_{\rm succ}^{(N,N_{\rm Fock},n_t)} - p_{\rm ref}^{(N,r)}
    =\frac{\left\|e^{-At}\ket{u_0}+E(t)\ket{u_0}\right\|_2^2 - \left\|e^{-At}\ket{u_0}\right\|_2^2}{|\alpha_{N,r}|^2}.
\end{align*}
The Cauchy--Schwarz inequality, \(\|\ket{u_0}\|_2=1\), and \(\|E(t)\|_2=\epsilon_{\rm tot}\) give
\begin{align*}
    \left|p_{\rm succ}^{(N,N_{\rm Fock},n_t)}- p_{\rm ref}^{(N,r)}\right|
    \leq \frac{2\left\|e^{-At}\ket{u_0}\right\|_2\epsilon_{\rm tot}+\epsilon_{\rm tot}^2}{|\alpha_{N,r}|^2},
\end{align*}
which proves Eq.~\eqref{eq:psucc-perturbation}.
\end{proof}

\subsection{Formalism for Inhomogeneous Time-Independent Solution}
\label{ss:inh}

We give a formal operator-level construction for the time-independent inhomogeneous setting. Finite-resource encoding, oscillator-state synthesis, success-probability analysis, and numerical validation remain future work. Let $A(t):=A=L+iH$, take a constant source $b(t):=b\neq0$ with normalized $\ket{b}$, and define the elapsed time $\tau:=t-s$. Equation~\eqref{eq:LCHS-sol-inhomo} then gives
\begin{align}
    u_{\rm inh}(t)
    := \int_0^t e^{-A\tau} b \, \dd \tau
    = \int_\Rbb \int_\Rbb g(k)\, \mathbf{1}_{[0,t]}(\tau)\, e^{-i \tau (kL+H)} b \, \dd k \, \dd \tau,
    \label{eq:LCHS-inhomo-cv}
\end{align}
where $\mathbf{1}_{[0,t]}$ is the indicator function of the interval $[0,t]$. To embed Eq.~\eqref{eq:LCHS-inhomo-cv} into a hybrid oscillator-qubit system, introduce a second oscillator encoding the elapsed time, with generalized position eigenbases $\{\ket{k}\}_{k\in\Rbb}$ and $\{\ket{\tau}\}_{\tau\in\Rbb}$ and position operators $\hat k$ and $\hat\tau$, and let the qubit register carry the source state $\ket{b}$. Choose normalized oscillator states $\ket{\psi_K},\ket{\phi_K}$ and $\ket{\psi_\tau^{(t)}},\ket{\phi_\tau^{(t)}}$, together with known scalars $\alpha_K$ and $\alpha_\tau(t)$, such that
\begin{align}
    \alpha_K\,\phi_K^*(k)\psi_K(k) = g(k), \qquad
    \alpha_\tau(t)\,\phi_\tau^{(t)*}(\tau)\psi_\tau^{(t)}(\tau) = \mathbf{1}_{[0,t]}(\tau). \label{eq:inhomo-kernel}
\end{align}
The scales $\alpha_K$ and $\alpha_\tau(t)$ play the role of $\alpha_g$ in Theorem~\ref{thm:CV--DV-lchs}. For normalized states, the Cauchy--Schwarz inequality gives $|\int\phi^*\psi|\leq1$. Together with $\int_\Rbb\mathbf{1}_{[0,t]}(\tau)\,\dd\tau=t$ and $\int_\Rbb g(k)\,\dd k=1$, this implies $|\alpha_\tau(t)|\geq t$ and $|\alpha_K|\geq1$. The normalization constraints generally preclude setting both scales to one, and both must be tracked classically.
Define the joint preparation state and postselection map
\begin{align*}
    \ket{\Psi_{\rm inh}(t)} := \ket{\psi_K}\otimes \ket{\psi_\tau^{(t)}} \otimes \ket{b}, \qquad
    \bra{\Phi_{\rm inh}(t)} := \bra{\phi_K}\otimes \bra{\phi_\tau^{(t)}} \otimes \mathbb I_q,
\end{align*}
and the hybrid unitary
\begin{align*}
    W_{\rm inh} := \exp\!\left[-i\left(\hat{k}\otimes \hat{\tau}\otimes L + I_K \otimes \hat{\tau}\otimes H\right)\right],
\end{align*}
which becomes a well-defined finite-dimensional unitary on truncated oscillator spaces. Its generator is diagonal in the two oscillator position variables and acts as $\tau(kL+H)$ on the qubit register, so the functional-calculus argument of Appendix~\ref{proof:CV--DV-lchs} yields
\begin{align*}
    \alpha_K\,\alpha_\tau(t)\,\bra{\Phi_{\rm inh}(t)} W_{\rm inh}\ket{\Psi_{\rm inh}(t)}
    = \int_0^t \int_{\Rbb} g(k)\,e^{-i\tau(kL+H)}\ket{b} \, \dd k \, \dd \tau
    = \ket{u_{\rm inh}(t)}.
\end{align*}
A postselected implementation would recover $\ket{u_{\rm inh}(t)}$ up to the known scale $[\alpha_K\alpha_\tau(t)]^{-1}$, whose magnitude enters the success probability in the same way as the inverse-overlap scale of Section~\ref{ss:postselection}.

Since $\mathbf{1}_{[0,t]}$ is not analytic, a smooth alternative is its Gaussian mollification. For $\varrho>0$ and $h_\varrho(x):=\frac{1}{\sqrt{\pi}\,\varrho}e^{-x^2/\varrho^2}$, the convolution $w_{\varrho,t}:=\mathbf{1}_{[0,t]}*h_\varrho$ evaluates to
\begin{align*}
    w_{\varrho,t}(\tau) = \frac12\left[\mathrm{erf}\!\left(\frac{\tau}{\varrho}\right)-\mathrm{erf}\!\left(\frac{\tau-t}{\varrho}\right)\right], \qquad
    \|w_{\varrho,t}-\mathbf{1}_{[0,t]}\|_{L^1(\mathbb R)} \leq \frac{2\varrho}{\sqrt{\pi}},
\end{align*}
where $\mathrm{erf}(x)=\frac{2}{\sqrt{\pi}}\int_0^x e^{-y^2}\,\dd y$ and the $L^1$ bound follows from translating $\mathbf{1}_{[0,t]}$ under the unit-mass mollifier $h_\varrho$. Replacing the indicator by $w_{\varrho,t}$ in the second factorization of Eq.~\eqref{eq:inhomo-kernel} and using the unitarity of each $e^{-i\tau(kL+H)}$ gives
\begin{align*}
    \|u_{\rm inh,\varrho}(t)-u_{\rm inh}(t)\|
    \leq \|b\|\,\|g\|_{L^1(\mathbb R)}\,\|w_{\varrho,t}-\mathbf{1}_{[0,t]}\|_{L^1(\mathbb R)}
    \leq \frac{2\varrho}{\sqrt{\pi}}\, \|b\|\,\|g\|_{L^1(\mathbb R)}.
\end{align*}

Two oscillators realize the formal continuous linear combination of the unitaries $e^{-i\tau(kL+H)}$ over the LCHS variable $k$ and elapsed time $\tau$. The first oscillator encodes $k$. Relative to the homogeneous construction, the second oscillator encodes the time-window weight $\mathbf{1}_{[0,t]}$ or its mollification $w_{\varrho,t}$. For a time-dependent source $b(s)$, a source-time entangled state replaces the fixed qubit input $\ket{b}$ and introduces an additional source-state preparation oracle in the qubit space.

\subsection{Gate Counts for \texorpdfstring{$n_t$}{nt} First-order Trotterization Blocks in \texorpdfstring{$d$}{d}-dimensional Heat Equation}
\label{ss:heat-count}
Consider the Pauli decomposition of $T^{(bc)}_m$ for the 1-dimensional heat equation
\begin{align*}
    T^{(bc)}_m = \sum_{\ell=1}^{N_L^{(bc)}(m)} c_\ell^{(bc)} P_\ell^{(bc)}.
\end{align*}
For one first-order factor $e^{-i\theta\hat x\otimes P_\ell^{(bc)}}$, an identity string contributes one hybrid gate, a non-identity string contributes one hybrid gate together with a parity ladder of $2(w(P_\ell^{(bc)})-1)$ CXs, and the basis changes contribute $2$ one-qubit gates for each $X$ and $4$ for each $Y$. Thus, we have
\begin{align*}
    N_{\rm hyb,1}^{(bc)}(m)
    &= N_L^{(bc)}(m), \\
    N_{\rm CX,1}^{(bc)}(m)
    &= 2\sum_{\ell=1}^{N_L^{(bc)}(m)} \max\!\bigl(w(P_\ell^{(bc)})-1,0\bigr), \\
    N_{1q,1}^{(bc)}(m)
    &= \sum_{\ell=1}^{N_L^{(bc)}(m)}\left[2\,n_X(P_\ell^{(bc)})+4\,n_Y(P_\ell^{(bc)})\right],
\end{align*}
where $n_X(P)$ and $n_Y(P)$ denote the numbers of $X$ and $Y$ in Pauli string $P$. Using the explicit decompositions in Section~\ref{1d-heat}, the Dirichlet decomposition contains one identity term and, for each $c$, exactly $2^c$ Pauli strings of weight $c+1$, while for the one-qubit gates the $c=0$ term contributes $2$ basis-change gates and each block $R_c^{(m)}$ with $c\geq 1$ contributes $3(c+1)2^c$. Evaluating the finite sums gives
\begin{align*}
    N_{\rm hyb,1}^{(D)}(m) &= 2^m = M, & N_{\rm CX,1}^{(D)}(m) &= 2\left[(m-2)2^m + 2\right], & N_{1q,1}^{(D)}(m) &= 3(m-1)2^m+2.
\end{align*}
For the Neumann boundary condition, the correction of Section~\ref{1d-heat-bc} adds $2^{m-1}-1$ non-identity $Z$-type strings and no basis-change gates, so
\begin{align*}
    N_{\rm hyb,1}^{(\rm N)}(m) &= 3\cdot 2^{m-1}-1, & N_{\rm CX,1}^{(\rm N)}(m) &= 2^{m-1}(5m-10)+6, & N_{1q,1}^{(\rm N)}(m) &= 3(m-1)2^m+2.
\end{align*}
For the periodic boundary condition, the simplified decomposition keeps the Dirichlet blocks $R_0^{(m)},\dots,R_{m-2}^{(m)}$ and replaces the last block by $2^{m-2}$ weight-$m$ strings, so
\begin{align*}
    N_{\rm hyb,1}^{(P)}(m) &= 3\cdot 2^{m-2}, & N_{\rm CX,1}^{(P)}(m) &= 2^{m-1}(3m-7)+4,\\
    N_{1q,1}^{(P)}(m) &= \begin{cases} 6, & m=2,\\ 2+(9m-13)2^{m-2}, & m\geq 3. \end{cases}
\end{align*}

For $d$ dimensions and $n_t$ product-formula steps, the coordinate factors commute, so the total counts are sums of the one-dimensional costs,
\begin{align*}
    N_{\rm hyb}^{(d)} = n_t\sum_{i=1}^d N_{\rm hyb,1}^{(bc_i)}(m_i), \qquad
    N_{\rm CX}^{(d)} = n_t\sum_{i=1}^d N_{\rm CX,1}^{(bc_i)}(m_i), \qquad
    N_{1q}^{(d)} = n_t\sum_{i=1}^d N_{1q,1}^{(bc_i)}(m_i),
\end{align*}
as summarized in Table~\ref{tab:heat-gate-scaling}. 

Exact commutation of the coordinate factors on disjoint qubit subregisters makes the $d$-dimensional split error-free. The approximation comes from the one-dimensional product formula within each factor. A telescoping argument over the commuting coordinate factors gives the additive budget \(\epsilon_{\rm heat}^{(d)} \leq \sum_{i=1}^d \epsilon_{\rm heat}^{(i)}\), where each \(\epsilon_{\rm heat}^{(i)}\) is estimated by Eq.~\eqref{eq:heat-trotter-error} with the grid spacing \(h_i\) and register size \(m_i\) of axis \(i\).

\subsection{The Huang--An kernel at a grid-selected operating point}\label{app:ha-retune}

We use a two-stage adaptive scan to quantify how much of the Huang--An kernel's sampling penalty arises from the point selected for \(g_\beta\) in the shared-point comparison of Table~\ref{tab:ha-comparison}. The coarse stage uses spacing \(0.1\) over \(r\in[1.0,2.0]\) and \(r'\in[0.1,1.0]\) to locate the region of interest. The fine stage evaluates both kernels on the common grid \(r\in[0.60,1.90]\), \(r'\in[0,0.35]\) with spacing \(0.05\) under the constraint \(r>r'\), at fixed kernel shapes, \(N=32\), \(N_{\rm Fock}=64\), and \(T=1\). Among completed grid points whose worst-boundary exact-map error is at most \(0.7\%\), the selection rule used for \(\theta^\star\) chooses the largest worst-boundary benchmark success probability. Ties favor the smaller error. 
Adaptive coefficient quadrature failed to reach its tolerance at \(28\) of the \(216\) fine-grid points for \(g_\beta\) and at \(14\) points for Huang--An. These points are excluded. The selected points are the best observed among completed evaluations. Global optimality remains undetermined.
For \(g_\beta\) the fine grid selects \((1.55,\,0.25)\), adjacent to \(\theta^\star\) and within a few percent of it in both selection quantities, so the main-text operating point is kept unchanged.

For the Huang--An kernel the rule selects \((r,r')=(0.95,\,0)\), interior in \(r\) and at the lower edge \(r'=0\) of the scanned grid. Without a preparation-basis squeeze the kernel is carried entirely by the Fock coefficients, the prepared state has exact support on \(n\leq N-1=31\), and the ordinary-Fock tail \(\tau_{\psi,64}\) vanishes identically, so the embedding bound of Eq.~\eqref{eq:fock-embedding-bound} is exhausted by the postselection state alone. Table~\ref{tab:ha-retune-point} compares the shared and the selected operating points.

\begin{table}[ht!]
\centering
\caption{The Huang--An kernel at the shared operating point of Table~\ref{tab:ha-comparison} and at its grid-selected point \((r,r')=(0.95,\,0)\), same conventions. The tail \(\tau_{\psi,64}=0\) at the selected point is exact, following from the finite Fock support at \(r'=0\).}
\label{tab:ha-retune-point}
{\small
\setlength{\tabcolsep}{5pt}
\begin{tabular}{lcc}
\toprule
 & Shared \((1.60,\,0.25)\) & Selected \((0.95,\,0)\) \\
\midrule
Coefficient-projection tail \(\epsilon_{\rm tr}(32)\) & \(5.11\times10^{-3}\) & \(2.44\times10^{-3}\) \\
Ordinary-Fock tail \(\tau_{\psi,64}\) & \(2.43\times10^{-6}\) & \(0\) \\
Non-Gaussianity \(\delta_{\rm nG}\) & \(0.63\) & \(1.42\) \\
Mean photon number \(\langle\hat n\rangle\) & \(0.49\) & \(1.85\) \\
Overlap scale \(|\alpha_{N,r}|\) & \(2.476\) & \(1.953\) \\
\(\varepsilon_F\) (D/N/P) & \(0.280\%\,/\,0.247\%\,/\,0.094\%\) & \(0.357\%\,/\,0.241\%\,/\,0.182\%\) \\
Benchmark \(p_{\rm succ}\) (D/N/P) & \(2.90\%\,/\,4.54\%\,/\,4.23\%\) & \(4.66\%\,/\,7.29\%\,/\,6.80\%\) \\
\bottomrule
\end{tabular}}
\end{table}

The common local \((r,r')\) scan raises every Huang--An benchmark probability by a factor of about \(1.6\) at the selected point, consistent with the overlap scale falling from \(2.476\) to \(1.953\). The largest of the three fixed-scale errors is \(0.357\%\). The non-Gaussianity and photon number rise to values comparable with those of \(g_\beta\), so the smaller \(\delta_{\rm nG}\) at the shared point is specific to that operating point. At the grid-selected Huang--An point, each benchmark probability remains below half of the corresponding shared-point \(g_\beta\) value in Table~\ref{tab:ha-comparison}, so the accuracy--sampling trade-off identified in Section~\ref{ss:kernel} persists after Huang--An-specific retuning on this grid.

The shared-point and selected-point quantities use a coefficient-integration window that increases with Hermite order and are stable under a further enlargement of the window and tighter quadrature tolerances.

Table~\ref{tab:ha-retune-breadth} reruns the seven breadth instances of Table~\ref{tab:clean_breadth} at the selected point, with generators, inputs, and conventions unchanged. The fixed-scale errors remain below the \(g_\beta\) values of Table~\ref{tab:clean_breadth} in every instance, by factors between \(2.9\) and \(4.5\) on the one-dimensional cases, while the benchmark probabilities are lower by approximately the ratio of the squared overlap scales. The two-dimensional stress case reaches \(1.33\%\) fixed-scale error, against \(7.40\%\) at \(\theta^\star\), at a benchmark probability of \(0.42\%\). These calculations use the Huang--An shape \((m,\delta)=(1,2)\) fixed from the source's total-complexity study and the \((r,r')\) point selected on the one-dimensional \(M=4\) benchmarks. The interpretation of the \(4\times4\) instance in Section~\ref{sec:clean_experiments} as a stress test of a shared parameter set is unchanged.

\begin{table}[ht!]
\centering
\caption{Breadth instances of Table~\ref{tab:clean_breadth} evaluated with the Huang--An kernel at its selected point \((r,r')=(0.95,\,0)\), full circuit-level maps with Law--Eberly preparation and \(n_t=100\). Columns follow Table~\ref{tab:clean_breadth}, and \(p_{\rm succ}\) gives the benchmark-input value with the range over computational-basis inputs in parentheses.}
\label{tab:ha-retune-breadth}
{\small
\setlength{\tabcolsep}{5pt}
\begin{tabular}{lccccc}
\toprule
Instance & \(D\) & \(\varepsilon_F\) & worst \(1-F_{\rm cond}\) & \(p_{\rm succ}\) [\%] & \(\epsilon_t\) \\
\midrule
1D heat, \(M=8\), Dirichlet & 8 & \(0.32\%\) & \(1.02\times10^{-5}\) & \(4.74\) (\(2.35\)--\(5.43\)) & \(8.34\times10^{-4}\) \\
1D heat, \(M=16\), Dirichlet & 16 & \(0.33\%\) & \(1.27\times10^{-5}\) & \(4.74\) (\(2.35\)--\(5.44\)) & \(9.56\times10^{-4}\) \\
1D heat, \(M=32\), Dirichlet & 32 & \(0.33\%\) & \(1.25\times10^{-5}\) & \(4.74\) (\(2.35\)--\(5.44\)) & \(9.83\times10^{-4}\) \\
AD, \(M=8\), Dirichlet & 8 & \(0.40\%\) & \(1.79\times10^{-5}\) & \(5.28\) (\(1.20\)--\(5.44\)) & \(1.16\times10^{-3}\) \\
AD, \(M=16\), Dirichlet & 16 & \(0.39\%\) & \(3.04\times10^{-5}\) & \(5.28\) (\(1.20\)--\(5.44\)) & \(1.21\times10^{-3}\) \\
AD, \(M=32\), Dirichlet & 32 & \(0.39\%\) & \(3.22\times10^{-5}\) & \(5.28\) (\(1.20\)--\(5.44\)) & \(1.21\times10^{-3}\) \\
2D heat, \(4\times4\), Dirichlet & 16 & \(1.33\%\) & \(3.38\times10^{-4}\) & \(0.42\) (\(0.21\)--\(0.83\)) & \(6.93\times10^{-4}\) \\
\bottomrule
\end{tabular}}
\end{table}

\subsection{A common \texorpdfstring{$(r,r')$}{(r,r')} scan on the two-dimensional generator}\label{app:ha-retune-2d}

We repeat the sweep of Appendix~\ref{app:ha-retune} using the \(4\times4\) two-dimensional generator as the selection instance. The first stage uses spacing \(0.1\) over \(r\in[0.6,2.2]\) and \(r'\in[0,1.0]\). The second stage scans the box \(r\in[0.9,1.6]\), \(r'\in[0,0.15]\) at spacing \(0.05\), all at the exact-truncated-map level. Under the error cut of the one-dimensional rule, one Huang--An point is admissible. The lowest observed \(g_\beta\) error on the box is \(1.537\%\), leaving its admissible set empty. The rule selects the sole admissible Huang--An point \((1.15,\,0)\) and no \(g_\beta\) point. Table~\ref{tab:ha-retune-2d-front} reports the error--probability Pareto fronts.

\begin{table}[ht!]
\centering
\caption{Pareto fronts of the fixed-scale error against the benchmark-input success probability for the \(4\times4\) two-dimensional heat instance over the fine box \(r\in[0.9,1.6]\), \(r'\in[0,0.15]\), exact-truncated-map level, \(N=32\), \(N_{\rm Fock}=64\), \(T=1\). The rule selects the sole admissible Huang--An point but no \(g_\beta\) point, so no pair of operating points is adopted.}
\label{tab:ha-retune-2d-front}
{\small
\setlength{\tabcolsep}{5pt}
\begin{tabular}{lccc}
\toprule
Kernel & \((r,\,r')\) & \(\varepsilon_F^{\rm 2D}\) & Benchmark \(p_{\rm succ}\) \\
\midrule
\(g_\beta\) & \((1.30,\,0.05)\) & \(1.537\%\) & \(1.123\%\) \\
\(g_\beta\) & \((1.40,\,0)\) & \(3.103\%\) & \(1.128\%\) \\
\(g_\beta\) & \((1.35,\,0)\) & \(3.135\%\) & \(1.158\%\) \\
\(g_\beta\) & \((1.30,\,0)\) & \(3.324\%\) & \(1.179\%\) \\
\(g_\beta\) & \((1.25,\,0)\) & \(3.649\%\) & \(1.184\%\) \\
\midrule
Huang--An & \((1.15,\,0)\) & \(0.669\%\) & \(0.381\%\) \\
Huang--An & \((1.10,\,0)\) & \(0.703\%\) & \(0.394\%\) \\
Huang--An & \((1.05,\,0)\) & \(0.801\%\) & \(0.406\%\) \\
Huang--An & \((1.00,\,0)\) & \(0.975\%\) & \(0.415\%\) \\
Huang--An & \((0.95,\,0)\) & \(1.263\%\) & \(0.418\%\) \\
\bottomrule
\end{tabular}}
\end{table}

Along each front, the benchmark probability varies by less than ten percent. The error spans factors of \(2.37\) for \(g_\beta\) and \(1.89\) for Huang--An. Across the two probability ranges, the \(g_\beta\) values exceed the Huang--An values by factors from \(2.7\) to \(3.1\). Every Huang--An point and four of the five \(g_\beta\) points lie at \(r'=0\). The minimum-error \(g_\beta\) endpoint uses \(r'=0.05\). Relative to the one-dimensional selections, the \(g_\beta\) front shifts to lower \(r\). The Huang--An front spans \(r=0.95\)--\(1.15\), beginning at its one-dimensional selected value. The two-dimensional generator has a wider spectral range than the one-dimensional selection instances, and the observed retuning shifts are kernel dependent.

\end{document}